\documentclass[11pt]{article}

\usepackage{amsmath}
\usepackage{amsthm}
\usepackage{times}
\usepackage{bm}
\usepackage{natbib} 
\usepackage[usenames, dvipsnames]{color}
\usepackage{todonotes}
\usepackage[margin=1.3in]{geometry} 
\usepackage[plain,noend]{algorithm2e}
\usepackage{graphicx}
\usepackage{subfigure}  
\usepackage{amsfonts}
\usepackage{amssymb}
\usepackage{mathtools}
\usepackage{bm}
\usepackage{epstopdf}
\usepackage{caption}
\usepackage{algorithmic}
\usepackage[affil-it]{authblk}

\usepackage{appendix}

\makeatletter

\providecommand{\@LN}[2]{}%

\makeatother

\newtheorem{theorem}{Theorem}
\newtheorem{lemma}{Lemma}

\newtheorem{proposition}{Proposition}

\newtheorem{assumption}{Assumption}

\DeclareMathOperator*{\iid}{\overset{ \scriptscriptstyle\text{IID}}{ \scriptscriptstyle\sim}}

\DeclareMathOperator*{\argmin}{argmin}
\DeclareMathOperator*{\argmax}{arg\,max}
\DeclareMathOperator*{\MMD}{\text{MMD}}

\begin{document}

\title{Statistical Inference for Generative Models \\ with Maximum Mean Discrepancy}

\author{Fran\c{c}ois-Xavier Briol$^{1,3}$, Alessandro Barp$^{2,3}$, Andrew B. Duncan$^{2,3}$, Mark Girolami$^{1,3}$ \\ 
$^{1}$University of Cambridge, Department of Engineering \\
$^{2}$Imperial College London, Department of Mathematics \\
$^{3}$The Alan Turing Institute}

\maketitle

\begin{abstract}
While likelihood-based inference and its variants provide a statistically efficient and widely applicable approach to parametric inference, their application to models involving intractable likelihoods poses challenges. In this work, we study a class of minimum distance estimators for intractable generative models, that is, statistical models for which the likelihood is intractable, but simulation is cheap.  The distance considered, maximum mean discrepancy (MMD), is defined through the embedding of probability measures into a reproducing kernel Hilbert space. We study the theoretical properties of these estimators, showing that they are consistent, asymptotically normal and robust to model misspecification. A main advantage of these estimators is the flexibility offered by the choice of kernel, which can be used to trade-off statistical efficiency and robustness. On the algorithmic side, we study the geometry induced by MMD on the parameter space and use this to introduce a novel natural gradient descent-like algorithm for efficient implementation of these estimators. We illustrate the relevance of our theoretical results on several classes of models including a discrete-time latent Markov process and two multivariate stochastic differential equation models.
\end{abstract}


\section{Introduction}
\label{sec:introduction}

Consider an open subset $\mathcal{X} \subset \mathbb{R}^d$ and denote by $\mathcal{P}(\mathcal{X})$ the set of Borel probability measures on this domain. We consider the problem of learning a probability measure $\mathbb{Q} \in \mathcal{P}(\mathcal{X})$ from identically and independently distributed (IID) realisations $\{y_j\}_{j=1}^m \iid \mathbb{Q}$.  We will focus on parametric inference with a parametrised family $\mathcal{P}_{\Theta}(\mathcal{X}) = \lbrace \mathbb{P}_{\theta}  \in \mathcal{P}(\mathcal{X}) \, : \, \theta \in \Theta \rbrace$, for an open set $\Theta\subset \mathbb{R}^p$ i.e. we seek $\theta^* \in \Theta$ such that $\mathbb{P}_{\theta^*}$ is closest to $\mathbb{Q}$ in an appropriate sense.  If $ \mathbb{Q} \in \mathcal{P}_{\Theta}(\mathcal{X})$ we are in the \emph{M-closed} setting,  otherwise we are in the \emph{M-open} setting. When $\mathbb{P}_{\theta}$ has a density $p(\cdot  |  \theta)$ with respect to the Lebesgue measure, then a standard approach is to use the maximum likelihood estimator (MLE):
\begin{align*}
\hat{\theta}_m^{\text{MLE}} & = \argmax_{\theta \in \Theta}  \frac{1}{m} \sum_{j=1}^m \log p(y_j|\theta).
\end{align*}
For complex models, a density may not be easily computable, or even exist and so the MLE need not be available.  In some cases it is possible to approximate the likelihood; see for example pseudo likelihood \citep{Besag1974}, profile likelihood \citep{Murphy2000} and composite likelihood \citep{Varin2011} estimation. It is sometimes also possible to access likelihoods in un-normalised forms i.e. $p(y|\theta) = \bar{p}(y|\theta)/C(\theta)$ where the constant $C(\theta)$ is unknown. This class of models is known as un-normalised models, or doubly-intractable models in the Bayesian literature, and a range of exact and approximate methods have been developped for this case; see for example the Markov chain Monte Carlo (MCMC) algorithms of \cite{Murray2006,Moller2006} or the score-based and ratio-based approaches of \cite{Hyvarinen2006,Hyvarinen2007,Gutmann2012}.

However, for many models of interest in modern statistical inference, none of the methods above can be applied straightforwardly and efficiently due to the complexity of the likelihoods involved. This is most notably the case for intractable generative models, sometimes also called implicit models or likelihood-free models; see \cite{Mohamed2016} for a recent overview. Intractable generative models are parametric families of probability measures for which it is possible to obtain realisations for any value of the parameter $\theta \in \Theta$, but for which we do not necessarily have access to a likelihood or approximation thereof. These models are used throughout the sciences, including in the fields of ecology \citep{Wood2010}, population genetics \citep{Beaumont2002} or astronomy \citep{Cameron2012}. They also appear in machine learning as black-box models; see for example generative adversarial networks (GANs) \citep{Goodfellow2014} and variational auto-encoders \citep{Kingma2014}.

Given a Borel probability space $(\mathcal{U}, \mathcal{F}, \mathbb{U})$, we will call generative model any probability measure which is the pushforward $G_{\theta}^{\#} \mathbb{U}$ of the probability measure $\mathbb{U}$ with respect to a measurable parametric map $G_\theta: \mathcal{U} \rightarrow \mathcal{X}$ called the \emph{generator}. To generate $n$ independent realisations from the model, we produce IID realisations $\{u_i\}_{i=1}^n \iid \mathbb{U}$ and apply the generator to each of these samples: $x_i = G_\theta(u_i)$ for $i=1,\ldots,n$. While it is straightforward to generate samples from these models, a likelihood function need not be available, given that an associated positive density may not be computable or even exist. We therefore require alternatives to the MLE.

The estimators studied in this paper fall within the class of minimum divergence/distance estimators \citep{Pardo2005,Basu2011}. These are estimators minimising some notion of divergence $D:\mathcal{P}(\mathcal{X}) \times \mathcal{P}(\mathcal{X}) \rightarrow \mathbb{R}_+$ (or an approximation thereof) between an empirical measure $\mathbb{Q}^m = \frac{1}{m}\sum_{j=1}^m \delta_{y_j}$ (where $\delta_{y_j}$ denotes a Dirac measure at $y_j$), obtained from the data $\{y_j\}^{m}_{j=1} \iid \mathbb{Q}$, and the parametric model:
\begin{align}\label{eq:minimum_distance_estimator}
\hat{\theta}^D_m & =  \argmin_{\theta \in \Theta} D(\mathbb{P}_{\theta}||\mathbb{Q}^m)
\end{align}
If $\mathbb{Q}^m$ was absolutely continuous with respect to $\mathbb{P}_{\theta}$, maximising the likelihood would correspond to minimising the Kullback-Leibler (KL) divergence which, given $\mathbb{P}_1,\mathbb{P}_2 \in \mathcal{P}(\mathcal{X})$, is defined as $D_{\text{KL}}(\mathbb{P}_1||\mathbb{P}_2) :=  \int_{\mathcal{X}}  \log (\mathrm{d}\mathbb{P}_1/\mathrm{d}\mathbb{P}_2) \mathrm{d}\mathbb{P}_1$, where $\mathrm{d}\mathbb{P}_1/\mathrm{d}\mathbb{P}_2$ is the Radon-Nikodym derivative of $\mathbb{P}_1$ with respect to $\mathbb{P}_2$. This approach to inference is useful for models with complicated or intractable likelihood, since the choice of divergence can be adapted to the class of models of interest. 

In previous works, minimum distance estimators for generative models have been considered based on the Wasserstein distance and its Sinkhorn relaxation; see \cite{Bassetti2006,Frogner2015,Montavon2016,Genevay2017,Frogner2018,Sanjabi2018}. These have the advantage that they can leverage extensive work in the field of optimal transport. In a Bayesian context, similar ideas are used in approximate Bayesian computation (ABC) methods \cite{Marin2012,Lintusaari2017} where synthetic data sets are simulated from the model then compared to the true data using some notion of distance. There, significant work has been put into automating the choice of distance \citep{Fearnhead2011}, and the use of the Wasserstein distance has also recently been studied \citep{Bernton2019}.\\

In this paper, we shall investigate the properties of minimal divergence estimators based on an approximation of \emph{maximum mean discrepancy} (MMD). Such estimators have already been used extensively in the machine learning literature with generators taking the form of neural networks  \citep{Dziugaite2015,Li2015,Li2017,Sutherland2017,Arbel2018,Binkowski2018,Romano2018,DosSantos2019} where they are usually called MMD GANs, but can be used more generally. Our main objective in this paper is to present a general framework for minimum MMD estimators, to study their theoretical properties and to provide an initial discussion of the impact of the choice of kernel. This study brings insights into the favourable empirical results of previous work in MMD for neural networks, and demonstrate more broadly the usefulness of this approach for inference within the large class of intractable generative models of interest in the statistics literature. As will be discussed, this approach is significantly preferable to alternatives based on the Wasserstein distance for models with expensive generators as it comes with significantly stronger generalisation bounds and is more robust in several scenarios. Our detailed contributions can be summarised as follows:
\begin{enumerate}
\item In Section \ref{sec:minimum_distance}, we introduce the MMD metric, minimum MMD estimators, and the statistical Riemannian geometry the metric induces on the parameter space $\Theta$. Through this, we rephrase the mimimum divergence estimator problem in terms of a gradient flow, thus obtaining a stochastic natural gradient descent method for finding the estimator $\theta^*$ which can significantly reduce computation cost as compared to stochastic gradient descent.  

\item In Section \ref{sec:MMDestimators_theory}, we focus on the theoretical properties of minimum MMD estimators and associated approximations. We use the information geometry of MMD to demonstrate generalisation bounds and statistical consistency, then prove that the estimator is asymptotically normal in the M-closed setting. These results give us necessary assumptions on the generator for the use of the estimators. We then analyse the robustness properties of the estimator in the M-open setting, establishing conditions for qualitative and quantitative robustness.

\item In Section \ref{sec:gaussian_model} we study the efficiency and robustness of minimum MMD estimators based on Gaussian kernels for classes of isotropic Gaussian location and scale models. We demonstrate the effect of the kernel lengthscale on the efficiency of the estimators, and demonstrate a tradeoff between (asymptotic) efficiency and robustness. For high-dimensional problems, we demonstate that choosing the lengthscale according to the median heuristic provides an asymptotic variance independent of dimensionality. We also extend our analysis to mixtures of kernels, providing insights on settings often considered in machine learning applications.

\item In Section \ref{sec:MMDestimators_experiments}, we perform numerical simulations to support the theory detailed in the previous sections, demonstrating the behaviour of minimum MMD estimators for a number of examples including estimation of unknown parameters for the g-and-k distribution, in a stochastic volatility model and for two systems of stochastic differential equations.
\end{enumerate}


\section{The Maximum Mean Discrepancy Statistical Manifold}\label{sec:minimum_distance}

We begin by formalising the notion of MMD and introduce the corresponding minimum MMD estimators. We then use tools from information geometry to analyse these estimators, which leads to a stochastic natural gradient descent algorithm for efficient implementation.

\subsection{Maximum Mean Discrepancy}

Let $k:\mathcal{X}\times\mathcal{X}\rightarrow \mathbb{R}$ be a Borel measurable kernel on $\mathcal{X}$, and consider the  reproducing kernel Hilbert space $\mathcal{H}_k$ associated with $k$  (see \cite{Berlinet2004}), equipped with inner product $\langle \cdot , \cdot \rangle_{\mathcal{H}_k}$ and norm $\lVert \cdot \rVert_{\mathcal{H}_k}$. Let $\mathcal{P}_k(\mathcal{X})$ be the set of  Borel probability measures $\mu$ such that $\int_{\mathcal{X}} \sqrt{k(x,x)}\mu(\mathrm{d}x) < \infty$. 
The \emph{kernel mean embedding} $\Pi_k(\mu) = \int_{\mathcal{X}} k(\cdot, y)\mu(\mathrm{d}y)$, intepreted as a Bochner integral, defines a continuous embedding from $\mathcal{P}_k(\mathcal{X})$ into $\mathcal{H}_k$. The mean embedding pulls-back the metric on $\mathcal H_k$ generated by the inner product to define a pseudo-metric on  $\mathcal{P}_k(\mathcal{X})$ called the maximum mean discrepancy $\MMD:\mathcal{P}_k(\mathcal{X})\times \mathcal{P}_k(\mathcal{X}) \rightarrow \mathbb{R}_+$, i.e., $\MMD(\mathbb{P}_1 || \mathbb{P}_2) = \lVert \Pi_k(\mathbb{P}_1) - \Pi_k(\mathbb{P}_2) \rVert_{\mathcal{H}_k}$. The squared-MMD has a particularly simple expression that can be derived through an application of the reproducing property ($f(x)=\langle f, k(\cdot,x)\rangle_{\mathcal{H}_k}$): 
\begin{align*}
{\MMD}^2(\mathbb{P}_1||\mathbb{P}_2)
& := 
\Big\|  \int_{\mathcal{X}} k(\cdot,x) \mathbb{P}_1(\mathrm{d}x) -\int_{\mathcal{X}} k(\cdot,x) \mathbb{P}_2(\mathrm{d}x) \Big\|_{\mathcal{H}_k}^2 \\
& =  \int_{\mathcal{X}} \int_{\mathcal{X}} k(x,y) \mathbb{P}_1(\mathrm{d}x) \mathbb{P}_1(\mathrm{d}y) - 2 \int_{\mathcal{X}} \int_{\mathcal{X}} k(x,y) \mathbb{P}_1(\mathrm{d}x) \mathbb{P}_2(\mathrm{d}y)\\
&  \qquad  + \int_{\mathcal{X}} \int_{\mathcal{X}} k(x,y) \mathbb{P}_2(\mathrm{d}x)\mathbb{P}_2(\mathrm{d}y),
\end{align*}
 thus providing a closed form expression up to calculation of expectations.  The MMD is in fact a \emph{integral probability pseudo-metric} \citep{Muller1997,Sriperumbudur2012,Sriperumbudur2016} since it can be expressed as:
\begin{align*}
{\MMD}(\mathbb{P}_1||\mathbb{P}_2)  & = \sup_{ \lVert f \rVert_{\mathcal{H}_k}\leq 1} \left|\int_{\mathcal{X}} f(x)\mathbb{P}_1(\mathrm{d}x) - \int_{\mathcal{X}} f(x)\mathbb{P}_2(\mathrm{d}x)\right|.
\end{align*} 
Integral probability metrics are prominent in the information-based complexity literature where they correspond to the worst-case integration error \citep{Dick2013,Briol2016}. If $\Pi_k$ is injective then the kernel $k$ is said to be characteristic \citep{Sriperumbudur2009}.  In this case MMD becomes a metric on $\mathcal{P}_k$ (and hence a statistical divergence).  A sufficient condition for $k$ to be characteristic is that $k$ is \emph{integrally strictly positive definite}, i.e. $\int_{\mathcal{X}}\int_{\mathcal{X}}k(x,y)\mathbb{P}(\mathrm{d}x)\mathbb{P}(\mathrm{d}y) = 0 $ implies that $\mathbb{P} = 0$ for all $\mathbb{P} \in \mathcal{P}_k$.  On $\mathcal{X}=\mathbb{R}^d$, \cite{Sriperumbudur2009} showed that the Gaussian and inverse multiquadric kernels are both integrally strictly positive definite.  We shall assume this condition holds throughout the paper, unless explicitly stated otherwise.

\subsection{Minimum MMD estimators}

This paper proposes to use MMD in a minimum divergence estimator framework for inference in intractable generative models. Given an unknown data generating distribution $\mathbb{Q}$ and a parametrised family of model distributions $\mathcal{P}_{\Theta}(\mathcal{X})$, we consider a minimum MMD estimator: 
\begin{align}\label{eq:MMD_estimator_m}
\hat{\theta}_m & =  \arg\min_{\theta \in \Theta}{\MMD}^2(\mathbb{P}_{\theta} || \mathbb{Q}^m),
\end{align}
 where $\mathbb{Q}^{m}(\mathrm{d}y) = \frac{1}{m}\sum_{i=1}^{m}\delta_{y_i}(\mathrm{d}y)$, and $\{y_i\}_{i=1}^n \iid \mathbb{Q}$. In the following we will use $\mathbb Q^m$ to denote both the random measure $\mathbb Q^m$ and the measure $\mathbb Q^m(\omega)$, and we shall assume that $\mathcal{P}_{\Theta}(\mathcal{X}) \subset \mathcal{P}_k(\mathcal{X})$. Several existing methodologies fall within this general framework, including kernel scoring rules \citep{Eaton1982} and MMD GANs \citep{Dziugaite2015,Li2015}. For analogous methodology in a Bayesian context, see kernel ABC \citep{Fukumizu2013,Park2015}.

In general, the optimisation problem will not be convex and the minimiser $\hat{\theta}_m$ will not be computable analytically. If the generator $G_{\theta}$ is differentiable with respect to $\theta$ with a computable Jacobian matrix, the minimiser will be a fixed point of the equation $\dot{\theta} = -\nabla_{\theta}{\MMD}^2(\mathbb{P}_{\theta} || \mathbb{Q}^m)$ where $\nabla_{\theta} = (\partial_{\theta_1},\ldots,\partial_{\theta_p})$. Assuming that the Jacobian $\nabla_{\theta} G_{\theta}$ is $\mathbb{U}$-integrable then the gradient term can be written as 
\begin{align*}
 \nabla_{\theta} {\MMD}^2(\mathbb{P}_{\theta}  ||  \mathbb{Q}^m) 
  &= 2\int_{\mathcal{U}} \int_{\mathcal{U}} \nabla_{1} k(G_{\theta}(u), G_{\theta}(v)) \nabla_{\theta}G_{\theta}(u)\mathbb{U}(\mathrm{d}u)\mathbb{U}(\mathrm{d}v)
   \\& - \frac{2}{m}\sum_{j=1}^{m}\int_{\mathcal{U}} \nabla_1 k(G_{\theta}(u), y_j) \nabla_{\theta}G_{\theta}(u)\mathbb{U}(\mathrm{d}u),
\end{align*}
where $\nabla_1 k$ corresponds to the partial derivative with respect to the first argument.  In practice it will not be possible to compute the integral terms analytically.  We can introduce a U-statistic approximation for the gradient as follows: 
\begin{align*}
\hat{J}_\theta(  \mathbb{Q}^m) = &
= \frac{2 \sum_{i\neq i'} \nabla_{\theta}G_{\theta}(u_i)\nabla_{1} k(G_{\theta}(u_{i}), G_{\theta}(u_{i'})) }{n(n-1)}
   - \frac{2 \sum_{j=1}^{m}\sum_{i=1}^{n} \nabla_{\theta}G_{\theta}(u_i)\nabla_1 k(G_{\theta}(u_i), y_j)}{nm},
\end{align*}
where $\{u_i\}_{i=1}^n \iid \mathbb{U}$.
This is an unbiased estimator in the sense that $\mathbb{E}[\hat{J}_\theta(  \mathbb{Q}^m) ]= \nabla_{\theta} {\MMD}^2(\mathbb{P}_{\theta}  ||  \mathbb{Q}^m) $, where the expectation is taken over the independent realisations of the $u_i's$. This allows us to use a stochastic gradient descent (SGD) \citep{Dziugaite2015,Li2015}: starting from $\hat{\theta}^{(0)} \in \Theta$, we iterate:
\begin{enumerate}
\item [(i)] Sample $\{u_i\}_{i=1}^n \iid \mathbb{U}$ and set  $x_i = G_{\hat{\theta}^{(k-1)}}(u_i)$ for $i=1,\ldots,n$.
\item [(ii)] Compute $\hat{\theta}^{(k)} =  \hat{\theta}^{(k-1)} -\eta_k \hat{J}_{\hat{\theta}^{(k-1)}}(  \mathbb{Q}^m)$.
\end{enumerate}
where $(\eta_k)_{k \in \mathbb{N}}$ is a step size sequence chosen to guarantee convergence (see \citep{robbins1985stochastic}) to the minimiser in Equation \ref{eq:MMD_estimator_m}. For large values of $n$, the SGD should approach $\hat{\theta}_m$, but this will come at significant computational cost. Let $\mathcal{X} \subseteq \mathbb{R}^d$ and $\Theta \subseteq \mathbb{R}^p$. The overall cost of the gradient descent algorithm is $\mathcal{O}\left((n^2+nm)d p\right)$ per iteration. This cost is linear in the number of data points $m$, but quadratic in the number of simulated samples $n$. It could be made linear in $n$ by considering approximations of the maximum mean discrepancy as found in \cite{Chwialkowski2015}. In large data settings (i.e. $m$ large), subsampling elements uniformly at random from $\{y_j\}_{j=1}^m$ may lead to significant speed-ups.

Clearly, when the generator $G_{\theta}$ and its gradient $\nabla_{\theta} G_{\theta}$ are computationally intensive, letting $n$ grow will become effectively intractable, and it will be reasonable to assume that the number of simulations $n$ is commensurate or even smaller than the sample size. To study the behaviour of minimum MMD estimators when synthetic data is prohibitively expensive, we consider the following minimum divergence estimator: $\hat{\theta}_{n, m} =  \argmin_{\theta \in \Theta} {\MMD}^2_{U,U}(\mathbb{P}^n_{\theta}||\mathbb{Q}^m)$ based on a U-statistic approximation of the MMD:
\begin{align*}
{\MMD}^2_{U,U}(\mathbb{P}_{\theta}^n||\mathbb{Q}^m)
& = 
\frac{\sum_{i \neq i'} k(x_i,x_{i'})}{n(n-1)}  - \frac{2 \sum_{j=1}^m \sum_{i=1}^n k(x_i,y_j)}{m n}  + \frac{\sum_{j \neq j'} k(y_j,y_{j'})}{m(m-1)}.
\end{align*}
where $\mathbb{P}_{\theta}^n = \frac{1}{n}\sum_{i=1}^n \delta_{x_i}$ for some $\{x_i\}_{i=1}^n \iid \mathbb{P}_\theta$. This estimator is closely related to the method of simulated moments \citep{Hall2005} and satisfies $\mathbb{E}[{\MMD}^2_{U,U}(\mathbb{P}_{\theta}^n||\mathbb{Q}^m)] = {\MMD}^2(\mathbb{P}_{\theta}||\mathbb{Q})$, thus providing an unbiased estimator of the square distance between $\mathbb{P}_{\theta}$ and $\mathbb{Q}$. While the estimator $\hat{\theta}_{n,m}$ is not used in practice (since we re-sample from the generator at each gradient iteration), it is an idealisation which gives us insights into situations where the gradient descent cannot be iterated for a large numbers of steps relative to the observed data-set size, and so we cannot appeal on the law of large numbers.

\subsection{The Information Geometry induced by MMD}

The two estimators $\hat{\theta}_m$ and $\hat{\theta}_{n,m}$ defined above are flexible in the sense that the choice of kernel and kernel hyperparameters will have a significant influence on the geometry induced on the space of probability measures. This section studies this geometry and develops tools which will later give us insights into the impact of the choice of kernel on the generalisation, asymptotic convergence and robustness of the corresponding estimators.

Let $\mathcal{P}_{\Theta}(\mathcal{X})$ be a family of measures contained in $\mathcal{P}_k(\mathcal{X})$ and parametrised by an open subset $\Theta \subset \mathbb{R}^p$.  Assuming that the map $\theta \rightarrow \mathbb{P}_{\theta}$ is injective, the MMD distance between the elements $\mathbb{P}_{\theta}$ and $\mathbb{P}_{\theta'}$ in $\mathcal{P}_k$ induces a distance between $\theta$ and $\theta'$ in $\Theta$. Under appropriate conditions this gives rise to a Riemmanian manifold structure on $\Theta$. The study of the geometry of such statistical manifolds lies at the center of information geometry \citep{Amari1987,Barndorff-Nielsen1978}.  Traditional information geometry focuses on the statistical manifold induced by the Kullback-Leibler divergence over a parametrised set of probability measures. This yields a Riemmanian structure on the parameter space with the metric tensor given by the Fisher-Rao metric.  A classic result due to \cite{cencov2000statistical} characterises this metric as the unique metric invariant under a large class of transformations (i.e. embeddings via Markov morphisms, see \citep{campbell1986extended,montufar2014fisher}).

In this section, we study instead the geometry induced by MMD. To fix ideas, we shall consider a generative model distribution of the form $\mathbb{P}_{\theta} = G_{\theta}^{\#} \mathbb{U}$ for $\theta \in \Theta$, where $(\mathcal{U}, \mathcal{F}, \mathbb{U})$ is an underlying Borel measure space.   We assume that (i) $G_{\theta}(\cdot)$ is $\mathcal{F}$-measurable for all $\theta \in \Theta$; (ii) $G_{\cdot}(u) \in C^1(\Theta)$ for all $u \in \mathcal{U}$; (iii) $\|\nabla_{\theta} G_{\theta}(\cdot)\| \in L^1(\mathbb{U})$, for all $\theta \in \Theta$.  Suppose additionally that the kernel $k$ has bounded continuous derivatives over $\mathcal{X}\times \mathcal{X}$.  Define the map $J:\Theta \rightarrow \mathcal{H}_k$ to be the Bocher integral $J(\theta) = \Pi_{k}(\mathbb{P}_{\theta})$. By \cite[Theorem 90]{hajek2014smooth}, assumptions (i)-(iii)  imply that the map $J$ is Fr\'{e}chet differentiable and
\begin{align*}
\partial_{\theta_i} J(\theta)(\cdot) & = \int_{\mathcal{U}} \nabla_2 k(\ \cdot \ , G_{\theta}(u))\partial_{\theta_i} G_{\theta}(u)\mathbb{U}(\mathrm{d}u).
\end{align*}
The map $J$ induces a degenerate-Riemannian metric $g(\theta)$ on $\Theta$ given by the pull-back of the inner product on $\mathcal{H}_k$. 
In particular its components in the local coordinate-system are $g_{ij}(\theta) = \langle \partial_{\theta_{i}} J(\theta), \partial_{\theta_{j}} J(\theta)\rangle_{\mathcal{H}_k}$ for $i, j \in \lbrace 1, \ldots, p \rbrace$.
By \citep[Lemma 4.34]{Steinwart2008}, it follows that for $i,j \in \lbrace 1, \ldots, p \rbrace$,
\begin{align}
\label{eq:metric_tensor}
g(\theta) & = \int_{\mathcal{U}} \int_{\mathcal{U}}\nabla_{\theta} G_{\theta}(u)^\top  \nabla_2 \nabla_1 k(G_{\theta}(u),G_{\theta}(v))  \nabla_{\theta} G_{\theta}(v) \mathbb{U}(\mathrm{d}u)\mathbb{U}(\mathrm{d}v),
\end{align}
where $\nabla_1\nabla_2 k(x,y) = \lbrace \partial_{x_i}\partial_{y_j} k(x, y) \rbrace_{i,j=1,\ldots, d}$.
The induced metric tensor is in fact just the information metric associated to the MMD-squared divergence (see \ref{MMD2_metrictensor}). Further details about the geodesics induced by MMD can be found in Appendix \ref{appendix:information_geometry}. This information metric will allow us to construct efficient optimisation algorithm and study the statistical properties of the minimum MMD estimators.

\subsection{MMD Gradient Flow}

Given the loss function $L(\theta) = {\MMD}^2(\mathbb{P}_{\theta} || \mathbb{Q}^m)$, a standard approach to finding a minimum divergence estimator is via gradient descent (or in our case stochastic gradient descent). Gradient descent methods aim to minimise a function $L$ by following a curve $\theta(t)$, known as the gradient flow, that is everywhere tangent to the direction of steepest descent of $L$. 
This direction depends on the choice of Riemannian metric $g$ on $\Theta$, and is given by $-\nabla_gL$ where $\nabla_gL$ denotes the Riemannian gradient (or covariant derivative) of $L$. 

A particular instance of gradient descent, based on the Fisher Information metric, was developed by Amari and collaborators \citep{Amari1998}. It is a widely used alternative to standard gradient descent methods and referred to as natural gradient descent. It has been successfully applied to a variety of problems in machine learning and statistics, for example reinforcement learning \citep{kakade2002natural}, neural network training \citep{park2000adaptive}, Bayesian variational inference methods \citep{hoffman2013stochastic} and Markov chain Monte Carlo \citep{Girolami2011}.  While the classical natural gradient approach is based on the Fisher information matrix induced by the KL divergence, information geometries arising from other metrics on probabilities have also been studied in previous works, including those arising from optimal transport metrics \citep{chen2018natural,li2018natural} and the Fisher divergence \citep{karakida2016adaptive}.

 As discussed above, a gradient descent method can be formulated as an ordinary differential equation for the  \emph{gradient flow} $\theta(t)$ which solves $\dot{\theta}(t) = -\nabla_{g}L(\theta(t))$ for some specified initial conditions. In local coordinates the Riemannian gradient can be expressed in terms of the standard gradient $\nabla_{\theta}$, formally $\nabla_g = g^{-1}(\theta) \nabla_{\theta}$, so we have $\dot{\theta}(t) = -g^{-1}(\theta)\nabla_{\theta}L(\theta)$. This flow can be approximated by taking various discretisations.  An explicit Euler discretisation yields the scheme: $\theta^{(k)} = \theta^{(k-1)} - \eta_k g^{-1}(\theta^{(k-1)})\nabla_{\theta} L(\theta^{(k-1)})$. Under appropriate conditions on the step-size sequence $(\eta_k)_{k \in \mathbb{N}}$ this gradient descent scheme will converge to a local minimiser of $L(\theta)$.  Provided that $\nabla_{\theta}L(\theta)$ and the metric tensor are readily computable,  the Euler discretisation yields a gradient scheme analogous to those detailed in \citep{Amari1987,Amari1998}. 

For the MMD case, we cannot evaluate $g$ from Equation \eqref{eq:metric_tensor} exactly since it contains intractable integrals against $\mathbb{U}$. We can however use a similar approach to that used for the stochastic gradient algorithm and introduce a U-statistic approximation of the intractable integrals:
\begin{align*}
  g_U(\theta) & = \frac{1}{n(n-1)}\sum_{i \neq j}\nabla_{\theta} G_{\theta}(u_i)^\top  \nabla_2 \nabla_1 k(G_{\theta}(u_i),G_{\theta}(u_j))  \nabla_{\theta} G_{\theta}(u_j),
\end{align*}
where $\{u_i\}_{i=1}^n$ are IID realisations from $\mathbb{U}$. We propose to perform optimisation using the following natural stochastic gradient descent algorithm: starting from $\hat{\theta}^{(0)} \in \Theta$, we iterate 
\begin{enumerate}
\item [(i)] Sample $\{u_i\}_{i=1}^n \iid \mathbb{U}$ and set  $x_i = G_{\hat{\theta}^{(k-1)}}(u_i)$ for $i=1,\ldots,n$.
\item [(ii)] Compute $\hat{\theta}^{(k)} =  \hat{\theta}^{(k-1)} -\eta_k g_{U}\big(\hat{\theta}^{(k-1)}\big)^{-1}\hat{J}_{\hat{\theta}^{(k-1)}}\left(\mathbb{Q}^m\right)$.
\end{enumerate}
The experiments in Section \ref{sec:MMDestimators_experiments} demonstrate that this new algorithm can provide significant computational gains. This could be particularly impactful for GANs, where a large number of stochastic gradient descent are currently commonly used. The approximation of the inverse metric tensor does however yield an additional computational cost due to the inversion of a dense matrix: $\mathcal{O}(((n^2+nm)p^2d+p^3))$ per iteration. When the dimension of the parameter set $\Theta$ is high, the calculation of the inverse metric at every step can hence be prohibitive. The use of online methods to approximate $g^{-1}$ sequentially without needing to compute inverses of dense matrices can be considered as in \citep{ollivier2018online}, or alternatively, approximate linear solvers could also be used to reduce this cost. 

In certain cases, the gradient of the generator $\nabla_{\theta} G_\theta$ may not be available in closed form, precluding exact gradient descent inference. An alternative is the method of finite difference stochastic approximation \citep{Kushner2003} can be used to approximate an exact descent direction.  Alternatively, one can consider other discretisations of the gradient flow. For example, a fully implicit discretisation yields the following scheme \citep{jordan1998variational}:
\begin{align}
\label{eq:implicit}
\theta^{(k)} & = \arg\min_{\theta \in \Theta}L(\theta) + \frac{1}{2\eta}{\MMD}^2(\mathbb{P}_{\theta} || \mathbb{P}_{\theta^{(k-1)}}),
\end{align}
where $\eta > 0$ is a step-size.  Therefore the natural gradient flow can be viewed as a motion towards a lower value of $L(\theta)$ but constrained to be close (in MMD) to the previous time-step. The constant $\eta$ controls the strength of this constraint, and thus can be viewed as a step size. The formulation allows the possibility of a natural gradient descent approach being adopted even if $\nabla_{\theta} L$ and $g$ are not readily computable.  Indeed, \eqref{eq:implicit} could potentially be minimised using some gradient-free optimisation method such as Nelder-Mead.


 \subsection{Minimum MMD Estimators and Kernel Scoring Rules}

Before concluding this background section, we highlight the connection between our minimum MMD estimators and scoring rules \citep{Dawid2007}. A scoring rule is a function $S:\mathcal{X} \times \mathcal{P}(\mathcal{X}) \rightarrow \mathbb{R}$ such that $S(x,\mathbb{P})$ quantifies the accuracy of a model $\mathbb{P}$ upon observing the realisation $x \in \mathcal{X}$ (see \citep{Gneiting2007} for technical conditions). We say a scoring rule is strictly proper if $\int_{\mathcal{X}} S(x,\mathbb{P}_1) \mathbb{P}_2(\mathrm{d}x)$ is uniquely minimised when $\mathbb{P}_1 = \mathbb{P}_2$. Any strictly proper scoring rule induces a divergence of the form $D_S(\mathbb{P}_1||\mathbb{P}_2) = \int_{\mathcal{X}} S(x,\mathbb{P}_1) \mathbb{P}_2(\mathrm{d}x)-\int_{\mathcal{X}} S(x,\mathbb{P}_2) \mathbb{P}_2(\mathrm{d}x)$. These divergences can then be used to obtain minimum distance estimators: $\hat{\theta}^S_m = \argmin_{\theta \in \Theta} D_S(\mathbb{P}_{\theta}||\mathbb{Q}^m) = \argmin_{\theta \in \Theta} \sum_{j=1}^m S(y_j,\mathbb{P}_{\theta})$. One way to solve this problem is by setting the gradient in $\theta$ to zero; i.e. solving $\textstyle \sum_{j=1}^m \nabla_{\theta} S(y_j,\mathbb{P}_{\theta}) = 0$, called estimating equations. 

The minimum MMD estimators $\hat{\theta}_m$ in this paper originate from the well-known kernel scoring rule \citep{Eaton1982,Dawid2007,Zawadzki2015,Steinwart2017,Masnadi-Shirazi2017}, which takes the form
\begin{align*}
S(x,\mathbb{P}) =k(x,x) - 2 \int_{\mathcal{X}} k(x,y) \mathbb{P}(\mathrm{d}y) + \int_{\mathcal{X}}\int_{ \mathcal{X}} k(y,z) \mathbb{P}(\mathrm{d}y)\mathbb{P}(\mathrm{d}z).
\end{align*}
This connection between scoring rules and minimum MMD estimators will be useful for theoretical results in the following section. Whilst the present paper focuses on minimum MMD estimators for generative models, our results also have implications for kernel scoring rules.


\section{Behaviour of Minimum MMD estimators} \label{sec:MMDestimators_theory}

The two estimators $\hat{\theta}_n$ and $\hat{\theta}_{n,m}$ defined above are flexible in the sense that the choice of kernel and kernel hyperparameters will have a significant influence on the geometry induced on the space of probability measures. This choice will also have an impact on the generalisation, asymptotic convergence and robustness of the estimators, as will be discussed in this section.

\subsection{Concentration and Generalisation Bounds for MMD}
In this section we will restrict ourselves to the case where $\mathcal{X} \subset \mathbb{R}^d$ and $\Theta \subset \mathbb{R}^p$ for $d, p \in \mathbb{N}$. Given observations $\lbrace y_i \rbrace_{i=1}^{m} \iid \mathbb{Q}$, it is clear that the convergence and efficiency of $\hat{\theta}_m$ and $\hat{\theta}_{n, m}$ in the limit of large $n$ and $m$ will depend on the choice of kernel $k$ as well as the dimensions $p$ and $d$.  As a first step, we obtain estimates for the out-of-sample error for each estimator, in the form of generalization bounds.   

The necessary conditions in this proposition are quite natural. They are required to ensure the existence of $\hat{\theta}_m$ and $\hat{\theta}_{n,m}$, and reclude models which are unidentifiable over a non-compact subset of parameters, i.e. models for which there are minimising sequences $\hat{\theta}_m$ of ${\MMD}(\mathbb{P}_{\theta}  ||  \mathbb{Q}^m)$ which are unbounded. While these assumptions must be verified on a case-by-case basis, for most models we would expect these conditions to hold immediately.
\begin{assumption}
\label{ass:existence}
\begin{enumerate}

	\item For every $\mathbb{Q} \in \mathcal{P}_{k}(\mathcal{X})$, there exists $c>0$ such that the set $\lbrace \theta \in \Theta : \MMD(\mathbb{P}_{\theta} || \mathbb{Q}) \leq \inf_{\theta' \in \Theta}  \MMD(\mathbb{P}_{\theta'} || \mathbb{Q}) + c \rbrace,$
	is bounded.

	\item For every  $n \in \mathbb{N}$ and $\mathbb{Q} \in \mathcal{P}_{k}(\mathcal{X})$, there exists $c_n > 0$ such that  the set 
	$\lbrace \theta \in \Theta : \MMD(\mathbb{P}_{\theta}^n || \mathbb{Q}) \leq \inf_{\theta' \in \Theta}  \MMD(\mathbb{P}_{\theta'} || \mathbb{Q}) + c_n \rbrace,$
	is almost surely bounded.
\end{enumerate}
\end{assumption}

\begin{theorem}[Generalisation Bounds]
\label{thm:generalisation}
Suppose that the kernel $k$ is bounded, and that Assumption \ref{ass:existence} holds, then with probability at least  $1-\delta$,
\begin{align*}
{\MMD}\left(\mathbb{P}_{\hat{\theta}_{m}} \big|\big| \mathbb{Q}\right) & \leq  \inf_{\theta \in \Theta} {\MMD}(\mathbb{P}_{\theta}||\mathbb{Q}) + 2\sqrt{\frac{2}{m}\sup_{x\in \mathcal X} k(x,x)}\left(2 + \sqrt{\log\left(\frac{1}{\delta}\right)}\right),
\end{align*}
and
\begin{align*}
{\MMD}\left(\mathbb{P}_{\hat{\theta}_{n,m}} \big|\big| \mathbb{Q}\right) & \leq  \inf_{\theta \in \Theta} {\MMD}(\mathbb{P}_{\theta}|| \mathbb{Q}) + 2\left(\sqrt{\frac{2}{n}}+ \sqrt{\frac{2}{m}}\right)\sqrt{\sup_{x\in\mathcal X}k(x,x)}\left(2 + \sqrt{\log\left(\frac{2}{\delta}\right)}\right).
\end{align*}
\end{theorem}
All proofs are deferred to Appendix \ref{appendix:proofs}. An immediate corollary of the above result is that the speed of convergence in the generalisation errors decreases as $n^{-\frac{1}{2}}$ and $m^{-\frac{1}{2}}$ with the rates being independent of the dimensions $p$ and $d$, and the properties of the kernel.  Indeed, if the kernel is translation invariant, then $k(x,x)$ will reduce to the maximum value of the kernel. A similar generalisation result was obtained in \cite{Dziugaite2015} for minimum MMD estimation of deep neural network models.  While the bounds are of the same form, Theorem \ref{thm:generalisation} only requires minimal assumptions on the smoothness of the kernel.  Moreover, all the constants in the bound are explicit, demonstrating clearly dimensional dependence.   Assumption \ref{ass:existence} is required to guarantee the existence of at least one minimiser, whereas this is implicitly assumed in \cite{Dziugaite2015}.  The key result which determines the rate is the following concentration inequality.
\begin{lemma}[Concentration Bound]
\label{lemma:concentration}
Assume that the kernel $k$ is bounded and let $\mathbb{P}$ be a probability measure on $\mathcal{X} \subseteq \mathbb{R}^d$.  Let $\mathbb{P}^n$ be the empirical measure obtained from $n$ independently and identically distributed samples of $\mathbb{P}$.  Then with probability $1-\delta$, we have that
\begin{align*}
{\MMD}(\mathbb{P} || \mathbb{P}^n) & \leq  \sqrt{\frac{2}{n}\sup_{x\in \mathcal{X}}k(x,x)}\left(1 + \sqrt{\log\left(\frac{1}{\delta}\right)}\right).
\end{align*}
\end{lemma}
See also \cite[Theorem 17]{Gretton2009} for an equivalent bound.  We can compare this result with \citep[Theorem 1]{fournier2015rate} on comparing the rate of convergence of Wasserstein-1 distance (denoted $W_1$) to the empirical measure, which implies that for $d > 2$ and $q$ sufficiently large, with probability $1-\delta$ we have $W_1(\mathbb{P} || \mathbb{P}^n) \leq C M_{q}^{1/q}(\mathbb{P}) \delta^{-1} n^{-1/d}$,
where $M_q(\mu) := \int_{\mathcal{X}} |x|^q \mu(\mathrm{d}x)$ and $C$ is a constant depending only on the constants $p,q$ and $d$.  This suggests that generalisation bounds analogous to Theorem \ref{thm:generalisation} for Wasserstein distance would depend exponentially on dimension, at least when the distribution is absolutely continuous with respect to the Lebesgue measure.  For measures support on a lower dimensional manifold, this bound has been recently tightened, see \cite{Weed2017} and also \cite{Weed2019}.  For Sinkhorn divergences, which interpolate between optimal transport and MMD distance \cite{Genevay2017} this curse of dimensionality can be mitigated \cite{Genevay2019} for measures on bounded domains.

\subsection{Consistency and Asymptotic Normality}
With additional assumptions, we can recover a classical strong consistency result. 
\begin{proposition}[Consistency]
\label{prop:consistency}
Suppose that Assumption \ref{ass:existence} holds and that there exists a unique minimiser $\theta^* \in \Theta$ such that ${\MMD}(\mathbb{P}_{\theta^*}  ||  \mathbb{Q}) = \inf_{\theta \in \Theta}{\MMD}(\mathbb{P}_{\theta} || \mathbb{Q})$.
Then $\lim_{m\rightarrow \infty}\hat{\theta}_m = \theta^*$ and $\lim_{m,n\rightarrow \infty} \hat{\theta}_{m,n} = \theta^*$ as $n, m \rightarrow \infty$, almost surely.
\end{proposition}

Theorem \ref{thm:generalisation} provides fairly weak probabilistic bounds on the convergence of the estimators $\hat{\theta}_m$ and $\hat{\theta}_{n,m}$ in terms of their MMD distance to the data distribution $\mathbb{Q}$.  Proposition \ref{prop:consistency} provides conditions under which these bounds translate to convergence of the estimators, however it is not clear how to extract quantitative information about the speed of convergence, and the efficiency of the estimator in general.   A classical approach to this is to establish the asymptotic normality of the estimators and characterise the efficiency in terms of the asymptotic variance.  We do this now, assuming that we are working in the $M$-close setting, i.e. assuming that $\mathbb{Q} = \mathbb{P}_{\theta^*}$ for some $\theta^*$.
\begin{theorem}[Central Limit Theorems]
\label{thm:asympt_normal2}
Suppose that $\mathbb{Q} = \mathbb{P}_{\theta^*}$ for some $\theta^* \in \Theta$  and that the conclusions of Proposition \ref{prop:consistency} hold.  Suppose that:
\begin{enumerate}
\item There exists an open neighbourhood $O \subset \Theta$ of $\theta^*$ such that $G_{\theta}$ is three times differentiable in $O$ with respect to $\theta$.
\item The information metric $g(\theta)$ is positive definite at $\theta = \theta^*$.
\item There exists a compact neighbourhood $K \subset O$  of $\theta^*$ such that $\int_{\mathcal{U}} \sup_{\theta \in K} \left\lVert \nabla^{(i)}G_{\theta}(u) \right\rVert \mathbb{U}(\mathrm{d}u) < \infty$ for $i=1,2,3$ where $\nabla^{(i)}$ denotes the mixed derivatives of order $i$ and $\lVert \cdot \rVert$ denotes the spectral norm.  
\item The kernel $k(\cdot, \cdot)$ is translation invariant, with bounded mixed derivatives up to order $2$. 
\end{enumerate}
Then as $k\rightarrow \infty$:
\begin{align*}
\sqrt{m}\left(\hat{\theta}_{m} - \theta^*\right) & \xrightarrow{d} \mathcal{N}(0, C),
\end{align*}
where $\xrightarrow{d}$ denotes convergence in distribution. The covariance matrix is given by the \emph{Godambe matrix} $C =  g(\theta^*)^{-1}\Sigma g(\theta^*)^{-1}$ where 
\begin{align*}
\Sigma  & =  \int_{\mathcal{U}}\left(\int_{\mathcal{U}}\left(\nabla_{1} k(G_{\theta^*}(u), G_{\theta^*}(v))\nabla_{\theta} G_{\theta^*}(u) - \overline{M}\right) \mathbb{U}(\mathrm{d}u)\right)^{\otimes 2} \mathbb{U}(\mathrm{d}v)
\end{align*}
and
\begin{align*}
\overline{M} & = \int_{\mathcal{U}} \int_{\mathcal{U}} \nabla_{1}k(G_{\theta^*}(u), G_{\theta^*}(v)) \nabla_{\theta} G_{\theta^*}(u)\mathbb{U}(\mathrm{d}u)\mathbb{U}(\mathrm{d}v).
\end{align*}
Here, $A\otimes B$ denotes the tensor product and $A^{\otimes 2} := A\otimes A$. Furthermore, suppose that:
\begin{enumerate}
\item [5] The kernel $k(\cdot, \cdot)$ has bounded mixed derivatives up to order $3$. 
\item [6] The indices satisfy $n = n_k$, $m = m_k$ where $n_k/(n_k + m_k) \rightarrow \lambda \in (0, 1)$,
\end{enumerate}
Then, as $k\rightarrow \infty$,
\begin{align*}\textstyle
\sqrt{n_k + m_k}\left(\hat{\theta}_{n, m} - \theta^*\right)& \xrightarrow{d} \mathcal{N}(0, C_{\lambda}),
\end{align*}
where $C_{\lambda} = (1/(1-\lambda)\lambda) C$.
\end{theorem}
We remark that the asymptotic covariance $C_{\lambda}$ is minimised when $\lambda = 1/2$, that is, when the number of samples $n$ generated from the model equals that of the data $m$ (at which point $C_{\lambda} = 4 C$). This means that it will be computationally inefficient to use $n$ much larger than $m$. We note that the variance also does not depend on any amplitude parameter of the kernel, or any location parameters in $\mathbb{U}$. To the best of our knowledge, there are no known analogous result for minimum Wasserstein or Sinkhorn estimators (except a one-dimensional result for the minimum Wasserstein estimator in the supplementary material of \citep{Bernton2019}).

Theorem \ref{thm:asympt_normal2} raises the question of efficiency of the estimator. The Cramer-Rao bound provides a lower bound on the variance of any unbiased estimator for $\mathbb{P}_{\theta}$, and it is well-known that it is attained by maximum likelihood estimators. The following result is an adaptation of the Cramer-Rao bound in \cite{godambe1960optimum} for our estimators, which are biased. 
\begin{theorem}[Cramer-Rao Bounds]
\label{prop:cramer_rao}
Suppose that the CLTs in Theorem \ref{thm:asympt_normal2} hold and that the data distribution $\mathbb{Q}$ satisfies $\mathbb{Q}=\mathbb{P}_{\theta^*},$
where $\mathbb{P_{\theta^*}=}G_{\theta^*}^{\#}\mathbb{U}$ is assumed
to have density $p(x|\theta^*)$. Furthermore, suppose that the MMD information
metric $g(\theta^*)$ and the Fisher information metric $F(\theta) = \int_{\mathcal{X}} \nabla_{\theta} \log p(x|\theta)\nabla_{\theta} \log p(x|\theta)^\top \mathbb{P}_{\theta}(\mathrm{d}x)$ are positive definite when $\theta = \theta^*$.  Then the asymptotic covariances $C$ and $C_{\lambda}$ of the estimator $\hat{\theta}_{m}$ and $\hat{\theta}_{n,m}$ satisfy the Cramer-Rao bound, i.e. $C-F(\theta^*)^{-1}$ and $C_{\lambda}-F(\theta^*)^{-1}$ are non-negative definite.
\end{theorem}
The results above demonstrate that we cannot expect our (biased) estimators to outperform maximum likelihood in the M-closed case. The efficiency of these estimators is strongly determined by the choice of kernel, in particular on the kernel bandwidth $l$. The following result characterises the efficiency as $l\rightarrow \infty$. 
\begin{proposition}[Efficiency with Large Lengthscales]\label{prop:bandwidth_limit}
Suppose that $k$ is a radial basis kernel, i.e. $k(x,y)=r(|x-y|^{2}/2l^{2})$, where $\lim_{s\rightarrow0}r'(s) < \infty$ and  $\lim_{s\rightarrow0}r''(s) < \infty$.  Let $C^{l}$ and $C^{l}_{\lambda}$ denote the asymptotic variance as a function of the bandwidth $l$ of $\hat{\theta}_{m}$ and $\hat{\theta}_{n,m}$ respectively. Then 
\begin{align}
\label{eq:limiting_variance}
\lim_{l\rightarrow\infty}C^{l} & =  \left(\nabla_{\theta}M(\theta)\right)^{\dagger}V(\theta)\left(\nabla_{\theta}M(\theta)\right)^{\dagger\top},
\end{align}
where  $M(\theta)$ and $V(\theta)$ are the mean and covariance of $p(x|\theta)$ respectively and $A^{\dagger}$ denotes the Moore-Penrose inverse of $A$. As a result, $\lim_{l\rightarrow\infty}C^{l}_{\lambda} = (1/(1-\lambda)\lambda) \left(\nabla_{\theta}M(\theta)\right)^{\dagger}V(\theta)\left(\nabla_{\theta}M(\theta)\right)^{\dagger\top}$.
\end{proposition}
In general, the minimum MMD estimators may not achieve the efficiency of maximum likelihood estimators in the limit $l\rightarrow \infty$, however in one dimension, the limiting covariance in Equation \ref{eq:limiting_variance} is a well known approximation for the inverse Fisher information \citep{jarrett1984bounds, stein2017pessimistic}, which is optimal. 

Before concluding this section on efficiency of minimum MMD estimators, we note that the asymptotic covariances $C$ and $C_{\lambda}$ of Theorem \ref{thm:asympt_normal2} could be used to create confidence intervals for the value of $\theta^*$ (only for the M-closed case). Although these covariances cannot be estimated exactly since they depend on $\theta^*$ and contain intractable integrals, they can be approximated using the generator at the current estimated value of the parameters. 

\subsection{Robustness}

This concludes our theoretical study of the M-closed case and we now move on to the M-open case. A concept of importance to practical inference is robustness when subjected to corrupted data \citep{Huber2009}. As will be seen below, minimum MMD estimators have very favourable robustness properties for this case.

Our first objective is to demonstrate qualitative robustness in the sense of \cite{Hampel1971}. More specifically, given some parametrized probability measure $\mathbb{P}_\theta$, we show that if two measures $\mathbb{Q}_1$ and $\mathbb{Q}_2$ are close in Prokhorov metric, then the distributions of the minimum distance estimators $\hat{\theta}_{m}^{i} \in \argmin_{\theta \in \Theta} \MMD^2(\mathbb{P}_\theta||\mathbb{Q}^m_i)$ and $\hat{\theta}^{i}_{n,m} \in \argmin_{\theta \in \Theta} \MMD^2(\mathbb{P}^n_\theta||\mathbb{Q}^m_i)$ for $i=1,2$ are respectively close.

\begin{theorem}[Qualitative Robustness]\label{thm:qualitative_robustness}
Suppose that (i) $\forall \mathbb{Q} \in \mathcal{P}_k(\mathcal{X})$ there exists a unique $\textstyle \theta^{\mathbb{Q}}$ such that $\textstyle \inf_{\theta\in \Theta}\MMD(\mathbb{P}_{\theta} || \mathbb{Q})=\MMD(\mathbb{P}_{\theta^{\mathbb{Q}}} ||\mathbb{Q})$ and (ii) $\forall \epsilon > 0$,  $\exists \delta > 0$ such that $\textstyle \lVert\theta - \theta^{\mathbb{Q}}\rVert \geq \epsilon$ implies that $\textstyle \MMD(\mathbb{P}_{\theta} || \mathbb{Q}) > \MMD(\mathbb{P}_{\theta^{\mathbb{Q}}}  || \mathbb{Q}) + \delta$.  Then $\textstyle \hat{\theta}_m$ is qualitatively robust in the sense of \cite{Hampel1971}. 

Additionally, suppose that for any empirical measure $\mathbb{U}^n$ on $n$ points, that (i') $\textstyle \forall \mathbb{Q} \in \mathcal{P}_k(\mathcal{X})$ there exists a unique $\theta^{\mathbb{Q}}$ such that $\textstyle \inf_{\theta \in \Theta}{\MMD}(G_{\theta}^{\#}\mathbb{U}^n || \mathbb{Q}) = {\MMD}(G_{\theta^{\mathbb{Q}}}^{\#}\mathbb{U}^n || \mathbb{Q})$ and (ii') $\forall \epsilon > 0$, $\exists \delta > 0$ such that $\lVert \theta-\theta^{\mathbb{Q}}\rVert \geq \epsilon$ implies that $\textstyle \MMD(G_{\theta}^{\#} \mathbb{U}^n ||\mathbb{Q}) > \MMD(G_{\theta^{\mathbb{Q}}}^{\#} \mathbb{U}^n ||\mathbb{Q}) + \delta$.  Then $\exists N$ such that $\hat{\theta}_{n,m}$ is qualitatively robust for $n \geq N$.   
\end{theorem}

The result above characterises the qualitative robustness of the estimators, but does not provide a measure of the degree of robustness which can be used to study the effect of corrupted data on the estimated parameters. An important quantity used to quantify robustness is the \emph{influence function} $\text{IF}:\mathcal{X} \times \mathcal{P}_{\Theta}(\mathcal{X}) \rightarrow \mathbb{R}$ where $\text{IF}(z,\mathbb{P}_{\theta})$ measures the impact of an infinitesimal contamination of the data generating model $\mathbb{P}_{\theta}$ in the direction of a Dirac measure $\delta_{z}$ located at some point $z \in \mathcal{X}$. The influence function of a minimum distance estimator based on a scoring rule $S$ is given by \citep{Dawid2014}: $ \text{IF}_S(z,\mathbb{P}_\theta) := ( \int_{\mathcal{X}} \nabla_\theta \nabla_\theta S(x,\mathbb{P}_\theta) \mathbb{P}_{\theta}(\mathrm{d}x) )^{-1} \nabla_{\theta}S(z,\mathbb{P}_{\theta})$. The supremum of the influence function over $z \in \mathcal{X}$ is called the gross-error sensitivity, and if it is finite, we say that an estimator is \emph{bias-robust} \citep{Hampel1971}. We can use the connection with kernel scoring rules to study bias robustness of our estimators.
\begin{theorem}[Bias Robustness]\label{prop:bias_robustness_MMD}
The influence function corresponding to the maximum mean discrepancy is given by $\text{IF}_{\MMD}(z,\mathbb{P}_\theta) = g^{-1}(\theta) \nabla_{\theta} \MMD(\mathbb{P}_{\theta},\delta_{z})$. Furthermore, suppose that $\nabla_1 k$ is bounded and $\int_{\mathcal{U}} \|\nabla_\theta G_\theta(u)\| \mathbb{U}(\mathrm{d}u) < \infty$, then the MMD estimators are bias-robust.
\end{theorem}

Note that the conditions for this theorem to be valid are less stringent than assumptions required for the CLT in Theorem \ref{thm:asympt_normal2}. As we shall see in the next section, there will be a trade-off between efficiency and robustness as the kernel bandwidth is varied.  We shall demonstrate this through the influence function.

Overall, these results demonstrating the qualitative and bias robustness of minimum MMD estimators provides another strong motivation for their use. For complex generative model, it is common to be in the M-open setting; see for example all of the MMD GANs applications in machine learning where neural networks are used as models of images. Although it is not realistically expected that neural networks are good models for this, our robustness results can help explain the favourable experimental results observed in that case. Note that, to the best of our knowledge, the robustness of Wasserstein and Sinkhorn estimators has not been studied.


\section{The Importance of Kernel Selection: Gaussian Models} \label{sec:gaussian_model} 

As should be clear from the previous sections, the choice of kernel will strongly influence the characteristics of minimum MMD estimators, including (but not limited to) the efficiency of the estimators, their robustness to misspecification and the geometry of the loss function.  In this section, we highlight some of these consequences for two particular models: a location and scale model for a Gaussian distribution. These models are illustrative problems for which many quantities of interest (such as the asymptotic variance and influence function) can be computed in closed form, allowing for a detailed study of the properties of minimum MMD estimators.

\subsection{Kernel Selection in the Literature}

A number of approaches for kernel selection have been proposed in the literature, most based on radial basis kernels of the form $k(x,y;l) = r(\lVert x - y \rVert/l)$, for some function $r:\mathbb{R}\rightarrow \mathbb{R}_{\geq 0}$. We now highlight each of these approaches, and  later discuss the consequences of our theoretical results in the case of Gaussian location and scale models.

\cite{Dziugaite2015} proposed to set the lengthscale using the median heuristic proposed in \cite{Gretton2008} for two-sample testing with MMD, and hence picked their lengthscale to be $\sqrt{\text{median}(\|y_i-y_j\|_{2}^2/2)}$ where $\{y_j\}_{j=1}^m$ is the data. This heuristic has previously been demonstrated to lead to high power in the context of two-sample testing for location models in \cite{Ramdas2015,Reddi2015}. See also \cite{Garreau2017} for an extensive investigation. \cite{Li2015,Ren2016,Sutherland2017} have demonstrated empirically that a mixture of squared-exponential kernels yields good performance, i.e. a kernel of the form $k(x,y) = \sum_{s=1}^S \gamma_s k(x,y;l_s)$ where $\gamma_1,\ldots,\gamma_S \in \mathbb{R}_+$ and the lengthscales $l_1,\ldots,l_S>0$ are chosen to cover a wide range of bandwidth. The weights can either be fixed, or optimised; see \cite{Sutherland2017} for more details.  As the sum of characteristic kernels is characteristic (see \cite{Sriperumbudur2009}) this is a valid choice of kernel.

Another approach orginating from the use of MMD for hypothesis testing consists of studying the asymptotic distribution of the test statistics, and choose kernel parameters so as to maximise the power of the test. This was for example used in \citep{Sutherland2017}. A similar idea could be envisaged in our case: we could minimise the asymptotic variance of the CLT obtained in the previous section. Unfortunately, this will not be tractable in general since computing the asymptotic variance requires knowing the value of $\theta^*$, but an approximation could be obtained using the current estimate of the parameter. 

Finally, recent work \citep{Li2017} also proposed to include the problem of kernel selection in the objective function, leading to a minimax objective. This renders the optimisation problem delicate to deal with in practice \citep{Bottou2017}. The introduction of several constraints on the objective function have however allowed significant empirical success \citep{Arbel2018,Binkowski2018}. We do not consider this case, but it will be the subject of future work.


\subsection{Gaussian Location Model}

To focus ideas we shall focus on a Gaussian location model for a $d$-dimensional istropic Gaussian distribution $\mathcal{N}(\theta ,\sigma^2 I_{d \times d})$ with unknown mean $\theta \in \mathbb{R}^d$ and known standard deviation $\sigma>0$. In this case, we take $\mathcal{U}=\mathcal{X}=\mathbb{R}^d$, $\mathbb{U}$ is a standard Gaussian distribution $\mathcal{N}(0,\sigma^2 I_{d\times d})$ and $G_{\theta}(u) = u+\theta$. The derivative of the generator is given by $\nabla_{\theta} G_{\theta}(u) = I_{d\times d}$. Although this is of course a fairly simple model which could be estimated by MLE, it will be useful to illustrate some of the important points for the implementation of MMD estimators. In the first instance, we consider the M-closed case where the data consists of samples $\{y_j\}_{j=1}^m \iid \mathbb{Q}$ where $\mathbb{Q} = \mathbb{P}_{\theta^*}$ and the kernel is given by $k(x,y;l) = \phi(x;y,l^2)$, where $\phi(x;y,l^2)$ is the probability density function of a Gaussian $\mathcal{N}(y,l^2 I_{d \times d})$. 
\begin{proposition}[Asymptotic Variance for Gaussian Location Models]
Consider the minimum MMD estimator for the location $\theta$ of a Gaussian distribution $\mathcal{N}(\theta, \sigma^2 I_{d \times d})$ using a Gaussian kernel $k(x, y) = \phi(x; y, l^2)$, then the estimator $\hat{\theta}_m$ has asymptotic variance given by
\begin{align}
\label{eq:location_gaussian_av}
C & =  \sigma^2 ((l^2+\sigma^2)(3\sigma^2+l^2))^{-\frac{d}{2}-1}(l^2+2 \sigma^2)^{d+2}I_{d\times d}.
\end{align}
\end{proposition}
The Fisher information for this model is given by $1/\sigma^2 I_{d \times d}$, and so in the regime $l \rightarrow \infty$ we recover the efficiency of the MLE, so that the Cramer-Rao bound in Theorem \ref{prop:cramer_rao} is attained. On the other hand, for finite values of $l$, the minimum MMD estimator will be less efficient than the MLE. For $l \rightarrow \infty$, the asymptotic variance is $O(1)$ with respect to $d$, but we notice that the asymptotic variance is $O(\alpha^{d+2})$ as $l \rightarrow 0$, where $\alpha = 2/\sqrt{3} \approx 1.155 > 1$. This demonstrates a curse of dimensionality in this regime. This transition in behaviour suggests that there is a critical scaling of $l$ with respect to $d$ which results in asymptotic variance independent of dimension.   
\begin{proposition}[Critical Scaling for Gaussian Location Models]
\label{prop:critical_scaling_location}
Consider the minimum MMD estimator for the location $\theta$ of a Gaussian distribution $\mathcal{N}(\theta, \sigma^2 I_{d \times d})$ using a single Gaussian kernel $k(x,y) = \phi(x; y, l^2)$ where $l = d^\alpha$.  The asymptotic variance is bounded independently of dimension if and only $\alpha \geq 1/4$.
\end{proposition} 
As previously mentioned, it has been demonstrated empirically that choosing the bandwidth according to the median heuristic results in good performance in the context of MMD hypothesis tests \citep{Reddi2015,Ramdas2015}. These works note that the median heuristic yields $l = O(d^{1/2})$, which lies within the dimension independent regime in Proposition \ref{prop:critical_scaling_location}. Our CLT therefore explains some of the favourable properties of this choice.

Clearly, the choice of lengthscale can have a significant impact on the efficiency of the estimator, but it can also impact other aspects of the problem. For example, the loss landscape of the MMD, and hence our ability to perform gradient-based optimisation, is severely impacted by the choice of kernel. This is illustrated in Figure \ref{fig:MMDestimators_gaussian_model} (top left) in $d=1$, where choices of lengthscale between $5$ and $25$ will be preferable for gradient-based optimisation routines since they avoid large regions where the loss function will have a gradient close to zero. Using a mixture of kernels with a range of lengthscale regimes could help avoid regients of near-zero gradient and hence be generally desirable for the gradient-based optimization.   A third aspect of the inference scheme which is impacted by the lengthscale is the robustness. We can quantify the influence of the kernel choice on robustness using the influence function.  Similar plots for different classes of kernels can be found in the Appendix in Figures \ref{fig:gaussian_loss_landscape_appendix}, \ref{fig:gaussian_robustness_pollution_appendix} and \ref{fig:gaussian_threshold_robustness_appendix}.
\begin{proposition}[Influence Function for Gaussian Location Models]
\label{prop:gaussian_robust}
Consider the MMD estimator for the location $\theta$ of a Gaussian model $\mathcal{N}(\theta, \sigma^2 I_{d \times d})$ based on a Gaussian kernel $k(x, y) = \phi(x; y, l^2)$. Then the influence function is given by: 
\begin{align*}
\text{IF}_{\text{MMD}}(\mathbb{P}_\theta, z) 
& = 
  2  \left(\frac{l^2 +2 \sigma^2}{l^2+\sigma^2}\right)^{\frac{d}{2}+1} \exp\left(-\frac{\|z-\theta\|_{2}^2}{2(l^2+\sigma^2)}\right) (z-\theta) .
\end{align*}
\end{proposition}
Note that the asymptotic variance \eqref{eq:location_gaussian_av} is minimised by taking $l$ arbitrarily large. Despite this, in practice we do not want to choose $l$ to be larger than necessary as this will poorly influence the robustness of the estimator. Clearly, for the location model, we see that $l$ controls the sensitivity of our estimators. For every finite $l$, we have the following uniform bound for the influence function
\begin{align*}
\sup_{z\in \mathbb{R}^d}\left|\text{IF}_{\MMD}(\mathbb{P}_{\theta^*},z)\right| & = 2 e^{-1/2} \sqrt{l^2+\sigma ^2} \left(\frac{l^2+2 \sigma ^2}{l^2+\sigma ^2}\right)^{\frac{d}{2}+1}.
\end{align*}
Taking $l \rightarrow \infty$ we have $\text{IF}_{\MMD}(\mathbb{P}_{\theta^*},z) \rightarrow (\theta^* - z)$, thus losing robustness in the limit.  As with asymptotic variance, the sensitivity will depend exponentially on dimension when $l$ is small.  Contrary to intuition, the uniform influence function minimum will not be attained when $l$ approaches zero, but rather at an intermediate point, when $l^2=d\sigma^2$, after which the influence to contamination will increase as $l\rightarrow \infty$.  The middle plot in Figure \ref{fig:MMDestimators_gaussian_model} (top) illustrates the effect of kernel bandwidth on robustness.  The figure plots the $l_1$ error between the estimated parameter $\hat{\theta}_{m}$ (for $n$) based on a polluted data sample $\mathbb{Q}(\mathrm{d}x) = (1-\epsilon)\phi(x; 0, 1) \mathrm{d} x + \epsilon \delta z$, for some $z \in \mathbb{R}^d$ where $\epsilon = 0.2$.  While the estimator is qualitatively robust, for higher kernel bandwidths, the estimator will undergo increasingly large excursions from $\theta^* = 0$ as the position of the contaminent point $z$ moves to infinity.  The second plot demonstrates the behaviour of the estimators as the pollution strength $\epsilon$ is increased from $0$ to $1$ and $z=(10,\ldots, 10)^\top$.  We observe that for small kernel bandwidths, the estimator undergoes a rapid transition around $\epsilon = 0.5$.  However, as the lengthscale is increased the estimator becomes increasingly sensitive to distance sample points to the extent that the error grows linearly with respect to $\epsilon$. Interestingly, additional experiments presented in Figure \ref{fig:Sinkhorn_robustness} of the Appendix indicate that Wasserstein-based estimators may not be robust.

\begin{figure}[t!]
\begin{center}
\begin{minipage}{0.75\textwidth}
\includegraphics[width=0.32\textwidth,clip,trim = 0 0 2.8cm 0]{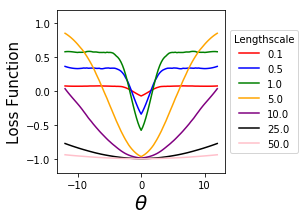}
\includegraphics[width=0.325\textwidth,clip,trim = 2.8cm 0 0 0]{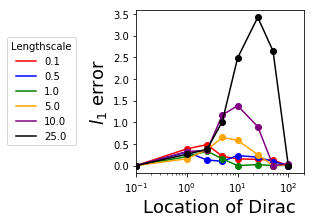}
\includegraphics[width=0.31\textwidth,clip,trim = 0 0 2.8cm 0]{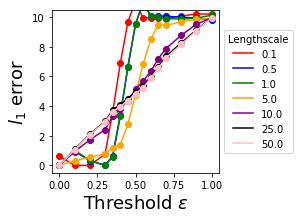}\\
\includegraphics[width=0.32\textwidth,clip,trim = 0 0 2.8cm 0]{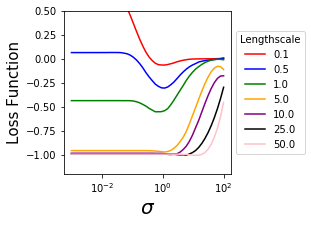}
\includegraphics[width=0.325\textwidth,clip,trim = 2.8cm 0 0 0]{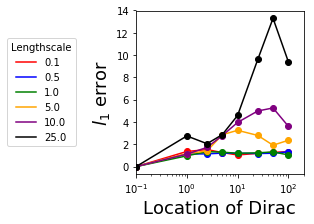}
\includegraphics[width=0.31\textwidth,clip,trim = 0 0 2.8cm 0]{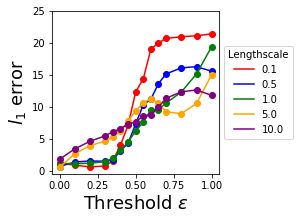}
\end{minipage}
\includegraphics[width=0.1\textwidth,clip,trim = 7.8cm 2cm 0 1cm]{Figures/MMD2_robust_varyingeps_vstats_d1_leg.png}
\end{center}
\caption{\textit{Gaussian location and scale models - Performance of the Gaussian RBF kernel (for $n,m$ large).} The top plots refer to the Gaussian location model, whilst the bottom plots refer to the Gaussian scale model. \textit{Left:} Comparison of the loss landscape for various lengthscale values in $d=1$. \textit{Center:} Robustness problem with varying location $x$ for the Dirac but threshold fixed to $\epsilon=0.2$ in $d=1$. \textit{Right:} Robustness problem with varying threshold but fixed location for the Dirac at $x=10$ in $d=1$.}
\label{fig:MMDestimators_gaussian_model}
\end{figure}


\subsection{Gaussian Scale Model}

The second model we consider is a $d$-dimensional isotropic Gaussian distribution $\mathcal{N}(\mu, \exp(2\theta)I_{d\times d})$ with known location parameter $\mu \in \mathbb{R}^d$. Since the asymptotic variance does not depend on any location parameters of $\mathbb{U}$, we will assume without loss of generality that the base measure $\mathbb{U}$ is a $d$-dimensional Gaussian with mean zero and identity covariance matrix, and that $G_{\theta}:\mathbb{R}^d \rightarrow \mathbb{R}^d$ is defined by $G_{\theta}(u) = \exp(\theta)u$. For simplicity, we assume that we are in the M-closed situation, where the true data distribution $\mathbb{Q}$ is given by $\mathbb{P}_{\theta^*}$ and the kernel is $k(x,y; l) = \phi(x; y, l^2)$. For this model, the 
conclusions in terms of efficiency, robustness and loss landscape are similar to those of the Gaussian location model. For example, we can once again compute the asymptotic variance of the CLT:
\begin{proposition}[Asymptotic Variance for Gaussian Scale Models]
Consider the minimum MMD estimator for the scale $\theta$ of a Gaussian distribution $\mathcal{N}(\mu, \exp(\theta)I_{d\times d})$ using a Gaussian RBF kernel $k(x,y; l) = \phi(x; y, l^2)$.  The asymptotic variance of the minimum MMD estimator $\hat{\theta}_m$ satisfies $C^l = g^{-1}(\theta^*)\Sigma g^{-1}(\theta^*)$ where the metric tensor at $\theta^*$ satisfies
\begin{align*}
  g(\theta^*) & = (2\pi)^{-d/2}\left(l^2 + 2e^{2\theta^*}\right)^{-d/2} d^2 K\left(d,l, e^{2\theta^*}\right),
\end{align*}
for a $K(d,l, s)$ is bounded with respect to $d, l,$ and $s$ and $K(d,0,s) = \frac{1}{4}(1+ 2d^{-1})$; and
\begin{align*}
  \Sigma = (2\pi)^{-d} d^2 e^{4\theta^*} \left(e^{\theta^*} + l^2\right)^{-2}\left(C_1 \left(l^2 + 3e^{2\theta^*}\right)^{-d/2}\left(l^2 + e^{2\theta^*}\right)^{-d/2} + C_2 \left(l^2+ 2e^{2\theta^*}\right)^{-d}\right),
\end{align*}
where $C_1$ and $C_2$ are bounded uniformly with respect to the parameters. Asymptotically, the asymptotic variance behaves as $C^l \sim  (2/\sqrt{3})^{d}/d^2$ for $l \ll 1$ and $C^l \sim l^4 / d^2$ for $l \gg 1$.
\end{proposition}
This result indicates that  the asymptotic variance grows exponentially with dimension as $l$ small.  In the other extreme, for $l$ going to infinity results in an asymptotic variance which is bounded independent of dimension. In fact, the following result characterises the choice of length-scale required for dimension independent efficiency.
\begin{proposition}[Critical Scaling for Gaussian Scale Models]
\label{prop:critical_scaling_scaling}
Consider the minimum MMD estimator for the scale $\theta$ of a Gaussian distribution $\mathcal{N}(\mu, \exp(\theta)I_{d\times d})$ with a single Gaussian kernel $k(x,y) = \phi(x; y, l^2)$ where $l = d^{\alpha}$. The asymptotic variance is bounded independently of dimension if and only if $\alpha \geq 1/4$.
\end{proposition}
This scaling is the same as for the Gaussian location model, indicating that a more general result on critical scaling for MMD estimators may exists. We reserve this issue for future work.  Once again, we notice (Figure \ref{fig:MMDestimators_gaussian_model}, bottom left) that the choice of lengthscale has a significant impact on the loss landscape. However, an interesting point is that values of the lengthscale which render the loss landscape easily amenable to gradient-based optimisation are different in the two cases. Numerical experiments clearly indicate that the choice of lengthscale may need to be adapted based on the parameters of interest.  The lengthscale has, once again, a significant impact on the robustness of the estimator, as demonstrated in the following result.

\begin{proposition}[Influence Function for Gaussian Scale Models]
Consider the minimum MMD estimator for the scale $\theta$ of a Gaussian model $\mathcal{N}(\mu, \exp(\theta)I_{d\times d})$ based on a single Gaussian kernel $k(x,y) = \phi(x; y, l^2)$. The influence function associated with this estimator is:
\begin{equation*}
\text{IF}_{\text{MMD}}(\mathbb{P}_\theta, z) =  \left(\frac{l^2+2e^{2\theta^*}}{l^2+e^{2\theta^*}}\right)^{\frac{d}{2}+2}
\frac{\left(l^2+e^{2\theta^*}-z^2\right) }{d (d+2) e^{2\theta^*}}  \exp\left(-\frac{z^2}{2 \left(l^2+e^{2\theta^*}\right)}\right).
\end{equation*}
\end{proposition}
In particular, for every finite $l$ we have that 
\begin{align*}
  \sup_{z\in \mathbb{R}^d}\left|\text{IF}_{\MMD}(\mathbb{P}_{{\theta}^*}, z)\right| & = \frac{4  e^{-3/2}\left(l^2+2 e^{2\theta^*}\right) }{d (d+2) e^{2\theta^*}}  \left(\frac{l^2+2e^{2\theta^*}}{l^2+e^{2\theta^*}}\right)^{\frac{d}{2}+1},
\end{align*}
independently of $z$, so that $l$ controls the sensitivity of the estimator.  Once again, we see exponential dependence on dimension for $l$ small, with the minimum uniform influence at an intermediate point, with the dependence increasing as $l\rightarrow \infty$.


\subsection{Using Mixtures of Gaussian Kernels}

In \cite{Li2015,Ren2016,Sutherland2017} it was observed  empirically that using mixtures of distributions offers advantageous performance compared to making single choices.  In particular, it circumvents issues arising from gradient descent due to \emph{vanishing gradients}, which can occur if the lengthscale of the kernel chosen to be too small, as can be seen in Figure \ref{fig:MMDestimators_gaussian_model}.  While multiple kernels offer advantage for gradient descent, we aim to understand where mixture kernels offer any advantages in terms of asymptotic efficiency and robustness.  Focusing on the Gaussian location model case, we have the following result.

\begin{figure}[t!]
\begin{center}
\includegraphics[width=0.35\textwidth,clip,trim = 0 0 0 0]{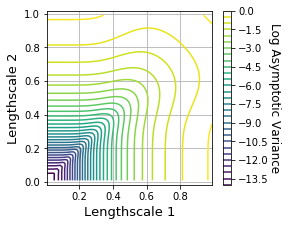} \hspace{5mm}
\includegraphics[width=0.345\textwidth,clip,trim = 0 0 0 0]{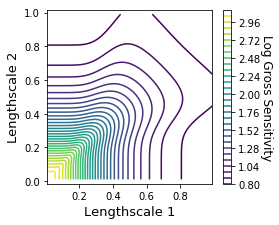}
\caption{\textit{Gaussian location model - Efficiency and Robustness for a Mixture of Kernels.} Asymptotic variance and gross sensitivity for minimum MMD estimators of the Gaussian location model as a function of $l_1$ and $l_2$ for a mixture of squared exponential kernels: $k(x,y) = \exp(-\|x-y\|_2^2/2l_1^2) + \exp(-\|x-y\|_2^2/2l_2^2)$.}
\label{fig:MMDestimators_gaussian_mixture}
\end{center}
\end{figure}

\begin{proposition}[Efficiency and Robustness  with Mixture Kernels for Gaussian Location]
Consider the minimum MMD estimator  for the location of a Gaussian distribution $\mathcal{N}(\theta,\sigma^2 I_{d \times d})$ using a Gaussian mixture kernel $k(x, y) = \sum_{s=1}^{S}\gamma_s\phi(x; y, l_s^2)$. Then the minimimum MMD estimator has asymptotic variance given by
\begin{equation}
\sigma^2 \frac{\sum_{s=1}^S \sum_{s'=1}^S \gamma_s \gamma_{s'} \left((l_s^2+\sigma^2)(l_{s'}^2+\sigma^2)+\sigma^2 (2\sigma^2+l_s^2+l_{s'}^2)\right)^{-\frac{d}{2}-1}}{\left(\sum_{s=1}^S \gamma_s (l_s^2+2 \sigma^2)^{-\frac{d}{2}-1}\right)^2}I_{d\times d}.
\end{equation}
Furthermore, the influence function is given by: 
\begin{align*}
\text{IF}_{\text{MMD}}(z,\mathbb{P}_\theta) 
& = 
 \frac{ 2 \sum_{s=1}^S \gamma_s  (l_s^2+\sigma^2)^{-\frac{d}{2}-1} \exp\left(-\frac{\|z-\theta\|_{2}^2}{2(l_s^2+\sigma^2)}\right)  (z-\theta) }{\sum_{s=1}^S \gamma_s (l_s^2 +2 \sigma^2)^{-\frac{d}{2}-1}}.
\end{align*}
\end{proposition}

In Figure \ref{fig:MMDestimators_gaussian_mixture} we plot the log asymptotic variance and log gross sensitivity for the Gaussian location model based on a mixture kernel composed of two Gaussian kernels with lengthscales $l_1$ and $l_2$.  What is interesting to note that there are choices of $(l_1, l_2)$ which give rise to higher efficiency and robustness than their individual counterparts.  Indeed, for example, choosing $l_1 = 0.8$ then the asymptotic variance will be minimised when $l_2 \approx 0.6$, although this choice will reduce bias-robustness. This figure appears to support the claim that mixture kernels can also provide increased performance beyond assisting gradient descent, and merits further investigation.


\section{Numerical Experiments} \label{sec:MMDestimators_experiments}

In this final section, we examine the impact of the choice of kernel on several applications. In particular, we highlight the importance of working with estimators which are robust to model misspecification. We start with two applications which are popular test-beds for inference for intractable generative models: the g-and-k distribution and a stochastic volatility model. We then move on to a problem of parameter inference for systems of stochastic differential equations, where we consider a parameter-prey model and a multiscale model. These examples allow us to demonstrate the advantage of our natural gradient descent algorithm, and the favourable robustness properties of the estimators.

\subsection{G-and-k distribution}

A common synthetic model in the literature on generative models is the g-and-k distribution \citep{Bernton2019,Prangle2017}. For this model, we only have access to the quantile function $G_{\theta}:[0,1] \rightarrow \mathbb{R}$ (also called inverse cumulative distribution function) given by:
\begin{align*}
G_\theta(u) & :=  a + b \left( 1 + 0.8 \frac{\big( 1 - \exp(- c \Phi^{-1}(u;0,1)\big)}{\big( 1 + \exp(-c \Phi^{-1}(u;0,1)\big)}\right) \big(1+(\Phi^{-1}(u;0,1))^2 \big)^{k} \Phi^{-1}(u;0,1)
\end{align*}
 where $\Phi^{-1}(u;0,1)$ refers to the $u$'th quantile of the standard normal distribution. The parameter of interest is $\theta = (\theta_1,\theta_2,\theta_3,\theta_4)$ where $\theta_1=a$ controls location, $\theta_2=b$ controls scale, $\theta_3=c$ controls skewness and $\theta_4 = \exp(k)$ controls kurtosis. Although this is a model defined on a one-dimensional space, the four parameters allow for a very flexible family of distributions. A rescaling of the last parameter is used to avoid instabilities. Since the quantile function is available, we can easily simulate from this model using inverse transform sampling. 

 We study the behaviour of the MMD estimators for this model in Figure \ref{fig:gandk_results}. For the left and center plots, we used both stochastic gradient descent and stochastic gradient descent with preconditioner to obtain an estimate of $\hat{\theta}_{m}$. We used a constant step-size for both algorithms (tuned for good performance) and ran each algorithm for $500$ iterations. The data is of size $m=30000$ but we used minibatches of size $200$, the simulated data was of size $n=200$, the kernel was Gaussian RBF with lengthscale $l=2$ and $\theta^* = (3,1,1,-\log(2))$. The large number of data points is used to guarantee that the minimiser can be recovered. We notice that both the stochastic gradient descent and stochastic natural gradient descent are able to recover $\theta_1^*$ and $\theta_2^*$ for a variety of initial conditions in the neighbourhood of the minimiser. On the other hand, as observed in the center plot, the stochastic gradient descent algorithm is very slow for $\theta_3^*$ and $\theta_4^*$, whereas the natural stochastic gradient descent algorithm is able to recover both of these parameters in a small number of steps. This clearly highlights the advantage of the rescaling of the parameter space provided by the preconditionner based on the geometry induced by the information metric.

\begin{figure}[t!]
\begin{center}
\includegraphics[width=0.32\textwidth]{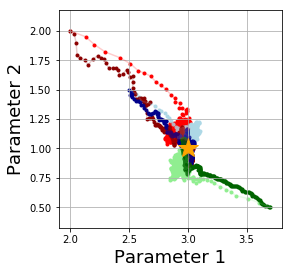}
\includegraphics[width=0.33\textwidth]{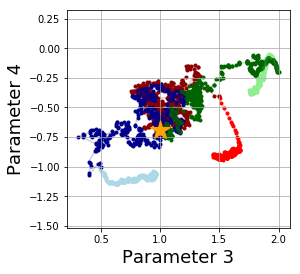}
\includegraphics[width=0.32\textwidth]{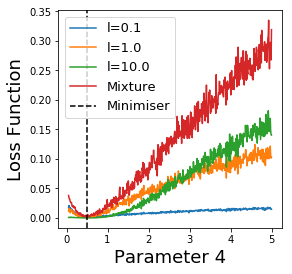}
\caption{\textit{Inference for the parameters of the g-and-k distribution using a maximum mean discrepancy estimator.} \textit{Left \& Center:} Several runs of a stochastic gradient descent (light blue, light red and light green) and a stochastic natural gradient descent (dark blue, dark red and dark green) algorithm on the MMD loss function with Gaussian RBF kernel with lengthscale $l=2$. The black dot corresponds to the minimiser. \textit{Right:} Estimate of the MMD loss function around the minimum as a function of $\theta_4$ for a Gaussian RBF kernel with varying choices of lengthscales and a mixture of all the Gaussian kernels.}
\label{fig:gandk_results}
\end{center}
\end{figure}

We also plot the MMD loss function as a function of $\theta_4$ in the neighborhood of $\theta^*$. The estimator is sensitive to the choice of lengthscale: in the case where we use a Gaussian RBF kernel, a lengthscale smaller of equal to $l=0.1$ or greater or equal to $l=10$ led to a loss function which is flat on a wide range of the space. In those cases, it will be difficult to obtain an accurate estimate of the parameter due to the noise in our estimates of the gradient. On the other hand, a lengthscale of $l=1$ allows us to provide more accurate results. Furthermore, the use of a mixture of all of these kernels also allows us to obtain accurate results without having to manually tune the choice of lenghtscale.
 

\subsection{Stochastic Volatility Model with Gaussian and Cauchy Noise}

Our second model is a stochastic volatility model \citep{Kim1998}, popular in the econometrics literature as a model of the returns on assets over time.
The model can be simulated from by sampling the first hidden variable $h_1 \sim \mathcal{N}(0, \sigma^2/(1-\phi^2))$ representing the initial volatility, then following the following set of equations:
\begin{align*}
h_t & =  \phi h_{t-1} + \eta_t, \qquad \eta_t \; \sim \; \mathcal{N}(0,\sigma^2), \\
y_t & =  \epsilon_t \kappa \exp(0.5 h_t), \qquad \epsilon_t \; \sim \; \mathcal{N}(0,1).
\end{align*}
where $y_t$ is the mean corrected return on holding an asset at time $t$, and $h_t$ the log-volatility at time $t$. The $\{y_t\}_{t=1}^T$ are observed data and $\{h_t\}_{t=1}^T$ are unobserved latent variables. This is therefore a generative model with parameters $(\phi,\kappa,\sigma)$, which we reparameterised with $\theta_1 = \log((1+\phi)/(1-\phi)), \theta_2 = \log \kappa, \theta_3 = \log(\sigma^2)$ to avoid numerical issues so that we want to recover $\theta = (\theta_1,\theta_2,\theta_3)$. The data dimension is $d = T$ and the parametric dimension is $p=3$.  The likelihood of these models is usually not available in closed form due to the presence of latent variables and hence given by $p(y_1,\ldots,y_T|\theta) = \int p(y_1,\ldots,y_T|h_1,\ldots,h_T,\theta) p(h_1,\ldots,h_T|\theta) \mathrm{d}h_1 \ldots \mathrm{d}h_T$ which is a high-dimensional intractable integral. Alternative approaches based on quasi-likelihood estimation or expectation-maximisation can be considered, but the approximation obtained may be unreliable. Furthermore, it may be preferable to make use of minimum MMD estimators since these will allow for robust inference, which is not the case for alternative approaches. 

In our experiments, we choose $T=30$ and considered inference with minimum MMD estimators with Gaussian kernels. Initially, we considered the M-closed case and generated $m=20000$ data points for $\theta^* = (0.98,0.65,0.15)$, which we then tried to infer by minimising the MMD loss function with a wide range of kernels. For the experimental results, we used stochastic gradient descent and stochastic natural gradient descent with minibatches of size $2000$, and used $n=45$. Results in Figure \ref{fig:stochastic_volatility} (top) demonstrate that our natural gradient algorithm is able to recover the parameters in around five thousand iterations whereas the gradient descent algorithm isn't close to convergence after $30000$ steps. Note that even though the dimension $d=30$, the parameter space has dimension $p=3$ so that the additional computational cost of the preconditioner is negligeable for this problem (and completely dwarfed by the cost of the generator).

\begin{figure}[t!]
\begin{center}
\includegraphics[width=0.32\textwidth]{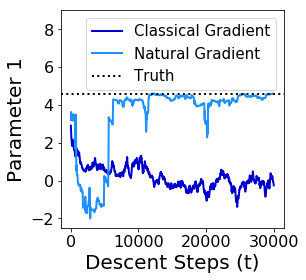}
\includegraphics[width=0.32\textwidth]{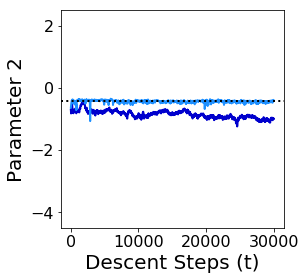}
\includegraphics[width=0.32\textwidth]{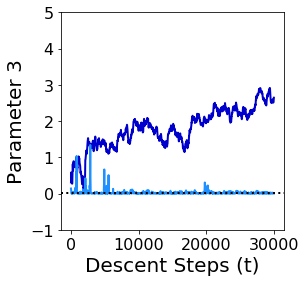}\\
\includegraphics[width=0.32\textwidth]{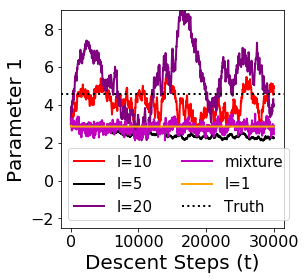}
\includegraphics[width=0.32\textwidth]{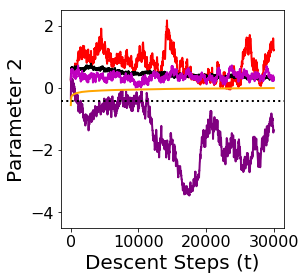}
\includegraphics[width=0.32\textwidth]{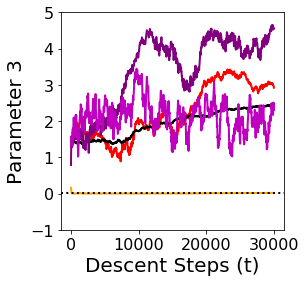}\\
\includegraphics[width=0.33\textwidth]{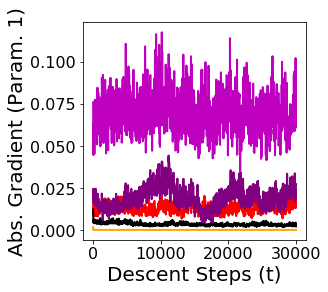}
\includegraphics[width=0.32\textwidth]{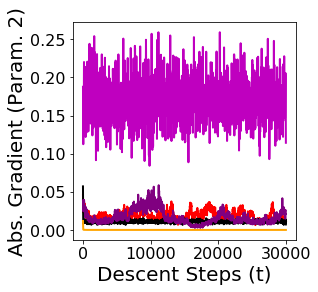}
\includegraphics[width=0.33\textwidth]{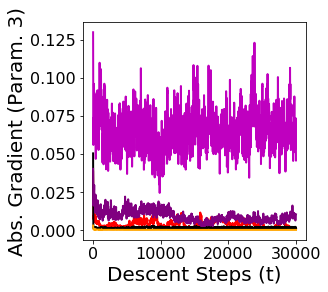}
\end{center}
\caption{\textit{Inference for stochastic volatility models.} \textit{Top:} Well-specified case - Gradient descent and natural gradient descent on the MMD loss function with a mixture of Gaussian RBF kernels with lengthscales $1,5,10,20,40$. \textit{Bottom:} Misspecified case - Gradient descent on the MMD loss function with a variety of kernels including a Gaussian RBF kernel with lengthscales $1$, $5$, $10$, $20$, as well as a mixture of all these kernels and a Gaussian RBF kernel with lengthscale $l=40$.}
\label{fig:stochastic_volatility}

\end{figure}

We then considered the M-open case, and introduced misspecification by simulating the $\epsilon_t$ values using IID realisations of a Cauchy distribution with location parameter $0$ and scale parameter $\sqrt{2/\pi}$. This distribution has the same median as the Gaussian distribution, and their probability density functions match at that point, but the Cauchy has much fatter tails.  The results of these experiments are available in Figure \ref{fig:stochastic_volatility}, and in each case we repeated the experiments with $4$ different stepsize choices and plot the best result. In the well-specified case, we notice that the natural gradient descent algorithm is able to take advantage of the local geometry of the problem and converges to $\theta^*$ in a small number of iterations. Further experiments with a larger range of kernels is available in Appendix \ref{sec:stochvol_appendix}, but the mixture kernel tended to work best.

In the misspecified case, we notice (as expected) that while none of the minimum MMD estimators is able to recover the true value of $\theta^*$, but that the inferred results remain stable, i.e. close to the truth.  The choice of kernel has a clear impact on the output. For Gaussian RBF kernels with lengthscales $l=1$ or $l=5$, the loss function is too flat for gradient descent and we are not able to move much from the initial parameter. For larger values of the lengthscale (e.g. $l=10,20,40$) and for the mixture kernel, we are able to use gradient-based optimisation but that it demonstrates increased sensitivity to model mispecification.  We note that the single Gaussian RBF kernels with large lengthscale are able to learn $\theta_1^*$ well in the sense that there is negligeable bias as compared to $\theta_2^*$ and $\theta_3^*$. This is likely due to the improvement in bias robustness expected for kernels with large lengthscale.


\subsection{Inference for Systems of Stochastic Differential Equations} 

For our third set of experiments, we use minimum MMD estimation to infer the initial condition and parameters for coupled systems of stochastic differential equations (SDEs). In general, will will consider a $d$-dimensional It\^{o} stochastic differential equation of the form
\begin{align}
\label{eq:sde1}
	dX_t = b(X_t; \theta_1)\,dt + \sigma(X_t; \theta_1)\,dW_t,
\end{align}
where $b:\mathbb{R}^d\times \Theta\rightarrow \mathbb{R}^d$, $\sigma:\mathbb{R}^d\times \Theta \rightarrow \mathbb{R}^{d\times k}$,  $W_t$ is a $k$ dimensional standard Brownian motion and with initial value $X_0 = \theta_2$  and where $(\theta_1, \theta_2) \in \Theta$ is a vector of unknown parameters to be determined. We assume that for each $\theta$ there is a unique solution to \eqref{eq:sde1} which depends continuously on the initial condition.

For any fixed $\theta$, provided we can simulate $X(t)$, at points $0= t_0 < t_1 < \ldots <  t_K = T$, then we can consider the generative model defined by $\mathbb{P}_{\theta} = G_{\theta}^{\#} \mathbb{U}$, where $\mathbb{U}$ is the Wiener path  measure for a $k$ dimensional standard Wiener process on $C[0,T]$ and $G_{\theta} = O_{\theta}\circ I_{\theta}$,  where $I_{\theta}:C[0,T]\rightarrow C[0,T]$ is the It\^{o} map, transforming the Wiener process to the solution $X(\cdot)$ of the SDE. Here, $O_{\theta}$ is an observation operator, for example mapping $w \in C[0,T]$ to $(w(t_1), \ldots, w(t_K))$ or any other smooth functional of the path which depends smoothly on $\theta$.  Note that it is trivial to incorporate observational noise and volatility parameters into the observation operator.

To perform MMD gradient descent for this model we must calculate the gradient of the forward map with respect to the parameters $\theta$.  Pathwise derivatives of the solution of \eqref{eq:sde1} with respect to initial conditions and coefficient parameters are well established \citep{kunita1997stochastic,gobet2005sensitivity,friedman2012stochastic} and are detailed in the following result; see also \cite{tzen2019neural} for a similar result arising in a similar context.

\begin{proposition}{\cite[Theorem 2.3.1]{kunita1997stochastic}}
\label{prop:pathwise_gradients}
Suppose that the drift $b(x; \theta_1)$ and diffusion tensor $\sigma(x; \theta_1)$ are Lipschitz with Lipschitz derivatives with respect to $x$ and $\theta_1$.  Then the pathwise derivative of $X_t$ with respect to the parameters $\theta_1$ is given by the solution of the Ito process,
\begin{align*}
d\left(\nabla_{\theta_1} X_t\right) = &\left(\nabla_x b(X_t; \theta_1)\nabla_{\theta_1} X_t + \nabla_{\theta_1}b(X_t; \theta_1)\right)\,dt  + \left(\nabla_x \sigma(X_t; \theta)\nabla_{\theta_1} X_t + \nabla_{\theta_1}\sigma(X_t; \theta_1)\right)\,dW_t
\end{align*}
with initial condition $\nabla_{\theta_1} X_0 = \mathbf{0}$ and the derivative of $X_t$ with respect to $\theta_2$ is given by 
\begin{align*}
d\left(\nabla_{\theta_2} X_t\right) = &\left(\nabla_x b(X_t; \theta_1)\nabla_{\theta_2} X_t \right)\,dt  + \left(\nabla_x \sigma(X_t; \theta)\nabla_{\theta_2} X_t \right)\,dW_t,
\end{align*}
where $\nabla_{\theta_2} X_0 = I$.  In particular, given a differentiable function $F$ of $X_{t_0}, \ldots, X_{t_K}$,
\begin{align*}
\nabla_{\theta_i}\mathbb{E}\left[F(X_{t_0}, \ldots, X_{t_K})\right] & =\mathbb{E}\left[\sum_{k=1}^{K}\nabla_{x_k}F(X_{t_0}, \ldots, X_{t_K}) \nabla_{\theta_i} X_{t_k}\right], \quad \text{ for } i=1,2.
\end{align*}
\end{proposition}
Before moving on to the experiments, we note that \cite{Abbati2019} proposed an alternative method, performing Gaussian process-based gradient matching for ODEs and SDEs with additive noise, by using MMD to fit a GP process inferred from the data to the SDE. However, the approach we propose permits parametric estimation for more general SDEs and noise models.


\subsubsection{Noisy Lotka-Volterra Model with Unknown Initial Conditions}

As an example, we consider the stochastic Lotka-Volterra model \citep{Volterra1926}, which consists of a pair of nonlinear differential equations  describing the evolution of two species through time:
\begin{align*}
\mathrm{d}\left(\begin{array}{c}X_{1,t} \\ X_{2,t}\end{array}\right) & = \left[\left(\begin{array}{c}1 \\ 0\end{array}\right)\theta_{11} X_{1,t} + \left(\begin{array}{c}-1 \\ 1\end{array}\right)\theta_{12} X_{1,t} X_{2,t} +  \left(\begin{array}{c}0 \\ -1\end{array}\right)\theta_{13} X_{2,t} \right]\mathrm{d}t \\
\quad &+ \left(\begin{array}{c}1 \\ 0\end{array}\right)\sqrt{\theta_{11} X_{1,t}} \mathrm{d} W_t^{(1)} + \left(\begin{array}{c}-1 \\ 1\end{array}\right)\sqrt{\theta_{12} X_{1,t} X_{2,t}}\mathrm{d}W_t^{(2)} +  \left(\begin{array}{c}0 \\ -1\end{array}\right)\sqrt{\theta_{13} X_{2,t}}\mathrm{d} W_t^{(3)},
\end{align*}
where the initial conditions $\theta_2 = (X_{1,0},X_{2,0})$ are unknown, but the parameters $\theta_1 = (\theta_{11},\theta_{12},\theta_{13})$ governing the dynamics are known.  While exact sampling methods for diffusions exist, see \cite{beskos2005exact}, for simplicity we shall employ an inexact Euler-Maruyama discretisation, choosing the step size sufficiently small to ensure stability of the discretisation.  We choose the ``true'' initial condition to be deterministic with value $\theta_2^* = (X_{1,0},X_{2,0})$. We fix a-priori the time horizon to $T=1$ and the parameters governing the equation to $\theta_1 = (\theta_{11},\theta_{12},\theta_{13}) = (5,0.025,6)$. In this case, $p=2$, $d=2$ and $n$ tends to be small (in the tens or hundreds). We consider the case where $n=50$.

\begin{figure}[t!]
\begin{center}
\includegraphics[width=0.28\textwidth]{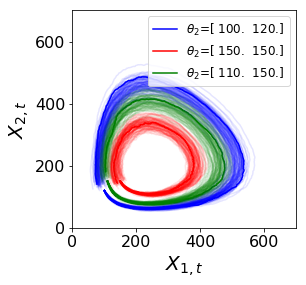}
\includegraphics[width=0.31\textwidth]{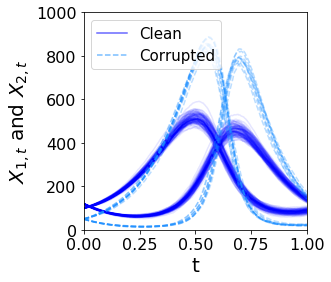}
\caption{\textit{Inference for the initial conditions of a Lotka-Volterra model with noisy dynamics.} \textit{Left:} $n=100$ realisations from the coupled stochastic differential equations for several initial conditions. \textit{Right:} $n=100$ realisations used for inference, including $90$ realisations from the correct model and $10$ which are corrupted.}
\label{fig:LotkaVolterraModel_results}
\end{center}
\end{figure}

Typical realisations for the system of coupled stochastic differential equations can be found in Figure \ref{fig:LotkaVolterraModel_results} (left) for several values of the initial conditions. As we would expect, the closer the initial conditions, the closer the realisations of stochastic differential equations will be. This clearly motivates the use of minimum MMD estimators. We are particularly interested interested in the behaviour of the estimators as a proportion of the data is corrupted. In particular, we will consider the problem of inferring initial conditions $\theta^*_2=(100,120)$ given realisations from this model which are corrupted by realisations from the model initialised at $\theta^\dagger_2 = (50,50)$. Realisations are provided in Figure \ref{fig:LotkaVolterraModel_results} (right) for the case with $10\%$ misspecification. 

We expect this type of misspecification to lead to severe issues for non-robust inference algorithms, but the bias robustness of minimum MMD estimators allows us to provide reasonable estimates of the parameter. This can be seen in Figure \ref{fig:LotkaVolterraModel_robustness} (left) where we plot estimates provided by MMD estimators for $\theta_{22}$ as a function of natural gradient steps for various proportion levels of corruption. This is compared to the Sinkhorn algorithm of \cite{Genevay2017}. As can be seen, the MMD estimator can recover the truth for a large proportion of corrupted samples whereas Wasserstein-based estimators are very sensitive to corrupted data.


\begin{figure}[t!]
\begin{center}
\includegraphics[width=0.41\textwidth]{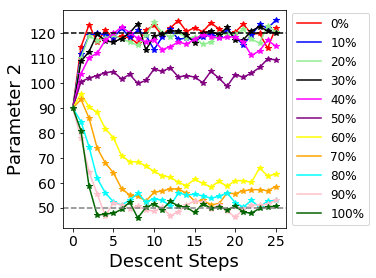}
\includegraphics[width=0.32\textwidth]{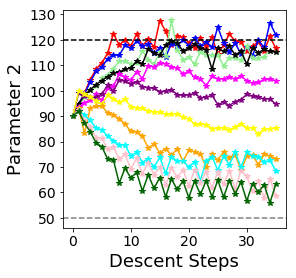}
\caption{\textit{Inference for the initial conditions of a Lotka-Volterra model with corrupted observations based on $m=100$ realisations and $n=50$ simulated data.} Each color correspond to a different percentage of corrupted observations. \textit{Left:} Stochastic gradient descent steps for minimum Sinkhorn estimator with $l_2$ cost and $\epsilon=1$ regularisation. \textit{Right:} Stochastic gradient descent steps for minimum MMD estimator with Gaussian RBF kernel and lengthscale $l=30$.}
\label{fig:LotkaVolterraModel_robustness}
\end{center}
\end{figure}


\subsubsection{Parametric Inference for a System of SDEs with Multiple Scales}

We consider a second example where we observe realisations of the following two-dimensional multiscale system
\begin{align}
\label{eq:multiscale_sde}
dX^\epsilon_t &= \left(\frac{\sqrt{\theta_{12}}}{\epsilon} Y^\epsilon_t + \theta_{11}X^\epsilon_t\right)\,dt, \qquad
dY^\epsilon_t = -\frac{1}{\epsilon^2}Y^\epsilon_t \,dt + \frac{\sqrt{2}}{\epsilon} dW_t,
\end{align}
where $W_t$ is a standard Brownian motion,  $0 < \epsilon \ll 1$ is a small length-scale parameter, $\theta_1 = (\theta_{11},\theta_{12})$ are unknown parameters governing the dynamics, and the initial conditions $\theta_2$ are known.  Such systems arise naturally in atmosphere/ocean science \citep{majda2001mathematical}, materials science \citep{weinan2011principles} and biology \citep{erban2006equation}, and the inference of such stochastic multiscale systems has been widely studied, see \citep{pavliotis2007parameter,krumscheid2018perturbation}.  

The process $Y_t^{\epsilon}$ is an Ornstein-Uhlenbeck process with vanishing autocorrelation controlled by $\epsilon$.  Formally, in the limit of $\epsilon\rightarrow 0$ it will behave as the derivative of Brownian motion.  One can formulate minimum MMD problem for estimating the parameters $\theta_{11}$ and $\theta_{12}$, appealing to Proposition \ref{prop:pathwise_gradients} to compute the MMD gradient.  However, a direct approach which involves integrating the SDEs in \eqref{eq:multiscale_sde} multiple times is computationally infeasible, due to the fact that the simulation step-size would need to be commensurate to the small scale parameter $\epsilon$. This motivates us to use a coarse grained model for estimating the unknown parameters. As $\epsilon \rightarrow 0$, the process $X_\cdot^{\epsilon}$ will converge weakly in $C[0,T]$ to a process $\overline{X}_\cdot$, given by the solution of the It\^{o} SDE:
\begin{align}
\label{eq:coarse_grained_sde}
d\overline{X}_t & = \theta_{11}\overline{X}_t + \sqrt{2\theta_{12}}\,dW_t,\quad t \in [0,T],
\end{align}
see \cite[Chapter 11]{pavliotis2008multiscale}. As the coefficients of this SDE do not depend on the small scale parameter, we are able to generate realisations far more efficiently than for \eqref{eq:multiscale_sde}.  We consider the minimum MMD estimator for $\theta_{11}$ and $\theta_{12}$ using \eqref{eq:coarse_grained_sde} as a model. This introduced model misspecification of an interesting nature:  for $\epsilon$ small, the path measures associated with \eqref{eq:multiscale_sde} and \eqref{eq:coarse_grained_sde} on $C[0,T]$ will be close with respect to the Levy-Prokhorov metric (which metrizes weak convergence) but not  with respect to stronger divergences such as total variation or KL divergence.  Indeed, the KL divergence between both measures will diverge as $\epsilon\rightarrow 0$.  As MMD induces a coarser topology than the Levy-Prokhorov metric, we expect that the MMD estimators will be robust with respect to this misspecification for $\epsilon$ small, whereas maximum likelihood estimators are known to be biased in this case \citep{pavliotis2007parameter}.

Suppose that we observe $100$ realisations of \eqref{eq:multiscale_sde} at discrete times $0.1, 0.2, \ldots, 1.0$ over a time horizon of $T=1$ with known initial conditions $\theta_2 = (1.0, 0.0)$ with true values of the parameters given by $\theta^{*}_1 = (-1/2, \sqrt{1/2})$.  We construct a minimum MMD estimator for $\theta_1$ using the coarse grained SDEs as a model.  In this case, $p=2$, $d=1$. To simulate the coarse-grained model, we use an Euler-Maruyama discretisation with a step-size of $10^{-2}$. We use natural gradient descent to minimise MMD, generating $n=100$ synthetic realisations of the coarse SDE \eqref{eq:coarse_grained_sde} per gradient step. In Figure \ref{fig:CoarseGraining_results} we plot the natural gradient descent trajectory for the estimators of $\theta_1$ for $\epsilon = 1, 0.5, 0.1$, respectively.  For $\epsilon=1$, where we anticipate the misspecification to be high, the minimum MMD estimator converges to the true value of $\theta_{11}$, but fails to recover the $\theta_{12}$ parameter (though remains within an order of magnitude).   Taking $\epsilon$ smaller we observe that the accuracy of the estimators increases, indicating that the MMD estimators capture the weak convergence of $\lbrace X^\epsilon_t \, , t\in [0,T]\rbrace$ to $\lbrace \overline{X}_t \, , \, t \in [0,T] \rbrace$.  We also note however that the volatility in the estimator for parameter $\theta_{11}$ is increasing as $\epsilon$ decreases, which suggests that the size of the simulated data (and perhaps also the size of the minibatches) must be increased as $\epsilon$ goes to $0$ to maintain a constant mean square error.
 
\begin{figure}[t!]
\begin{center}
\includegraphics[width=0.3\textwidth]{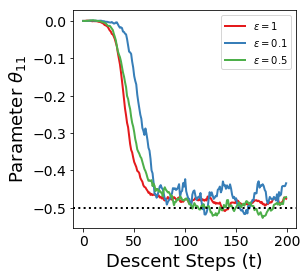}
\includegraphics[width=0.29\textwidth]{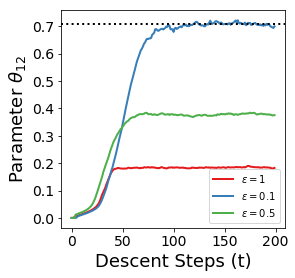}
\caption{\textit{Inference for the parameters of a two-scale stochastic process using a coarse grained model.} The plots show the convergence of the estimators to the truth values (dashed-lines) as the number of gradient descent steps increase, for data coming from \eqref{eq:multiscale_sde}.}
\label{fig:CoarseGraining_results}
\end{center}
\end{figure}



\section{Conclusion}

This paper studied a class of statistical estimators for models for which the likelihood is unknown, but for which we can simulate realisations given parameter values. Our estimators are based on minimising U-statistic approximations of the maximum mean discrepancy squared. We provided several results on their asymptotic properties and robustness, as well as a novel natural-gradient descent algorithm for efficient implementation. As demonstrated first in our theory, then later in the experiments, the choice of reproducing kernel allows for great flexibility and can help us trade-off statistical efficiency with robustness. 

This methodology clearly provides a rigorous approach to parametric estimation of complex black-box models for which we can only evaluate the forward map and its gradient. 
The natural robustness properties of these estimators make them a clear candidate for fitting models to engineering systems which are subject to intermittent sensor failures. Our theory also provides insights into the behaviour of MMD estimators for neural networks such as MMD GANs.

There are several directions in which this work could be extended.  Firstly, we note this methodology can be readily applied to other continuum models such as ordinary differential equations and (stochastic) partial differential equations with noisy parameters.  In these cases, adjoint based methods can be exploited to reduce the cost of computing gradients.  

A second direction which is promising relates to model reduction or \emph{coarse graining}, where a complex, very expensive model is replaced by a series of smaller models which are far cheaper to simulate.  We believe that minimum MMD based estimators are an excellent candidate for effecting these coarse graining approaches thanks to their robustness properties.  

Finally, we note that a drawback of this methodology is the poor scaling as a function of data-size.  Indeed, the cost of computing MMD grows quadratically with data-size.  This clearly motivates a second direction of research involving the use of cheaper approximate estimators for maximum mean discrepancy, such as \citep{Chwialkowski2015}.

\section{Acknowledgements} 
The authors are grateful to Tamara Broderick, Arthur Gretton and Dougal Sutherland for helpful discussions. FXB was supported by the EPSRC grants [EP/L016710/1, EP/R018413/1]. 
AB was supported by a Roth Scholarship funded by the Department of Mathematics of Imperial College London. AD and MG were supported by the Lloyds Register Foundation Programme on Data-Centric Engineering. MG was supported by the EPSRC grants [EP/J016934/3, EP/K034154/1, EP/P020720/1, EP/R018413/1], an EPSRC Established Career Fellowship, the EU grant [EU/259348]. All four authors were also supported by The Alan Turing Institute under the EPSRC grant [EP/N510129/1].

{
\setlength{\bibsep}{1pt}
\bibliographystyle{plainnat}
\bibliography{MMD_paper.bib}
}

\appendix

\clearpage
\newpage

{\LARGE
\begin{center}
\textbf{Supplementary Material for ``Statistical Inference for Generative Models with Maximum Mean Discrepancy''}
\end{center}
}
\vspace{3mm}

The supplementary materials are structured as follows. Section \ref{appendix:geometry_MMD} provides further discussion on the geometry induced by MMD on parametric families of probability distributions, and in particular derives the corresponding metric tensor, gradient flow and geodesics. Section \ref{appendix:proofs} contains all the proofs of results in the paper, including asymptotic results and results on robustness. Section \ref{appendix:gaussian_models} contains the derivation of important quantities for the Gaussian models. Finally, Section \ref{appendix:experiments} contains further details on the numerical experiments.


\section{Geometry of the MMD Statistical Manifold} \label{appendix:geometry_MMD}

In this appendix we complement Section \ref{sec:minimum_distance} and provide additional details on the Riemmanian manifold induced by the MMD metric.

\subsection{Identification of the Information Metric Tensor}\label{MMD2_metrictensor}

Identifying $\mathbb{P}_\theta$ as the pushforward $G_{\theta}^{\#} \mathbb U$, we have:
\begin{align*}
{\MMD}^2(\mathbb{P}_\alpha || \mathbb{P}_\beta)  & = 
\int_{\mathcal{U}} \int_{\mathcal{U}} k(G_{\alpha}(u),G_{\alpha}(v)) \mathbb{U}(\mathrm{d}u) \mathbb{U}(\mathrm{d}v) 
  -2 \int_{\mathcal{U}} \int_{\mathcal{U}} k(G_{\alpha}(u),G_{\beta}(v)) \mathbb{U}(\mathrm{d}u) \mathbb{U}(\mathrm{d}v) \\ 
  &  + \int_{\mathcal{U}} \int_{\mathcal{U}} k(G_{\beta}(u),G_{\beta}(v)) \mathbb{U}(\mathrm{d}u) \mathbb{U}(\mathrm{d}v)
\end{align*}
Taking the derivative with respect to $\alpha$ and $\beta$, and noticing that:
\begin{align*}
 \partial_{\beta^k} \partial_{\alpha^j} k(G_{\alpha}(u),G_{\beta}(v)) 
 & =  \sum_{l,i} \partial_{2^l} \partial_{1^i}k(G_{\alpha}(u),G_{\beta}(v)) \partial_{\alpha^j} G_{\alpha}^i(u) \partial_{\beta^k} G_{\beta}^l(v) \\
  & =  \big( \nabla_{\alpha} G_{\alpha}(u)^\top  \nabla_2 \nabla_1 k(G_{\alpha}(u),G_{\beta}(v))  \nabla_{\beta} G_{\beta}(v)  \big)_{jk}
\end{align*}

which yields the expression for the information metric associated to the $\text{MMD}^2$ divergence. 
Let $\mathcal H$ be a Hilbert space viewed as a Hilbert manifold. 
As usual we identify the tangent spaces $T_p \mathcal H \cong \mathcal H$, and the Riemannian metric is $m(f,g) = \langle f,g \rangle$ for any $f,g \in \mathcal H$. 
Let $\Psi: S \to \mathcal H$ be a differentiable injective immersion (i.e., its derivative is injective), from a finite-dimensional manifold $S$.
Then $\Psi$ induces a Riemannian structure on $S$ given by the pull-back Riemannian metric $g=\Psi^*m$.
If $x^i$ are local coordinates on $S$, and $\partial_{x^i}$ is the associated local basis of vector fields, then the components of $g$ are defined by 
\begin{align*}
g_{ij} = g(\partial_{x^i},\partial_{x^j}) = 
m\big( d\Psi(\partial_{x^i}),d\Psi(\partial_{x^j}) \big),
\end{align*}
where $d\Psi:TS \to T\mathcal H$ is the differential/tangent map (here $TS$ denotes the tangent bundle of $S$, or, roughly, the set of vectors tangent to $S$).
When $S$ is an open subset of $\mathbb{R}^n$,
since the Frechet partial derivative $\nabla_{x^j} \Psi(x)$ is the derivative of the function 
$t\mapsto \Psi(x^1,\ldots, x^{j-1},t,x^{j+1},\ldots, x^n)$, 
of the curve $t\mapsto (x^1,\ldots, x^{j-1},t,x^{j+1},\ldots, x^n)$ is precisely the curve tangent to the vector $\partial_{x^j}|_x$, we have 
$\nabla_{x^j} \Psi = d\Psi(\partial_{x^j}) $ (see \citep{lang2012fundamentals} page 28).
Hence $ g_{ij} = m\big( \nabla_{x^i} \Psi, \nabla_{x^j} \Psi\big)$.

Note that if $\Psi$ is not an immersion, the pullback Riemannian metric will in general just be a degenerate quadratic form rather than a positive definite one.

Let $S$ be a statistical manifold, i.e., $x\in S$ is associated to a probability measure $P_x$ (we assume the map $x \mapsto P_x$ is a bijection).
We can define a divergence on $S$ by 
$ D\big(P_{\alpha},P_{\beta} \big) = \| \Psi(\alpha)-\Psi(\beta) \|^2 $, 
which is the pull-back of the square-metric $(f,g) \mapsto \| f - g \|^2$ on $\mathcal H$ induced by the inner product. 
The corresponding information metric has components $I_{ij}$ in a local coordinate chart given by
\begin{align*}
I_{ij} & =-\frac{\partial}{\partial \alpha^k \partial{\beta^j}} D\big(p(\alpha),p(\beta) \big)|_{\alpha=\beta=\theta} = 2\partial_{\beta^j}\partial_{\alpha^k} \langle  \Psi(\alpha), \Psi(\beta) \rangle |_{\alpha=\beta=\theta} .
\end{align*}
Suppose now that $\mathcal H_k$ is a RKHS, and $\Psi$ is defined as the mean-embedding.
Then
\begin{align*} \langle  \Psi(\alpha), \Psi(\beta) \rangle
& = \int_{\mathcal{X}} \int_{\mathcal{X}} k(x,y) P_{\alpha}(\mathrm{d}x)P_{\beta}(\mathrm{d}y).
\end{align*}
In particular if the measures in $S$ can be written as either $P_{\alpha}= G_{\alpha}^{\#}\mu$, or $ P_{\alpha}(\mathrm{d}x)=p_{\alpha}(x) \mu(\mathrm d x)$ for some fixed measure $\mu$, then
$\partial_{\beta^j}\partial_{\alpha^k} \langle  \Psi(\alpha), \Psi(\beta) \rangle |_{\alpha=\beta=\theta}=\langle \partial_{\theta^j} \Psi(\theta),\partial_{\theta^k} \Psi(\theta) \rangle $ and we recover the pullback Riemannian metric.

\subsection{Geodesics of the MMD metric} \label{appendix:information_geometry}

The following result summarises the properties of the geodesics induced by the MMD metric on $\mathcal{P}_k$.  

\begin{proposition}[The MMD Information Metric]\label{prop:MMD2_metrictensor}
Suppose that $k$ is a characteristic kernel with a bounded continuous derivative and that assumptions (i)-(iv) stated above hold.   If the matrix $g(\theta) = (g_{ij}(\theta))_{i,j = 1,\ldots, p}$ is positive definite on $\Theta$, then the MMD metric on $\mathcal{P}_k$ induces a Riemannian geometry $(\Theta, g)$ on $\Theta$. The metric induced on $\Theta$ is given by
\begin{align}
\label{eq:induced_distance}
d_{MMD}^2(\theta | \theta') & = \inf_{\theta(t)\in C^1(0,1)}\left[\int_0^1 \dot{\theta}(t)^\top g(\theta)\dot{\theta}(t)dt  :  \theta(0) = \theta, \theta(1) = \theta'\right],
\end{align}
for all $\theta, \theta' \in \Theta$.  Geodesics in $(\Theta, g)$ are given by infimisers by \eqref{eq:induced_distance} and satisfy the following system of ODEs \cite{dubrovinmodern}:
\begin{equation}
\begin{aligned}
&\dot{\theta}(t) - g^{-1}(\theta(t))S(t) = 0 \\
&\dot{S}(t) - \frac{1}{2}S(t)^\top \nabla_{\theta}g(\theta(t))^{-1}S(t) = 0.
\end{aligned}
\end{equation}
\end{proposition}

Sufficient conditions for $g$ being positive definite need to be verified on a case by case basis.   Since $(\mathcal{P}_k, \MMD)$ is a length space \citep{papadopoulos2014metric,burago2001course}, it follows immediately that geodesics in this metric is via teleportation of mass, i.e. a geodesic connecting $\mathbb{P}_1$ and $\mathbb{P}_2$ in $\mathcal{P}_k$ is defined by $\mathbb{P}_t = (1-t)\mathbb{P}_1 + t \mathbb{P}_2, t \in [0,1]$.  This will not be the case for $(\Theta, g)$ as geodesics $\theta(t)$ must be constrained to ensure that $\mathbb{P}_{\theta(t)} \in \mathcal{P}_{\Theta}$.


\section{Proofs of Main Results}
\label{appendix:proofs}

In this appendix, we give the proofs of all lemmas, propositions and theorems in the main text.

\subsection{Proof of Theorem \ref{thm:generalisation}}

Before moving on to Theorem \ref{thm:generalisation}, we show the following result, which proves that the there is a uniform bound between the different versions of the MMD discrepancy.  First, for convenience we define the following approximation to MMD between a measure $\mathbb{P}$ and a empirical measure $\mathbb{Q}^m(\mathrm{d}y) = \frac{1}{m}\sum_{i=1}^{m}\delta_{y_i}(\mathrm{d}y)$:
\begin{align*}
{\MMD}_U^2(\mathbb{P}  ||  \mathbb{Q}^m) & = \int_{\mathcal{X}} \int_{\mathcal{X}} k(x,y)\mathbb{P}(\mathrm{d}x)\mathbb{P}(\mathrm{d}y) - \frac{2}{m}\int_{\mathcal{X}} \sum_{i=1}^{m} k(x, y_i)\mathbb{P}(\mathrm{d}x) + \frac{1}{m(m-1)}\sum_{i \neq i'} k(y_i, y_{i'}).
\end{align*}
Note that if $\{y_j\}_{j=1}^m \iid \mathbb{Q}$ then $\mathbb{E}[{\MMD}_U^2(\mathbb{P} || \mathbb{Q}^m)]= {\MMD}^2(\mathbb{P} || \mathbb{Q})$.

\begin{lemma}
\label{lemma:mmd_bounds}
Suppose that $k$ is bounded, then for any two $\mathbb{P},\mathbb{Q} \in \mathcal{P}_k(\mathcal{X})$ and empirical distribution $\mathbb{Q}^{m} = \frac{1}{m}\sum_{i=1}^m \delta_{y_j}$ in $\mathcal{P}_k(\mathcal{X})$ made of independently and identically distributed realisations of $\mathbb{Q}$, we have: $\left| {\MMD} ^2_{U}(\mathbb{P}  ||  \mathbb{Q}^m) - {\MMD}^2(\mathbb{P}  ||  \mathbb{Q}^m) \right|  \leq 2m^{-1}\sup_{x\in \mathcal X}k(x,x)$ and:
\begin{align*}
{\MMD}^2(\mathbb{P}  ||  \mathbb{Q}) 
& =  \mathbb{E}[ {\MMD}^2(\mathbb{P}  ||  \mathbb{Q}^m)] + m^{-1}\left(\int_{\mathcal{X}}\int_{\mathcal{X}} k(x,y)\mathbb{Q}(\mathrm{d}x)\mathbb{Q}(\mathrm{d}y) - \int_{\mathcal{X}} k(x,x)\mathbb{Q}(\mathrm{d}x)\right).  
\end{align*}
Similarly, when computing the MMD squared between $\mathbb{Q}^m$ and $\mathbb{P}^n = \frac{1}{n} \sum_{i=1}^n \delta_{x_i} \in \mathcal{P}_k(\mathcal{X})$ (made out of IID realisations from $\mathbb{P}$) $\left| {\MMD}^2_{U, U}(\mathbb{P}^n  ||  \mathbb{Q}^m) - {\MMD}^2(\mathbb{P}^n  ||  \mathbb{Q}^m) \right| \leq 2\left(m^{-1}+ n^{-1}\right)\sup_{x\in  \mathcal X}k(x,x)$, and similarly:
\begin{align*}
{\MMD}^2(\mathbb{P}  ||  \mathbb{Q}) 
& =  \mathbb{E}[ {\MMD}^2(\mathbb{P}^n  ||  \mathbb{Q}^m)] + \left(m^{-1}+n^{-1}\right)\left(\int_{\mathcal{X}}\int_{\mathcal{X}} k(x,y)\mathbb{Q}(\mathrm{d}x)\mathbb{Q}(\mathrm{d}y) - \int_{\mathcal{X}} k(x,x)\mathbb{Q}(\mathrm{d}x)\right).  
\end{align*}
\end{lemma}
\begin{proof}
We see that
\begin{align*}
	& {\MMD}^2_{U}(\mathbb{P}  ||  \mathbb{Q}^m) - {\MMD}^2(\mathbb{P}  ||  \mathbb{Q}^m) \\
  &= (m(m-1))^{-1}\sum_{i\neq j} k(y_i, y_j) - m^{-2}\sum_{i=1}^m \sum_{j=1}^{m} k(y_i, y_j) \\
	&=
  (m(m-1))^{-1}(1-(m(m-1))m^{-2})\sum_{i\neq j}k(y_i, y_j) 
  - m^{-2}\sum_{i=1}^{m}k(y_i, y_i) \\
	&= m^{-1}\Big((m(m-1))^{-1}\sum_{i\neq j}k(y_i, y_j) 
  - m^{-1} \sum_{i=1}^{m}k(y_i, y_i)\Big).
\end{align*}
Since the kernel is bounded, it follows that $\left|{\MMD}^2_{U}(\mathbb{P}  ||  \mathbb{Q}^m) - {\MMD}^2(\mathbb{P}||\mathbb{Q}^m)\right| \leq 2m^{-1}\sup_{x\in  \mathcal X}k(x,x)$ as required. The second statement follows in a similar fashion and from the fact that ${\MMD}^2_U$ is an unbiased estimator of ${\MMD}^2$.  Similarly for the discrepancy ${\MMD}^2_{U,U}$:
\begin{align*}
{\MMD}^2_{U,U}(\mathbb{P}^n  ||  \mathbb{Q}^m) - {\MMD}^2(\mathbb{P}^n  ||  \mathbb{Q}^m) 
& =
	m^{-1}(m(m-1))^{-1}\sum_{i\neq j}k(y_i, y_j) - m^{-1}\sum_{i=1}^{m}k(y_i, y_i)) \\
	& \quad + n^{-1}((n(n-1))^{-1}\sum_{i\neq j}k(x_i, x_j) - n^{-1}\sum_{i=1}^{n}k(x_i, x_i)),
\end{align*}
so $\left| {\MMD}^2_{U,U}(\mathbb{P}^n  ||  \mathbb{Q}^m) - {\MMD}^2(\mathbb{P}^n  ||  \mathbb{Q}^m)\right| \leq 2(m^{-1} + n^{-1})\sup_{x\in \mathcal X}k(x,x)$ and the final equation holds similarly.
\end{proof} 


We now establish conditions under which a minimiser of the empirical loss always exists.  
\begin{lemma}
\label{lemma:existence} 
Suppose that the kernel $k$ is continuous and bounded and that the map $\theta \rightarrow G_{\theta}(u)$ continuous for almost every $u \in \mathcal{U}$ and $\theta \in \Theta$. Then given $n, m \in \mathbb{N}$ the following statements hold.
\begin{enumerate}
	\item Let $\epsilon^* = \inf_{\theta \in \Theta} {\MMD}(\mathbb{P}_{\theta}  ||  \mathbb{Q}^m)$.  Then  if for some $\epsilon = \epsilon(m, \omega) > 0$ the set $$\left\lbrace \theta \in \Theta  : {\MMD}(\mathbb{P}_{\theta} || \mathbb{Q}^m) \leq \epsilon^* + \epsilon \right\rbrace \subset \Theta,$$ is bounded then $\arg\inf_{\theta \in \Theta}{\MMD}(\mathbb{P}_{\theta} ||  \mathbb{Q}^m) \neq \emptyset$.
	
	\item  Let $\epsilon^* = \inf_{\theta \in \Theta} {\MMD}(\mathbb{P}_{\theta}^n  ||  \mathbb{Q}^m)$, if for some $\epsilon = \epsilon(n, m, \omega) > 0$ the set $$\left\lbrace \theta \in \Theta  : {\MMD}(\mathbb{P}^n_{\theta} || \mathbb{Q}^m) \leq \epsilon^* + \epsilon \right\rbrace \subset \Theta,$$ is bounded then $\arg\inf_{\theta \in \Theta} {\MMD}(\mathbb{P}_{\theta}^n ||  \mathbb{Q}^m) \neq \emptyset$.
\end{enumerate}
\end{lemma}
\begin{proof}
The continuity assumption on $G_{\theta}$ implies that the map $\theta \rightarrow {\MMD}(\mathbb{P}_{\theta}  ||  \mathbb{Q}^m)$ is continuous from $\Theta$ to $[0, \infty)$. By definition of the infimum, it follows that $\left\lbrace \theta \in \Theta  : {\MMD}(\mathbb{P}_{\theta} || \mathbb{Q}^m) \leq \epsilon^* + \epsilon \right\rbrace\neq \emptyset$.  Moreover, by continuity of the map, the set is closed and bounded in $\Theta$ and thus compact in $\Theta$.  The map $\theta \rightarrow {\MMD}(\mathbb{P}_{\theta} || \mathbb{Q}^m)$ therefore will attain its minimum within the set, and so the first statement follows.  The result for the second estimator follows in an analogous fashion.
\end{proof}

We now provide the key concentration inequality.
\begin{proof}[Proof of Lemma \ref{lemma:concentration}]
Let $\mathcal{F}_{k} = \left\lbrace f \in \mathcal{H}_k \, : \, \lVert f \rVert_{\mathcal{H}_k} \leq 1 \right\rbrace$.  By definition, we have ${\MMD}(\mathbb{P} || \mathbb{P}^n) = \sup_{f \in \mathcal{F}_k}|\int_{\mathcal{X}} f(x)  \mathbb{P}(\mathrm{d}x) - \frac{1}{n}\sum_{i=1}^n f(x_i)|$. Define $h(x_1, \ldots, x_n) = \sup_{f \in \mathcal{F}_k}| \frac{1}{n}\sum_{i=1}^n (f(x_i)- \int_{\mathcal{X}} f(x)  \mathbb{P}(\mathrm{d}x))|$. By definition, for all $\{x_i\}_{i=1}^n$, $x_i' \in \mathcal{X}$, 
\begin{align*}
& |h(x_1,\ldots, x_{i-1}, x_i, x_{i+1},\ldots, x_n) - h(x_1,\ldots, x_{i-1}, x_i', x_{i+1},\ldots, x_n)|\\
 & \leq 2 n^{-1} \sup_{x\in \mathcal{X}}k(x,x)^{1/2}.
|h(x_1,\ldots, x_{i-1}, x_i, x_{i+1},\ldots, x_n) - h(x_1,\ldots, x_{i-1}, x_i', x_{i+1},\ldots, x_n)| \\
& \leq 2 n^{-1} \sup_{x\in \mathcal{X}}k(x,x)^{1/2}.
\end{align*}
By McDiarmid's inequality \citep{Mcdiarmid1989} we have that for any $\varepsilon >0$: $Pr( {\MMD}(\mathbb{P} || \mathbb{P}^n) - \mathbb{E}[{\MMD}(\mathbb{P} || \mathbb{P}^n)]  \geq  \varepsilon ) \leq  \exp ( - 2\varepsilon^2 / 4 n^{-1} \sup_{x \in \mathcal{X}} k(x,x))$.
Setting the RHS to be $\delta$, it follows that with probability greater than $1 - \delta$, 
\begin{align*}
{\MMD}(\mathbb{P} || \mathbb{P}^n) - \mathbb{E}\left[{\MMD}(\mathbb{P} || \mathbb{P}^n)\right] & < \sqrt{2 n^{-1}\sup_{x\in \mathcal{X}}k(x,x)\log(1/\delta)}.
\end{align*}
From Jensen's inequality and Lemma \ref{lemma:mmd_bounds}, we obtain that
\begin{align*}
\mathbb{E}\left[{\MMD}(\mathbb{P} || \mathbb{P}^n)\right] 
& \leq  \mathbb{E}[{\MMD}^2(\mathbb{P} || \mathbb{P}^n)]^{1/2}  \leq  \sqrt{2 n^{-1}}\sup_{x \in \mathcal{X}}k(x,x)^{1/2},
\end{align*}
 so that the advertised result holds.
\end{proof}


We now prove Theorem \ref{thm:generalisation}: 
\begin{proof}
From Lemma \ref{lemma:mmd_bounds} and the fact that $\sqrt{a+b} \leq \sqrt{a}+\sqrt{b}$, we obtain $\forall \mathbb{P} \in \mathcal{P}_k(\mathcal{X})$, $| {\MMD}_U(\mathbb{P} || \mathbb{Q}^m) - {\MMD}(\mathbb{P} ||\mathbb{Q}^m)| \leq \sqrt{2 m^{-1} \sup_{x \in \mathcal{X}}k(x,x)}$. In particular, since $\theta \rightarrow {\MMD}(\mathbb{P}_{\theta} ||  \mathbb{Q}^m)$ is bounded from below, using the above inequality and the definition of $\hat{\theta}_{m}$, we obtain that:
\begin{align*}
	{\MMD}\left(\mathbb{P}_{\hat{\theta}_m} \big|\big| \mathbb{Q}^m\right) 
	& \leq  {\MMD}_U\left(\mathbb{P}_{\hat{\theta}_m}  \big|\big|  \mathbb{Q}^m\right) + \sqrt{2 m^{-1} \sup_{x \in \mathcal{X}}k(x,x)}\\
	& =  \inf_{\theta \in \Theta} {\MMD}_U(\mathbb{P}_{\theta}  ||  \mathbb{Q}^m) + \sqrt{2 m^{-1} \sup_{x \in \mathcal{X}}k(x,x)} \\
	&\leq  \inf_{\theta \in \Theta} {\MMD}(\mathbb{P}_{\theta}  ||  \mathbb{Q}^m) + 2\sqrt{2 m^{-1} \sup_{x \in \mathcal{X}}k(x,x)}.
\end{align*}
We can then write:
\begin{align*}
&  {\MMD}\left(\mathbb{P}_{\hat{\theta}_{m}} \big|\big| \mathbb{Q}\right) - \inf_{\theta \in \Theta} {\MMD}(\mathbb{P}_{\theta} || \mathbb{Q}) \\
&\leq  
{\MMD}\left(\mathbb{P}_{\hat{\theta}_{m}} \big|\big| \mathbb{Q}\right) 
- {\MMD}\left(\mathbb{P}_{\hat{\theta}_{m}} \big|\big| \mathbb{Q}^m\right) 
+ {\MMD}\left(\mathbb{P}_{\hat{\theta}_{m}} \big|\big| \mathbb{Q}^m\right) - \inf_{\theta \in \Theta} {\MMD}(\mathbb{P}_{\theta} || \mathbb{Q}) \\
& \leq 
 {\MMD}\left(\mathbb{P}_{\hat{\theta}_{m}} \big|\big| \mathbb{Q}\right)  - {\MMD}\left(\mathbb{P}_{\hat{\theta}_{m}} \big|\big| \mathbb{Q}^m\right) + \inf_{\theta \in \Theta}{\MMD}(\mathbb{P}_{\theta} || \mathbb{Q}^m)  \\
 & \quad - \inf_{\theta \in \Theta} {\MMD}(\mathbb{P}_{\theta} || \mathbb{Q})+  2\sqrt{2 m^{-1} \sup_{x \in \mathcal{X}}k(x,x)}.
\end{align*}
Since the $\theta$-indexed family ${\MMD}(\mathbb{P}_{\theta}|| \cdot )$ is uniformly bounded (since $k$ is bounded), and using that for bounded functions $f,g:\mathbb R \rightarrow \mathbb R$, $| \inf_{\theta}f(\theta) -\inf_{\theta}g(\theta)| \leq \sup_{\theta} |f-g|$ and the reverse triangle inequality, we further obtain that
\begin{align*}
& {\MMD}\left(\mathbb{P}_{\hat{\theta}_{m}} \big|\big| \mathbb{Q}\right) - \inf_{\theta \in \Theta} {\MMD}(\mathbb{P}_{\theta} || \mathbb{Q}) \\ 
&\leq 2\sup_{\theta \in \Theta} \left| {\MMD}(\mathbb{P}_{\theta} || \mathbb{Q})  - {\MMD}(\mathbb{P}_{\theta} || \mathbb{Q}^m)\right| +  2\sqrt{2 m^{-1} \sup_{x \in \mathcal{X}}k(x,x)}\\
&\leq 2\sup_{\theta \in \Theta}\MMD(\mathbb{Q}  ||  \mathbb{Q}^m)+  2\sqrt{2 m^{-1} \sup_{x \in \mathcal{X}}k(x,x)}.
\end{align*}
Applying Lemma \ref{lemma:concentration}, with probability $1-\delta$, 
\begin{align*}
{\MMD}\left(\mathbb{P}_{\hat{\theta}_{m}} \big|\big| \mathbb{Q}\right) - \inf_{\theta \in \Theta} {\MMD}(\mathbb{P}_{\theta} || \mathbb{Q}) & \leq  2\sqrt{2m^{-1}\sup_{x\in \mathcal{X}} k(x,x)}(2 + \sqrt{\log(1/\delta)}),
\end{align*}
as required. For the second generalisation bound, note that
\begin{align*}
& {\MMD}\left(\mathbb{P}_{\hat{\theta}_{n,m}} \big|\big|  \mathbb{Q}\right) - \inf_{\theta \in \Theta} {\MMD}(\mathbb{P}_{\theta} ||  \mathbb{Q}) \\
&\leq 
 {\MMD}\left(\mathbb{P}_{\hat{\theta}_{n,m}} \big|\big|  \mathbb{Q}\right) - {\MMD}\left(\mathbb{P}^n_{\hat{\theta}_{n,m}} \big|\big|  \mathbb{Q}\right) + {\MMD}\left(\mathbb{P}^n_{\hat{\theta}_{n,m}} \big|\big|  \mathbb{Q}\right)  - {\MMD}\left(\mathbb{P}^n_{\hat{\theta}_{n,m}} \big|\big|  \mathbb{Q}^m\right)  \\
& + {\MMD}\left(\mathbb{P}^n_{\hat{\theta}_{n,m}} \big|\big|  \mathbb{Q}^m \right) - \inf_{\theta\in\Theta}{\MMD}(\mathbb{P}^n_{\theta} ||  \mathbb{Q})   + \inf_{\theta\in\Theta}{\MMD}(\mathbb{P}^n_{\theta} ||  \mathbb{Q}) - \inf_{\theta\in\Theta}{\MMD}(\mathbb{P}_{\theta} ||  \mathbb{Q}).
\end{align*}
We can bound the individual terms on the RHS as follows via the triangle inequality,
\begin{align*}
& {\MMD}\left(\mathbb{P}_{\hat{\theta}_{n,m}} \big|\big|  \mathbb{Q}\right) - \inf_{\theta \in \Theta} {\MMD}(\mathbb{P}_{\theta} ||  \mathbb{Q}) \\
&\leq 
 {\MMD}\left(\mathbb{P}_{\hat{\theta}_{n,m}} \big|\big|  \mathbb{P}^n_{\hat{\theta}_{n,m}}\right) + {\MMD}(\mathbb{Q}^m  ||  \mathbb{Q})+ {\MMD}(\mathbb{P}^n_{\hat{\theta}_{n,m}} ||  \mathbb{Q}^m)  \\
&  - \inf_{\theta\in\Theta}{\MMD}(\mathbb{P}^n_{\theta} ||  \mathbb{Q}) + \inf_{\theta\in\Theta}{\MMD}(\mathbb{P}^n_{\theta} ||  \mathbb{Q}) - \inf_{\theta\in\Theta}{\MMD}(\mathbb{P}_{\theta} ||  \mathbb{Q}).
\end{align*}
Similarly as above,
\begin{align*}
	{\MMD}\left(\mathbb{P}^n_{\hat{\theta}_{n,m}}  \big|\big|  \mathbb{Q}^m\right) 
	&\leq 
	 {\MMD}_{U,U}\left(\mathbb{P}_{\hat{\theta}_{n,m}}^n  \big|\big|  \mathbb{Q}^m\right) + \sqrt{2\left(m^{-1} + n^{-1}\right) \sup_{x \in \mathcal{X}}k(x,x)} \\
	& =  \inf_{\theta \in \Theta} {\MMD}_{U, U}(\mathbb{P}^n_{\theta}  ||  \mathbb{Q}^m)
   + \sqrt{2\left(m^{-1} + n^{-1}\right) \sup_{x \in \mathcal{X}}k(x,x)} \\
	&\leq  \inf_{\theta \in \Theta} {\MMD}(\mathbb{P}^n_{\theta}  ||  \mathbb{Q}^m)
   + 2\sqrt{2\left(m^{-1} +n^{-1}\right) \sup_{x \in \mathcal{X}}k(x,x)}.
\end{align*}
Similarly we obtain that
\begin{align*}
&  {\MMD}\left(\mathbb{P}^n_{\hat{\theta}_{n,m}} \big|\big|  \mathbb{Q}^m\right) - \inf_{\theta\in\Theta}{\MMD}(\mathbb{P}^n_{\theta} ||  \mathbb{Q}) \\
&\leq \inf_{\theta \in \Theta} {\MMD}(\mathbb{P}^n_{\theta}  ||  \mathbb{Q}^m) - \inf_{\theta\in\Theta}{\MMD}(\mathbb{P}^n_{\theta} ||  \mathbb{Q}) 
+2\sqrt{2\left(m^{-1} + n^{-1}\right) \sup_{x \in \mathcal{X}}k(x,x)}\\
&\leq  \sup_{\theta \in \Theta} \left|{\MMD}(\mathbb{P}^n_{\theta}  ||  \mathbb{Q}^m) - {\MMD}(\mathbb{P}^n_{\theta} ||  \mathbb{Q})\right| 
+2\sqrt{2\left(m^{-1} + n^{-1}\right) \sup_{x \in \mathcal{X}}k(x,x)} \\
&\leq  {\MMD}( \mathbb{Q}  ||  \mathbb{Q}^m) 
+ 2\sqrt{2\left(m^{-1} + n^{-1}\right) \sup_{x \in \mathcal{X}}k(x,x)},
\end{align*}
and  $\inf_{\theta\in\Theta}{\MMD}(\mathbb{P}^n_{\theta} ||  \mathbb{Q}) - \inf_{\theta\in\Theta}{\MMD}(\mathbb{P}_{\theta} ||  \mathbb{Q}) 
\leq
 \sup_{\theta\in\Theta}{\MMD}(\mathbb{P}_{\theta}^n  || \mathbb{P}_{\theta})$. Combining these inequalities we obtain, 
\begin{align*}
& {\MMD}\left(\mathbb{P}_{\hat{\theta}_{n,m}} \big|\big|  \mathbb{Q}\right) - \inf_{\theta \in \Theta} {\MMD}(\mathbb{P}_{\theta} ||  \mathbb{Q}) \\
&\leq 2\sup_{\theta \in \Theta} {\MMD}(\mathbb{P}_{\theta} ||  \mathbb{P}^n_{\theta})  + 2{\MMD}(\mathbb{Q}^m  ||  \mathbb{Q}) + 2\sqrt{2\left(m^{-1} + n^{-1}\right) \sup_{x \in \mathcal{X}}k(x,x)}. 
\end{align*}
 Applying Lemma \ref{lemma:concentration}  with probability $1- 2\delta$,
\begin{align*}
& {\MMD}\left(\mathbb{P}_{\hat{\theta}_{n,m}} \big|\big|  \mathbb{Q}\right) - \inf_{\theta \in \Theta} {\MMD}(\mathbb{P}_{\theta} ||  \mathbb{Q})  \\
& \leq
 2 (\sqrt{2 n^{-1}}+ \sqrt{2 m^{-1}})\sqrt{\sup_{x \in \mathcal{X} }k(x,x)} (1 + \sqrt{\log(1/\delta)}) +2\sqrt{2 (m^{-1} + n^{-1}) \sup_{x \in \mathcal{X}}k(x,x)} \\
&\leq 
2 (\sqrt{2 n^{-1}} + \sqrt{2 m^{-1}})\sqrt{\sup_{x\in\mathcal{X}}k(x,x)} (2 + \sqrt{\log(1/\delta)}).
\end{align*}
\end{proof}


\subsection{Proof of Proposition \ref{prop:consistency}}

\begin{proof}
Given $m \in \mathbb{N}$ define the event 
\begin{align*}
A_m & = \left\lbrace \left|\MMD\left(\mathbb{P}_{\hat{\theta}_m} \big|\big|  \mathbb{Q}\right) - \inf_{\theta' \in \Theta}\MMD(\mathbb{P}_{{\theta}'} || \mathbb{Q})\right| > 2\sqrt{\frac{2}{m}\sup_x k(x,x)}(2 + \sqrt{2 \log m}) \right\rbrace.
\end{align*}
From Theorem \ref{thm:generalisation}  (where we have set $\delta =1/m^2$),  $\mathbb{Q}(A_m) \leq \frac{1}{m^2}$, and so $\sum_{m} \mathbb{Q}(A_m) < \infty$.  The Borel Cantelli lemma implies that $\mathbb{Q}$-almost surely, there exists $M \in \mathbb{N}$ such that for all $m \geq M$, 
\begin{align*}
  \MMD\left(\mathbb{P}_{\hat{\theta}_m} \big|\big| \mathbb{Q}\right) - \inf_{\theta' \in \Theta}\MMD(\mathbb{P}_{{\theta}'} || \mathbb{Q}) \leq 2\sqrt{\frac{2}{m}\sup_x k(x,x)}(2 + \sqrt{2 \log m}).
\end{align*}
Since the right hand side converges to zero, it follows that $\MMD(\mathbb{P}_{\hat{\theta}_m} || \mathbb{Q}) \rightarrow  \inf_{\theta' \in \Theta}\MMD(\mathbb{P}_{{\theta}'} || \mathbb{Q}) = \MMD(\mathbb{P}_{{\theta}^*} || \mathbb{Q})$, $\mathbb{Q}$-almost surely.  By assumption (ii), the set $\lbrace \hat{\theta}_m\rbrace_{m \in \mathbb{N}}$ is bounded almost surely and thus possesses at least one limit point in $\Theta$. Moreover each subsequence $(\hat{\theta}_{n_k})_{k \in \mathbb{N}}$ satisfies $\MMD(\mathbb{P}_{\hat{\theta}_{m_k}} || \mathbb{Q}) \rightarrow   \MMD(\mathbb{P}_{{\theta}^*} || \mathbb{Q})$, so that any limit point must equal $\theta^*$, thus establishing almost sure convergence.  The consistency for the estimator $\hat{\theta}_{m,n}$ follows in an analogous manner.
\end{proof}


\subsection{Proof of Theorem \ref{thm:asympt_normal2}}

\begin{proof}

We shall prove the result only for the estimator $\hat{\theta}_{n, m}$ since the proof of the   central limit theorem for $\hat{\theta}_m$ follows in an entirely analogous way. 
Recall that
\begin{align*}
 {\MMD}^2_{U,U}(\mathbb{P}_{\theta}^n  ||  \mathbb Q^m) 
  =  \frac{1}{n(n-1)}&\sum_{i \neq j} k(G_\theta(u_i), G_\theta(u_j))+ \\
  &- \frac{2}{mn} \sum_{i=1}^{n}\sum_{j=1}^{m} k(G_\theta(u_i), y_j) + \frac{1}{m(m-1)}\sum_{i \neq j} k(y_i, y_{j})  
\end{align*}
  
For $n, m \in \mathbb{N}$ define $F_{n, m}(\theta) = F_{n,m}(\theta, \omega)$ by 
$F_{n, m}(\theta) = {\MMD}^2_{U,U}\left(\mathbb{P}_{\theta}^n  || \mathbb{Q}^m\right)$. By definition $\hat{\theta}_{n,m}$ is a local minimum for $F_{n,m}$, so the first order optimality condition implies that  $\nabla_{\theta} F_{n, m}(\hat{\theta}_{n,m}) = 0$. Since $\Theta$ is open, by applying the mean value theorem to $\nabla_{\theta} F_{n,m}$ we obtain $0 = \nabla_{\theta}F_{n, m}(\theta^*)  +\nabla_{\theta}\nabla_{\theta}F_{n, m}(\tilde{\theta})(\hat{\theta}_{n,m} - \theta^*)$, where $\tilde{\theta}$ lies on the line between $\hat{\theta}_{n,m}$ and $\theta^*$.  Let $\{u_i\}_{i=1}^n$ be independently and identically distributed realisations from $\mathbb{U}$.  Since $\mathbb{Q} = G_{\theta^*}^\# \mathbb{U}$, there exist $\lbrace \tilde{u}_1,\ldots, \tilde{u}_m \rbrace$ which are $\mathbb{U}$ distributed and independent from $\lbrace u_i \rbrace$ such that
\begin{align*}
\nabla_{\theta}F_{n, m}(\theta^*) 
& =  \frac{2}{n(n-1)}\sum_{i \neq j}\nabla_{1} k(G_{\theta^*}(u_i), G_{\theta^*}(u_j)) \nabla_{\theta} G_{\theta^*}(u_i) \\
&   - \frac{2}{nm}\sum_{i=1}^{n}\sum_{j=1}^{m}\nabla_{1} k(G_{\theta^*}(u_i), G_{\theta^*}(\tilde{u}_j)) \nabla_{\theta} G_{\theta^*}(u_i).
\end{align*} 
Note that $\mathbb{E}[\nabla_{\theta} F_{n, m}(\theta^*)] = 0$. We wish to characterise the fluctuations of $\nabla_{\theta}F_{n, m}(\theta^*)$ as $n,m \rightarrow \infty$. 
Define the U-statistic 
$U_1 =(n(n-1))^{-1} \sum_{i\neq j} h(u_i, u_j)$,  where 
\begin{align*}
h(u, v) = \nabla_{1} k(G_{\theta^*}(u), G_{\theta^*}(v))\nabla_{\theta} G_{\theta^*}(u) + \nabla_{1} k(G_{\theta^*}(v), G_{\theta^*}(u)) \nabla_{\theta} G_{\theta^*}(v),
\end{align*}
and the U-statistic $U_2=  (nm)^{-1} \sum_{i,j=1}^{n,m} g(u_i, \tilde{u}_j),$ where 
\begin{align*}
g(u, v) = 2\nabla_{1} k(G_{\theta^*}(u), G_{\theta^*}(v)) \nabla_{\theta} G_{\theta^*}(u).
\end{align*}

From the calculations above we have $\nabla_{\theta}F_{n, m}(\theta^*) = U_1 - U_2$. Following \cite{VanderVaart1998}  we make use of the Hajek projection principle to identify  $U_1 - U_2$ as small perturbation of a sum of independently and identically distributed random variables, from which a central limit theorem can be obtained, see also \cite{hoeffding1948class,Lehmann1951}. 
 To this end, we look for a projection onto the set of all random variables of the form $\sum_{i=1}^{n} \hat{h}_i(u_i) - \sum_{j=1}^{m} \hat{g}_i(\tilde{u}_j)$, where $\hat{h}_i$ and $\hat{g}_i$ are square-integrable measurable functions.  Let $M = \mathbb{E}[U_1]= \mathbb{E}[U_2]$, the Hajek projection principle \cite[Chap. 11 \& 12]{VanderVaart1998} implies that $U_1 - M$ has projection $\hat{U}_1  =  \frac{2}{n}\sum_{i=1}^{n} h_1(u_i)$, where $h_1(u) = \mathbb{E}_{X} h(u, X) - M$.  Similarly, $U_2 - M$ has projection $\hat{U}_2  =  \frac{1}{n}\sum_{i=1}^{n} g_{1}(u_i) + \frac{1}{m}\sum_{i=1}^m g_2(\tilde{u}_i)$, where $g_1(u) = \mathbb{E}g(u, Y) - M $ and $g_2(y) = \mathbb{E}g(X, y) - M$. By the central limit theorem for identically and independently distributed random variables, $\sqrt{n+m}(\hat{U}_1 - \hat{U}_2) \xrightarrow{d} \mathcal{N}(0, \Sigma)$, where $\Sigma =  A + B - 2C$ and 
\begin{align*}
A & = 
\lim_{k\rightarrow\infty} 4(n_k+m_k) n_k^{-2}\sum_{i=1}^{n_k} \mbox{Cov}[h_1(u_i)] \\
&  =  4 \lambda^{-1} \int_{\mathcal{U}} \left(\int_{\mathcal{U}} \left(h(u , v) - M\right)\mathbb{U}(\mathrm{d}v) \otimes \int_{\mathcal{U}} \left(h(u, w) - M\right)\mathbb{U}(\mathrm{d}w) \right) \mathbb{U}(\mathrm{d}u),
\end{align*}
where $\lambda$ is defined in Assumption 4.  Similarly,
\begin{align*}
B & =  \lim_{k\rightarrow \infty} (n_k+m_k/n_k^2)\sum_{i=1}^{n_k}\mbox{Cov}[g_1(u_i)] + (m_k + n_k/m_k^2)\sum_{i=1}^{m_k}\mbox{Cov}[g_2(\tilde{u}_i)] \\
	& =   \lambda^{-1} \int_{\mathcal{U}} \left(\int_{\mathcal{U}} (g(u, v) - M) \mathbb{U}(\mathrm{d}v)\otimes \int_{\mathcal{U}} (g(u, w) - M) \mathbb{U}(\mathrm{d}w)\right)\mathbb{U}(\mathrm{d}u) \\
	&  \quad + (1-\lambda)^{-1} \int_{\mathcal{U}} \left(\int_{\mathcal{U}} (g(u, v) - M) \mathbb{U}(\mathrm{d}u) \otimes \int_{\mathcal{U}} (g(w, v) - M) \mathbb{U}(\mathrm{d}w)\right)\mathbb{U}(\mathrm{d}v), \\
C & =  2\lim_{k\rightarrow \infty} (n_k + m_k)n_k^{-2}\mbox{Cov}\left[\sum_{i=1}^{n_k}h_1(u_i), \sum_{i=1}^{n_k}g_1(u_i)\right]  \\
  & =  2 \lambda^{-1} \int_{\mathcal{U}} \int_{\mathcal{U}} \left(h(u,v) - M\right)\mathbb{U}(\mathrm{d}v) \otimes \int_{\mathcal{U}} \left(g(u,w) - M\right)\mathbb{U}(\mathrm{d}w) \mathbb{U}(\mathrm{d}u),
\end{align*}
Substituting the values of $g$ and $h$ we arrive at $\Sigma$. We will show that the remainder term  $R_k  = \sqrt{n_k + m_k}((U_1 - \hat{U}_1) +(U_2 - \hat{U}_2))$ converges to $0$ in probability, as $k \rightarrow \infty$, which will imply  the desired result, by Slutsky's theorem.  This term has expectation zero for all $k \in \mathbb{N}$. Moreover
\begin{align*}
\mathbb{E}[|R_k|]^2 
& \leq 
2(n_k + m_k)n_k^{-1}n_k\text{Tr}(\mbox{Cov}[U_1 - \hat{U}_1]) + 2(n_k + m_k)\text{Tr}(\mbox{Cov}[U_2 - \hat{U}_2]).
\end{align*}
Using the fact that $n_k(n_k + m_k)^{-1} \rightarrow \lambda$ as $k \rightarrow \infty$, and by \cite[Theorem 12.3]{VanderVaart1998}, the first term converges on the right hand side converges to $0$.  For the second term, from \cite[Theorem 12.6]{VanderVaart1998} both $(n_k + m_k)\mbox{Tr}\left(\mbox{Cov}[U_2]\right)$ and $(n_k + m_k)\mbox{Tr}(\mbox{Cov}[\hat{U}_2])$ converge to
\begin{align}
\label{eq:var1}
& \lambda^{-1}\mbox{Tr}\left(\int_{\mathcal{U}} \left(\int_{\mathcal{U}} \left(g(u, v)- M\right)\mathbb{U}(\mathrm{d}v)\right)^{\otimes 2} \mathbb{U}(\mathrm{d}u)\right) \\
& + (1-\lambda)^{-1}\mbox{Tr}\left(\int_{\mathcal{U}} \left(\int_{\mathcal{U}} \left(g(u, v)- M\right)\mathbb{U}(\mathrm{d}u)\right)^{\otimes 2} \mathbb{U}(\mathrm{d}v)\right).
\end{align}
It remains to consider $\mbox{Cov}[U_2, \hat{U}_2]$ which is given by
\begin{align*}
	& {(n_k+m_k)}\mathbb{E}\left[\left(n_k^{-1} \sum_{i=1}^{n_k} g_1(u_i) + m_k^{-1}\sum_{i=1}^m g_2(\tilde{u}_i) - 2M\right)\otimes\left((n_k m_k)^{-1}\sum_{i,j=1}^{n_k, m_k}g(u_i, \tilde{u_j}) - M \right)\right] \\
	& =  ((n_k+m_k)n_k^{-1})(n_k+m_k)^{-1} \sum_{i=1}^{n_k} \mathbb{E}[g_1(u_i)^{\otimes 2}] \\
	& \quad + \left(\frac{n_k+m_k}{m_k}\right)\frac{1}{n_k+m_k} \sum_{i=1}^{m_k} \mathbb{E}[g_2(\tilde{u}_i)^{\otimes 2}] - 2(n_k + m_k)M\otimes M ,
\end{align*}
so that $\mbox{Tr}(\mbox{Cov}[U_2, \hat{U}_2])$ converges to \eqref{eq:var1} as $k\rightarrow \infty$, and so $\mbox{Cov}[U_2 - \hat{U_2}]\rightarrow 0$ as required. Now consider the term $H_{m,n} = \nabla_{\theta}\nabla_{\theta} F_{n, m}(\tilde{\theta})$ in the first order Taylor expansion, where $\tilde{\theta}$ lies along the line between $\theta^*$ and $\hat{\theta}_{n,m}$. 
 We show that $\nabla_\theta \nabla_{\theta}F_{n, m}(\tilde{\theta})$ converges to $g(\theta^*)$ as $n, m\rightarrow \infty$. 
  To this end, consider $H^{a,b}_{m, n}(\tilde{\theta}) - g_{ab}(\theta^*)$,
where
\begin{align*}
H^{a,b}_{m, n}({\theta}) & =  (n(n-1))^{-1}\partial_{\theta_{a}}\partial_{\theta_{b}}\sum_{i \neq j} k(G_{\theta}(u_i), G_{\theta}(u_j)) 
- 2(nm)^{-1}\partial_{\theta_{a}}\partial_{\theta_{b}}\sum_{i,j=1}^{n,m}k(G_{\theta}(u_i), y_j).
\end{align*}
Then we have that $|H^{a,b}_{m, n}(\tilde{\theta}) - g_{ab}(\theta^*)| \leq |H^{a,b}_{m, n}(\tilde{\theta}) - g_{ab}(\tilde{\theta})| + |g_{ab}(\tilde{\theta}) - {g}_{ab}(\theta^*)|$. Since $\hat{\theta}_{n,m} \rightarrow \theta^*$ almost surely, it follows that $\tilde{\theta} \rightarrow \theta^*$, and so for $n, m$ sufficiently large, $\tilde{\theta}$ almost surely lies in the compact set $K$.  Thus $|H^{a,b}_{m, n}(\tilde{\theta}) - {g}_{ab}(\theta^*)| \leq \sup_{\theta \in K}|H^{a,b}_{m, n}({\theta}) - {g}_{ab}({\theta})| + |{g}_{ab}(\tilde{\theta}) - {g}_{ab}(\theta^*)|$.

 It follows from the assumptions and the dominated convergence theorem that $\theta\rightarrow g_{ab}(\theta)$ is continuous on $K$.  By Assumption 3, the first three $\theta$-derivatives of $G_{\theta}$ are bounded in $K$ and so the conditions of Lemma \ref{lem:unif_lln} hold, so that the first term  goes to zero in probability. The second term converges to zero by continuity on $K$.  Since $g$ is assumed to be invertible, there exists $m=m(\omega), n = n(\omega)$ after which $H_{m,n}(\tilde{\theta})$ is also invertible, so that by Slutsky's theorem
\begin{align*}
\sqrt{n_k + m_k}(\hat{\theta}_{n,m} - \theta^*) 
& = 
 -(H_{m,n})^{-1}\sqrt{n_k + m_k}\nabla_{\theta}F_{n,m}(\theta^*)  \; \xrightarrow{d} \; \mathcal{N}(0, {g}(\theta^*)^{-1}\Sigma {g}(\theta^*)^{-1}).
\end{align*}
\end{proof}
\begin{lemma}
\label{lem:unif_lln}
Let $K$ be a compact set and  $q_1(x, y, \theta) = \partial_{\theta_{a}}\partial_{\theta_{b}}k(G_{\theta}(x), G_{\theta}(y))$ and $q_2(x, y, \theta) = 2\partial_{\theta_{a}}\partial_{\theta_{b}}k(G_{\theta}(x), y)$. Suppose that for $\theta_1, \theta_2 \in K$ we have, $|q_1(x,y, \theta_1) - q_1(x,y, \theta_2)| \leq (\theta_1 - \theta_2)Q_1(x, y)$ and $|q_1(x,y, \theta_1) - q_1(x,y, \theta_2)| \leq (\theta_1 - \theta_2)Q_2(x, y)$, where $\int_{\mathcal{X}}\int_{\mathcal{X}} Q_1(x, y)\mathbb{U}(\mathrm{d}x)\mathbb{U}(\mathrm{d}y) < \infty$ and $\int_{\mathcal{X}}\int_{\mathcal{X}} Q_2(x, y)\mathbb{U}(\mathrm{d}x)\mathbb{Q}(\mathrm{d}y) < \infty$. Then $\sup_{\theta \in K}|H_{m, n}^{a,b}(\theta) - g_{ab}(\theta)| \xrightarrow{p} 0$ as $m \wedge n \rightarrow \infty$.
\end{lemma} 

\begin{proof}
We show that the spaces of functions $\mathcal{Q}_1 = \lbrace q_1(\cdot, \cdot, \theta)  :  \theta \in K \rbrace$ and $\mathcal{Q}_2 =  \lbrace q_2(\cdot, \cdot, \theta)  :  \theta \in K \rbrace$ are Euclidean in the sense of \cite{Nolan1987}. Let $\epsilon > 0$ and let $\theta_1, \ldots, \theta_M \in K$ be centers of an $\epsilon$--cover of $K$, where $M = \mbox{diam}(K)/\epsilon$.  Given $q_i \in \mathcal{Q}_i$, $i=1,2$, there exists $\theta_k$ such that $|q_i(\cdot, \cdot, \theta_k) -   q_i(\cdot, \cdot, \theta)| \leq \epsilon Q_i(\cdot, \cdot)$, and so, given a measure $\mu$ on $(\Theta, \mathcal{B}(\Theta))$ such that $\mu(Q_i) < \infty$ we have $\mu|q_i(\cdot, \cdot, \theta_k) -   q_i(\cdot, \cdot, \theta)| \leq \epsilon \mu(Q_i)$, therefore $N_1(\epsilon, \mu, \mathcal{Q}_i) \leq \mbox{diam}(K)\epsilon^{-1}$.  Invoking \cite[Theorem 7]{Nolan1987} for $q_1$ and  \cite[Theorem 2.9]{Neumeyer2004} for $q_2$, we obtain the required result.
\end{proof}

\subsection{Proof of Theorem \ref{prop:cramer_rao}}
\begin{proof}
Define the function
\begin{align*}
h(x,\theta) = &2\int_{\mathcal{U}} \int_{\mathcal{U}} \nabla_{1} k(G_{\theta}(u), G_{\theta}(v)) \nabla_{\theta}G_{\theta}(u)\mathbb{U}(\mathrm{d}u)\mathbb{U}(\mathrm{d}v)\\
    & - 2 \int_{\mathcal{U}} \nabla_1 k(G_{\theta}(u), x) \nabla_{\theta}G_{\theta}(u)\mathbb{U}(\mathrm{d}u),
\end{align*}
which satisfies $\int_{\mathcal{X}} h(x,\theta)\mathbb{P}_{\theta}(\mathrm{d}x)=0$ for all $\theta \in \Theta$.  Differentiating this integral with respect to $\theta$ yields
$\int_{\mathcal{X}} \nabla_{\theta}h(x,\theta)\mathbb{P}_{\theta}(\mathrm{d}x)
=
-\int_{\mathcal{X}} h(x,\theta)\otimes\nabla_{\theta}p(x\,|\,\theta)\mathrm{d}x$, where $p(x \, | \theta)$ is the density of $\mathbb{P}_{\theta}$ with respect to the Lebesgue measure on $\mathcal{X}$.

Let $X\sim \mathbb{P}_{\theta}$. Consider the covariance of $(h(X,\theta),\nabla\log p(X|\theta))^{\top},$ then
\begin{align*}
\text{Cov} &(h(X,\theta),\nabla\log p_{\theta}(X))^{\top} \\
 & =\left(\begin{array}{cc}
\int_{\mathcal{X}} h(x,\theta)\otimes h(x,\theta)p(x|\theta) \mathrm{d}x& 
\int_{\mathcal{X}} h(x,\theta)\otimes\nabla\log p(x|\theta)p(x|\theta)\mathrm{d}x\\
\int_{\mathcal{X}} h(x,\theta)\otimes\nabla\log p(x|\theta)p(x|\theta) \mathrm{d}x &
\int_{\mathcal{X}} \nabla\log p(x|\theta)\otimes \nabla\log p(x|\theta)p(x|\theta) \mathrm{d}x
\end{array}\right)\\
 & =\left(\begin{array}{cc}
\int_{\mathcal{X}} h(x,\theta)\otimes h(x,\theta)p(x|\theta) \mathrm{d}x &
-\int_{\mathcal{X}} \nabla_{\theta}h(x,\theta)p(x|\theta) \mathrm{d}x\\
-\int_{\mathcal{X}} \nabla_{\theta}h(x,\theta)p(x|\theta) \mathrm{d}x &
 F(\theta)
\end{array}\right),
\end{align*}
where $F(\theta)$ is the Fisher information matrix. Since this is
a covariance matrix, the determinant is non-negative, and so
\begin{align*}
\det(F(\theta))\det\left(\int_{\mathcal{X}} h(x,\theta)\otimes h(x,\theta)\mathbb{P}_{\theta}(\mathrm{d}x)-\left(\int_{\mathcal{X}} \nabla_{\theta}h(x,\theta)\mathbb{P}_{\theta}(\mathrm{d}x)\right) F^{-1}(\theta)\left(\int_{\mathcal{X}} \nabla_{\theta}h(x,\theta)\mathbb{P}_{\theta}(\mathrm{d}x)\right)\right)& \geq 0.
\end{align*}
Since the Fisher information is positive at $\theta = \theta^*$ this implies that $\det F(\theta)>0$ and
so 
\begin{align*}
\int_{\mathcal{X}} h(x,\theta)\otimes h(x,\theta)\mathbb{P}_{\theta}(\mathrm{d}x)
-\left(\int_{\mathcal{X}} \nabla_{\theta}h(x,\theta)\mathbb{P}_{\theta}(\mathrm{d}x)\right)F^{-1}(\theta) \left(\int_{\mathcal{X}} \nabla_{\theta}h(x,\theta)\mathbb{P}_{\theta}(\mathrm{d}x)\right)
\end{align*}
 is non-negative definite. We note that $\int_{\mathcal{X}} \nabla_{\theta}h(x,\theta)p(x) \mathrm{d}x = g(\theta)$ is the information metric associateed with the $\mbox{MMD}$ induced
distance and is positive definite at $\theta = \theta^*$.  It follows
that $(1/4)g^{-1}(\theta) (\int_{\mathcal{X}} h(x,\theta)\otimes h(x,\theta)\mathbb{P}_{\theta}(\mathrm{d}x)) g^{-1}(\theta)-  F^{-1}(\theta)$ is non-negative definite at $\theta = \theta^*$. Since
\begin{align*}
\int_{\mathcal{X}} & h(x,\theta)\otimes h(x,\theta)\mathbb{P}_{\theta}(\mathrm{d}x)  =4\int_{\mathcal{U}} \left(\int_{\mathcal{U}}\nabla_{1}k(G_{\theta}(u),G_{\theta}(v))^{\top} \nabla_{\theta}G_{\theta}(u)\mathbb{U}(\mathrm{d}u)-\mathcal{M}\right)^{\otimes2}\mathbb{U}(\mathrm{d}v),
\end{align*}
where 
$\mathcal{M} = 
\int_{\mathcal{U}} \int_{\mathcal{U}} \nabla_{1}k(G_{\theta}(u),G_{\theta}(v))\mathbb{U}(\mathrm{d}u)^{\top}\nabla_{\theta}G_{\theta}(u)\mathbb{U}(\mathrm{d}v)$ we see that
$\frac{1}{4}g^{-1}(\theta)(\int_{\mathcal{X}} h(x,\theta)\otimes h(x,\theta)\mathbb{P}_{\theta}(\mathrm{d}x)) g^{-1}(\theta)$
equals the asymptotic variance C for the estimator $\hat{\theta}_{m}$ and so $C-F^{-1}(\theta)$ is positive definite when $\theta = \theta^*$ giving the advertised intequality.

Now since $C_{\lambda} = (1/(1-\lambda)\lambda) C \succeq C$, it follows that $C_{\lambda}-F^{-1}(\theta)$ is also positive definite when $\theta=\theta^*$ and the Cramer-Rao bound also holds for the estimator $\hat{\theta}_{n,m}$.
\end{proof}

\subsection{Proof of Proposition \ref{prop:bandwidth_limit}}
\begin{proof}
We have that $\nabla_{1}k(x,y)=((x-y)/l^{2})r'(|x-y|^{2}/2l^{2})$, and $\nabla_{1}\nabla_{2}k(x,y)=-l^{-2} r'(|x-y|^{2}/2l^{2})-l^{-4}(x-y)^{2} r''(|x-y|^{2}/2l^{2}))$.
We first note that the metric tensor $g$ satisfies
\begin{align*}
l^{2}g(\theta) & \xrightarrow{l^{2}\xrightarrow{}\infty}R \int_{\mathcal{U}} \int_{\mathcal{U}}\nabla G_{\theta}(u)\nabla G_{\theta}(v)^{\top}\mathbb{U}(\mathrm{d}u)\mathbb{U}(\mathrm{d}v)=\nabla_{\theta}M(\theta)\nabla_{\theta}M(\theta)^{\top},
\end{align*}
where $M(\theta)=\int_{\mathcal{X}} x p(x|\theta)\,\mathrm{d}x$ and $R = \lim_{s\rightarrow \infty} r'(s)$.  Defining $S(\theta)=\int_{\mathcal{U}} |G_{\theta}(u)|^{2}\mathbb{U}(du)$ we obtain:

\begin{align*}
l^{4}\Sigma & \xrightarrow{l^{2}\xrightarrow{}\infty}R^2\int_{\mathcal{U}} \Bigg[ \left(\int_{\mathcal{U}} \nabla_{\theta}G(u)\cdot(G_{\theta}(u)-G_{\theta}(v))\mathbb{U}(\mathrm{d}u)\right)\\
& \qquad \qquad \otimes\left(\int_{\mathcal{U}} \nabla_{\theta}G(w)\cdot(G_{\theta}(w)-G_{\theta}(v))\mathbb{U}(\mathrm{d}w)\right)\Bigg]\mathbb{U}(\mathrm{d}v) \\ 
 &\qquad -R^2 \left(\int_{\mathcal{U}} \int_{\mathcal{U}} \nabla_{\theta}G(u)(G_{\theta}(u)
  - G_{\theta}(v))\mathbb{U}(\mathrm{d}u)\mathbb{U}(\mathrm{d}v)\right)^{\otimes2}\\
 & =R^2\nabla_{\theta}M(\theta)\cdot\left(V(\theta)+M(\theta\right) M(\theta))\,\nabla_{\theta}M(\theta)^{\top}-\frac{R^2}{4}\left(\nabla_{\theta}|M(\theta)|^{2}\right)\left(\nabla_{\theta}|M(\theta)|^{2}\right)^{\top}\\
 & =R^2\nabla_{\theta}M(\theta)\cdot V(\theta)\,\nabla_{\theta}M(\theta)+\frac{R^2}{4}\left(\nabla_{\theta}|M(\theta)|^{2}\right)\left(\nabla_{\theta}|M(\theta)|^{2}\right)^{\top}-\frac{R^2}{4}\left(\nabla_{\theta}|M(\theta)|^{2}\right)\left(\nabla_{\theta}|M(\theta)|^{2}\right)^{\top}\\
 & =R^2\nabla_{\theta}M(\theta)\cdot V(\theta)\,\nabla_{\theta}M(\theta).
\end{align*}
Combining we obtain
\begin{align*}
\lim_{l\rightarrow\infty}C^l & =\left(\nabla_{\theta}M(\theta)\,\,\nabla_{\theta}M(\theta)^{\top}\right)^{-1}\nabla_{\theta}M(\theta)\cdot V\,\nabla_{\theta}M(\theta)\left(\nabla_{\theta}M(\theta)\,\,\nabla_{\theta}M(\theta)^{\top}\right)^{-1}\\
 & =\left(\nabla_{\theta}M(\theta)\right)^{\dagger}V(\theta)\left(\nabla_{\theta}M(\theta)\right)^{\dagger\top}.\\
\end{align*}
\end{proof}


\subsection{Proof of Theorem \ref{thm:qualitative_robustness}}

\begin{proof}
Let $\epsilon > 0$, and let $\delta$ be as in assumption (ii).  Suppose that  $\mathbb{Q}_1, \mathbb{Q}_2 \in \mathcal{P}_{k}$ satisfy $d_{BL}(\mathbb{Q}_1, \mathbb{Q}_2) < (1+\sqrt{k(0,0)})\delta/2$, where $d_{BL}$ denotes the Bounded Lipschitz or Dudley metric \citep{dudley2018real}.    By \citep[Theorem 21]{Sriperumbudur2009} it follows that  $\MMD(\mathbb{Q}_1 ||  \mathbb{Q}_2) < \delta/2$. Let $\theta^{(1)}$ and $\theta^{(2)}$ be the minimum MMD estimators which exist by assumption (i).  By the triangle inequality:
\begin{align*}
\MMD(\mathbb{P}_{\theta^{(2)}}||\mathbb{Q}_1) 
& \leq 
\MMD(\mathbb{P}_{\theta^{(2)}}||\mathbb{Q}_2) + \MMD(\mathbb{Q}_1|| \mathbb{Q}_2) \leq \MMD(\mathbb{P}_{\theta^{(2)}}||\mathbb{Q}_2) + \delta/2.
\end{align*}
Suppose that $\left|\theta - \theta^{(2)}\right|> \epsilon,$ then:
\begin{align*}
\MMD(\mathbb{Q}_1||\mathbb{P}_{\theta}) 
& \geq 
 \MMD(\mathbb{Q}_2|| \mathbb{P}_\theta) - \MMD(\mathbb{Q}_1||\mathbb{Q}_2) 
\geq 
 \MMD(\mathbb{Q}_2|| \mathbb{P}_\theta) - \delta/2\\
 & >  
 \MMD (\mathbb{Q}_2|| \mathbb{P}_{\theta^{(2)}}) + \delta/2
  \geq  
 \MMD (\mathbb{Q}_1|| \mathbb{P}_{\theta^{(2)}})
\end{align*}
This implies that $\theta^{(1)}$ must be in the ball $\{\theta : |\theta - \theta^{(2)}| <\epsilon\}$, i.e. that $|\theta^{(1)} - \theta^{(2)}| < \epsilon$ as required.  This implies that the map $T:\mathcal{P}_{k} \rightarrow \Theta$ defined by $T(\mathbb{Q}) = \arg\inf_{\theta\in \Theta}{\MMD}^2(\mathbb{P}_{\theta}  ||  \mathbb{Q})$ is continuous with respect to the weak topology on $\mathcal{P}_k$.  In particular, for $\mathbb{Q}^{m} = \frac{1}{m} \sum_{j=1}^m \delta_{y_j}$, since $T(\mathbb{Q}^{m})=\arg\inf_{\theta\in \Theta}{\MMD}^2_U(\mathbb{P}_{\theta}  ||  \mathbb{Q}^{m})$, by \citep[Theorem 2]{cuevas1988qualitative} it follows that the estimator $\hat{\theta}_m$ is qualtiatively robust.  The proof of that $\hat{\theta}_{n,m}$ is eventually qualitatively robust follows similarly.
\end{proof}


\subsection{Proof of Theorem \ref{prop:bias_robustness_MMD}}
\begin{proof}
Consider the influence function obtained from the kernel scoring rule: $\text{IF}_{\text{MMD}}(z,\mathbb{P}_\theta) =  \left( \int_{\mathcal{X}} \nabla_\theta \nabla_\theta S_{\text{MMD}}(x,\mathbb{P}_\theta) \mathbb{P}_{\theta}(\mathrm{d}x) \right)^{-1} \nabla_{\theta}S_{\text{MMD}}(z,\mathbb{P}_{\theta})$.  It is straightforward to show that under assumptions (i-iii), both the first and the second term are bounded in $z$, which directly implies that the whole influence function is bounded and hence the estimator is bias-robust.
\end{proof}

\section{Gaussian Location and Scale Models}\label{appendix:gaussian_models}

Throughout this section, we will repeatedly use the fact that the product of Gaussian densities can be obtained in closed form using the following expression:
\begin{align*}
\phi(x; m_1, \sigma_1^2)\phi(x;m_2, \sigma_2^2) = \phi(m_1;m_2,\sigma_1^2 + \sigma_2^2) \phi\left(x; \frac{m_1 \sigma_2^2+m_2 \sigma_1^2}{\sigma_1^2+ \sigma_2^2} , \frac{\sigma_1^2 \sigma_2^2}{\sigma_1^2+\sigma_2^2}\right)
\end{align*}
 where we denote by $\phi(x;m,\sigma)$ the density of a $d$-dimensional Gaussian with mean equals to some $m\in \mathbb{R}$ times a vector of ones, and covariance $\sigma^2$ times a $d$-by-$d$ identity matrix. Furthermore, we also use the following identities: $\int_{\mathcal{U}} u^\top A u \phi(u,0,\sigma) \mathrm{d}u = \sigma \text{Tr}(A)$ and $\int_{\mathcal{U}} \|u\|_{2}^4 \phi(u,0,\sigma) \mathrm{d}u = (d^2+2d) \sigma^2$.
\vspace{2mm}

\subsection{Gaussian Location Model - Asymptotic Variance in high dimensions}
\begin{proof}
 The generator is given by $G_{\theta}(u) = u + \theta$ and $\mathbb{U}$ is $\mathcal{N}(0,\sigma^2 I_{d\times d})$ distributed. Assume $\theta^*$ is the truth. We wish to compute the asymptotic variance of the estimator $\hat{\theta}_m$ of $\theta^*$.  First, we observe that the mean term satisfies: $\overline{M}  = 0$ since $k(u,v) = \phi(u; v, l^2)$ is symmetric with respect to $u$ and $v$.  Consider the  term:
\begin{align*}
&  \int_{\mathcal{U}} \nabla_\theta G_{\theta^*}(u) \nabla_1 k(G_{\theta^*}(u),G_{\theta^*}(v)) \mathbb{U}(\mathrm{d}u)\\
& = 
- \int_{\mathcal{U}} (u- v) l^{-2} \exp(-(u-v)^2 /2l^2)(2 \pi l^2)^{-\frac{d}{2}} \exp(-u^2/2 \sigma^2) (2 \pi \sigma^2)^{-\frac{d}{2}} \mathrm{d} u \\
& =  
- \int_{\mathcal{U}} (u- v) l^{-2} \phi(u;v,l^2) \phi(u;0,\sigma^2) \mathrm{d}u \\
& = 
- \int_{\mathcal{U}} (u- v) l^{-2} \phi(u;v\sigma^2 (l^2+\sigma^2)^{-1},l^2\sigma^2(l+\sigma^2)^{-1}) \phi(v;0,l^2+\sigma^2) \mathrm{d}u \\
& = 
- l^{-2} [(v\sigma^2(l^2+\sigma^2)^{-1} - v) \phi(v;0,l^2+\sigma^2)] 
 = 
- (l^2+\sigma^2)^{-1}  v \; \phi(v;0,l^2+\sigma^2) 
\end{align*}
Then, we have:
\begin{align*}
\Sigma 
& = 
\int_{\mathcal{U}} \left[ \int_{\mathcal{U}}  \nabla_1 k(u,v) \mathbb{U}(\mathrm{d}u)\right] \otimes \left[ \int_{\mathcal{U}}  \nabla_1 k(w,v) \mathbb{U}(\mathrm{d}w)\right] \mathbb{U}(\mathrm{d}v)\\
& = 
(l^2 +\sigma^2)^{-2} \int_{\mathcal{U}}  v\otimes v \phi(v;0,l^2+\sigma^2) \phi(v;0,l^2+\sigma^2) \phi(v;0,\sigma^2) \mathrm{d}v\\
& =  
(l^2+\sigma^2)^{-2} \int_{\mathcal{U}} v\otimes v \phi\left(v;0,\frac{l^2+\sigma^2}{2}\right) \phi(v;0,2(l^2+\sigma^2)) \phi(v;0,\sigma^2) \mathrm{d}v\\
& =  
(l^2 +\sigma^2)^{-2} \phi(0;0,2(l^2+\sigma^2)) \phi\left(0;0,\frac{3\sigma^2+l^2}{2}\right) \int_{\mathcal{U}} v\otimes v \phi\left(v;0,\frac{\sigma^2(l^2+\sigma^2)}{(3 \sigma^2 +l^2)}\right)  \mathrm{d}v\\
& =  (l^2+\sigma^2)^{-2} (2\pi )^{-\frac{2 d}{2}} (2(l^2 + \sigma^2))^{-\frac{d}{2}} \left(\frac{3\sigma^2 +l^2}{2}\right)^{-\frac{d}{2}} \frac{\sigma^2(l^2+\sigma^2)}{(3 \sigma^2 +l^2)}I_{d\times d}\\
& =  \sigma^2 (2 \pi)^{-d} (l^2+\sigma^2)^{-\frac{d}{2}-1} (3 \sigma^2 + l^2)^{-\frac{d}{2}-1}I_{d\times d}
\end{align*}

We can compute the metric tensor $g(\theta^*)$ similarly:
\begin{align*}
g(\theta^*)
& = 
\int_{\mathcal{U}} \int_{\mathcal{U}} \nabla_{\theta} G_{\theta^*}(u) \nabla_1 \nabla_2 k(u,v) \nabla_{\theta} G_{\theta^*}(u)^\top \mathbb{U}(\mathrm{d}u) \mathbb{U}(\mathrm{d}v) \\
& =  
\int_{\mathcal{U}} \int_{\mathcal{U}}  \nabla_1 \nabla_2 k(u,v)  \phi(u;0,\sigma^2) \phi(v;0,\sigma^2) \mathrm{d}u \mathrm{d}v \\
& = 
\int_{\mathcal{U}} \int_{\mathcal{U}} k(u,v)  \nabla \phi(u;0,\sigma^2)\otimes \nabla \phi(v;0,\sigma^2)  \mathrm{d}u \mathrm{d}v \\
& =  
\int_{\mathcal{U}} \int_{\mathcal{U}} k(u,v)  u \otimes v \sigma^{-4} \phi(u;0,\sigma^2) \phi(v;0,\sigma^2)  \mathrm{d}u \mathrm{d}v \\
& = 
\sigma^{-4} \int_{\mathcal{U}} \int_{\mathcal{U}} \phi(v;u,l^2) u \otimes v   \phi(u;0,\sigma^2) \phi(v;0,\sigma^2)  \mathrm{d}u \mathrm{d}v \\
& = 
\sigma^{-4} \int_{\mathcal{U}} \int_{\mathcal{U}} \phi\left(v;\frac{u\sigma^2}{(l^2+\sigma^2)},\frac{l^2 \sigma^2}{(l^2+\sigma^2)} \right) u \otimes v   \phi(u;0,l^2+\sigma^2) \phi(u;0,\sigma^2)  \mathrm{d}u \mathrm{d}v \\
& =  
\frac{1}{(\sigma^2+l^2)\sigma^{2}} \int_{\mathcal{U}} u \otimes u  \phi(u;0,l^2+\sigma^2) \phi(u;0,\sigma^2)  \mathrm{d}u \\
& =  
\frac{1}{(\sigma^2+l^2)\sigma^{2}}\phi(0;0,l^2+2\sigma^2) \int_{\mathcal{U}} u\otimes u \phi\left(u;0,\frac{(l^2+\sigma^2)\sigma^2}{l^2+2 \sigma^2}\right) \mathrm{d}u\\
& = 
 (2 \pi)^{-\frac{d}{2}} (l^2+2 \sigma^2)^{-\frac{d}{2}-1}I_{d\times d}.
\end{align*}
Combining the results above we get the advertised result.
\end{proof}

\subsection{Proof of Proposition \ref{prop:critical_scaling_location}}

\begin{proof}
The asymptotic variance satisfies
\begin{align*}
C & = \sigma^2 ((d^{2\alpha}+\sigma^2)(3\sigma^2+d^{2\alpha}))^{-\frac{d}{2}-1}(d^{2\alpha}+2 \sigma^2)^{d+2} = \sigma^2\left(\frac{1 + 4\sigma^2d^{-2\alpha} + 3\sigma^4d^{-4\alpha}}{1 + 4\sigma^2 d^{-2\alpha} + 4\sigma^4d^{-4\alpha}}  \right)^{-d/2-1}.
\end{align*}
Taking logarithms, we obtain:
\begin{align*}
  \log C & = 2\log \sigma - \left(\frac{d}{2} + 1\right)\left( \log(1 + 4\sigma^2d^{-2\alpha} + 3d^{-4\alpha}\sigma^4) - \log(1 + 4\sigma^2 d^{-2\alpha} + 4d^{-4\alpha}\sigma^4)\right).
\end{align*}
By l'Hopital's rule
\begin{align*}
  \lim_{d\rightarrow \infty}\log C & =  \lim_{d\rightarrow \infty} 2\left(1 + \frac{d}{2}\right)^2 \frac{4 \alpha  \sigma ^4}{d \left(6 \sigma ^6 d^{-2 \alpha }+6 \sigma ^2 d^{2 \alpha }+d^{4 \alpha }+11 \sigma ^4\right)},
\end{align*}
which is converges to $\sigma^4$ if $\alpha = 1/4$, converges to $0$ if $\alpha > 1/4$ and converges to infinity if $\alpha < 1/4$.  The critical scaling $C^l$ follows immediately from this.
\end{proof}


\subsection{Gaussian Location Model - Asymptotic Variance for estimator with a Mixture of Gaussian RBF Kernels}
\begin{proof}
A straightforward calculation yields $\overline{M}_T = \sum_{s=1}^S \gamma_s \overline{M}_s =0$. Also, $g(\theta^*)  =  (2\pi)^{-\frac{d}{2}}\sum_{s=1}^S \gamma_s (l_s^2+2 \sigma^2)^{-\frac{d}{2}-1}I_{d\times d}$. Furthermore, we have:
\begin{align*}
\Sigma_T 
& = 
\int_{\mathcal{U}} \left[ \sum_{s=1}^S \gamma_{s} \int_{\mathcal{U}}  \nabla_1 k_s(u,v) \mathbb{U}(\mathrm{d}u) \right] \otimes \left[ \sum_{s'=1}^S \gamma_{s'} \int_{\mathcal{U}}  \nabla_1 k_{s'}(w,v) \mathbb{U}(\mathrm{d}w) \right] \mathbb{U}(\mathrm{d}v)\\
& =  \int_{\mathcal{U}}  \left[ \sum_{s=1}^S \gamma_{s} \frac{ v \phi(v;0,l_s^2 +\sigma^2)}{(l_s^2+\sigma^2)}\right]  \left[ \sum_{s'=1}^S \gamma_{s'} \frac{ v \phi(v;0,l_s^2 +\sigma^2)}{(l_{s'}^2+\sigma^2)}\right] \phi(v;0,\sigma^2) \mathrm{d}v \\
& =  
\sum_{s=1}^S \sum_{s'=1}^S \frac{\gamma_{s} \gamma_{s'}}{(l_s^2+\sigma^2)(l_{s'}^2+\sigma^2)}   \int_{\mathcal{U}} (v \otimes v)  \phi(v;0,l_s^2 +\sigma^2)\phi(v;0,l_s^2 +\sigma^2)\phi(v;0,\sigma^2) \mathrm{d}v \\
& = 
 \sum_{s=1}^S \sum_{s'=1}^S \frac{\gamma_{s} \gamma_{s'} \phi(0;0,2\sigma^2 +l_s^2 +l_{s'}^2) \phi\left(0;0,\frac{(l_s^2 +\sigma^2)(l_{s'}^2+\sigma^2)}{(2\sigma^2 +l_s^2 +l_{s'}^2)}+\sigma^2\right)}{(l_s^2+\sigma^2)(l_{s'}^2+\sigma^2)}\\
& 
\qquad \times  \int_{\mathcal{U}} (v \otimes v)\phi\left(v;0,\frac{\sigma^2(l_s^2+\sigma^2)(l_{s'}^2+\sigma^2)}{(l_s^2+\sigma^2)(l_{s'}^2+\sigma^2)+ \sigma^2(2\sigma^2+l_s^2+l_{s'}^2)}\right) \mathrm{d}v \\
& = 
 \sum_{s=1}^S \sum_{s'=1}^S \frac{\gamma_{s} \gamma_{s'} (2\pi)^{-d} \left((l_s^2 +\sigma^2)(l_{s'}^2+\sigma^2)+ \sigma^2 (2\sigma^2 +l_s^2 +l_{s'}^2) \right)^{-\frac{d}{2}}}{(l_s^2+\sigma^2)(l_{s'}^2+\sigma^2)} \\
& 
 \qquad \times  \left(\frac{\sigma^2(l_s^2+\sigma^2)(l_{s'}^2+\sigma^2)}{(l_s^2+\sigma^2)(l_{s'}^2+\sigma^2)+ \sigma^2(2\sigma^2+l_s^2+l_{s'}^2)}\right)I_{d\times d} \\
& =  
\sum_{s=1}^S \sum_{s'=1}^S \gamma_{s} \gamma_{s'} (2\pi)^{-d} \sigma^2  \left((l_s^2 +\sigma^2)(l_{s'}^2+\sigma^2)+ \sigma^2 (2\sigma^2 +l_s^2 +l_{s'}^2) \right)^{-\frac{d}{2}-1}I_{d\times d}
\end{align*}
Combining the above in the formula $C_T = g(\theta^*)^{-1} \Sigma_T g(\theta^*)^{-1}$ gives the answers.
\end{proof}

\subsection{Gaussian Location Model - Robustness with Mixture of Gaussian RBF Kernels}

\begin{proof}
Following the lines of Proposition \ref{prop:gaussian_robust} we obtain
\begin{align*}
\nabla_{\theta} \text{MMD}^2_T(\mathbb{P}_\theta,\delta_z)
& =  
\nabla_{\theta} \sum_{s=1}^S \gamma_s \left[ \int_{\mathcal{U}} \int_{\mathcal{U}}  k_s(G_{\theta}(u),G_{\theta}(v)) \mathbb{U}(\mathrm{d}u) \mathbb{U}(\mathrm{d}v) - 2  \int_{\mathcal{U}} k_s(G_{\theta}(u),z) \mathbb{U}(\mathrm{d}u)\right] \\
& =  
-2  \sum_{s=1}^S \gamma_s \int_{\mathcal{U}}\nabla_{\theta} k_s(G_{\theta}(u),z) \mathbb{U}(\mathrm{d}u)
= 
-2  \sum_{s=1}^S \gamma_s  \int_{\mathcal{U}} \nabla_{\theta} G_{\theta}(u) \nabla_1 k_s(G_{\theta}(u),z) \mathbb{U}(\mathrm{d}u)\\
& = 
 - 2 \sum_{s=1}^S \gamma_s  \int_{\mathcal{U}}  \frac{(u - (z - \theta))}{l_s^2} \phi(u,z-\theta,l_s^2) \phi(u;0,\sigma^2) \mathrm{d}u\\
& =  
- 2 \sum_{s=1}^S \gamma_s  \int_{\mathcal{U}}  \frac{(u - (z - \theta))}{l^2_s} \phi(z,\theta,l_s^2+\sigma^2) \phi\left(u;\frac{(z-\theta)\sigma^2}{l_s^2+\sigma^2},\frac{l_s^2 \sigma^2}{l_s^2+\sigma^2}\right) \mathrm{d}u\\
& = 
2 \sum_{s=1}^S \gamma_s  \phi\left(z;\theta,l_s^2+\sigma^2\right) \frac{1 }{(l_s^2+\sigma^2)}  (z-\theta)\\
& =  2 (2\pi)^{-\frac{d}{2}} \sum_{s=1}^S \gamma_s  (l_s^2+\sigma^2)^{-\frac{d}{2}-1} \exp\left(-\frac{\|z-\theta\|_{2}^2}{2(l_s^2+\sigma^2)}\right)  (z-\theta) 
\end{align*}
We conclude using the derivation of $g(\theta^*)$ in the previous proof and using the definition of influence function.

\end{proof}

\subsection{Gaussian Scale Model: Asymptotic Variance Calculation for a single Gaussian kernel}

\begin{proof}
Clearly, $\nabla_{\theta}G(u) = e^{\theta}u$. Define $s = e^{2\theta^*}$, then the metric tensor at $\theta^*$ is given by
\begin{align*}
g(\theta^*) &= \int_{\mathcal{U}} \int_{\mathcal{U}} \nabla_{\theta}G_{\theta^*}(u) \cdot \nabla_1\nabla_2 k(G_{\theta^*}(u), G_{\theta^*}(v))\nabla_{\theta}G_{\theta^*}(v) \mathbb{U}(\mathrm{d}u)\mathbb{U}(\mathrm{d}v)\\
&= \int_{\mathcal{U}} \int_{\mathcal{U}} e^{\theta^*}u \cdot \nabla_1\nabla_2 k(e^{\theta^*}u, e^{\theta^*}v)e^{\theta^*}v \mathbb{U}(\mathrm{d}u)\mathbb{U}(\mathrm{d}v)\\
&= \int_{\mathcal{U}} \int_{\mathcal{U}} x \cdot \nabla_1\nabla_2 k(x, y)y \phi(x; 0, s)\phi(y; 0, s)\mathrm{d}x\mathrm{d}y \\
&= \int_{\mathcal{U}} \int_{\mathcal{U}} \nabla\cdot\left(x \phi(x; 0, s)\right) k(x, y)\nabla\cdot\left(y \phi(y; 0, s)\right)\mathrm{d}x\mathrm{d}y \\
&= \int_{\mathcal{U}} \int_{\mathcal{U}} \left(d - \frac{|x|^2}{s}\right) \phi(x; y, l^2)\left(d - \frac{|y|^2}{s}\right)\phi(x; 0, s)\phi(y; 0, s)\mathrm{d}x\mathrm{d}y \\
&= A_1 + A_2 + A_3,
\end{align*}
where 
\begin{align*}
A_1 &= d^2\int_{\mathcal{U}} \int_{\mathcal{U}} \phi(x;y,l^2)\phi(x;0, s)\phi(y;0, s)\mathrm{d}x\mathrm{d}y\\
  &= d^2\int_{\mathcal{U}} \int_{\mathcal{U}} \phi\left(x; \frac{ys}{l^2+s}, \frac{l^2s}{l^2+s}\right) \phi\left(y; 0, l^2+s\right)\phi(y; 0, s)\mathrm{d}x\mathrm{d}y \\
  &= d^2\int_{\mathcal{U}}  \phi\left(y; 0, l^2+s\right)\phi(y; 0, s)\mathrm{d}y\\
  &= d^2\int_{\mathcal{U}}  \phi\left(y; 0, \frac{(l^2+s)s}{l^2+2s}\right)\phi(0; 0, l^2+2s)\mathrm{d}y\\
  &= d^2 {(2\pi)}^{-\frac{d}{2}}{(l^2+2s)}^{-\frac{d}{2}}.
\end{align*}
\begin{align*}
A_2 &= -\frac{2d}{s} \int_{\mathcal{U}} \int_{\mathcal{U}} |x|^2 \phi\left(x; \frac{ys}{l^2+s}, \frac{l^2s}{l^2+s}\right)\phi\left(y; 0, l^2+s\right)\phi\left(y; 0, s\right)\mathrm{d}x\mathrm{d}y\\
&= -\frac{2d}{s}\int_{\mathcal{U}} \left(\frac{dl^2s}{l^2+s} + \frac{|y|^2 s^2}{(l^2+s)^2}\right) \phi\left(y; 0, l^2+s\right)\phi\left(y; 0, s\right)\mathrm{d}y\\
&= -\frac{2d}{l^2+s}\int_{\mathcal{U}} \left({dl^2} + \frac{|y|^2 s}{(l^2+s)}\right) \phi\left(y; 0, \frac{(l^2+s)s}{l^2+2s}\right)\phi\left(0; 0, l^2+2s\right)\mathrm{d}y\\
&= -\frac{2d^2}{l^2+s}(l^2+2s)^{-d/2}{(2\pi)}^{-d/2}\left[l^2+\frac{s^2}{l^2+2s}\right].
\end{align*}
\begin{align*}
A_3 &= \frac{1}{s^2}\int_{\mathcal{U}} \int_{\mathcal{U}} |x|^2|y|^2 \phi\left(x; \frac{ys}{l^2+s}, \frac{l^2s}{l^2+s}\right)\phi\left(y; 0, l^2+s\right)\phi(y; 0, s)\mathrm{d}y\mathrm{d}x\\
&= \frac{1}{s^2}\int_{\mathcal{U}} \left(\frac{dl^2s}{l^2+s}+\frac{|y|^2s^2}{{(l^2+s)}^2}\right)|y|^2\phi\left(y; 0, l^2+s\right)\phi(y; 0, s)\mathrm{d}y\\
&= \frac{1}{s^2}\int_{\mathcal{U}} \left(\frac{dl^2s}{l^2+s}+\frac{|y|^2s^2}{{(l^2+s)}^2}\right)|y|^2\phi\left(y; 0, \frac{(l^2+s)s}{l^2+2s}\right)\phi(0; 0, l^2+2s)\mathrm{d}y\\
&= {(2\pi)}^{-d/2}{(l^2+2s)}^{-d/2}\left[\frac{d^2l^2}{l^2+2s}+(d^2+2d)\frac{s^2}{(l^2+2s)^2}\right].
\end{align*}
It follows that $g(\theta^*) = {(2\pi)}^{-d/2}{(l^2+2s)}^{-d/2}d^2 K(d,l, s)$, where $K(d,l,s)$ is bounded with respect to $d$, $l$ and $s$ and $K(d, 0, s) = (1 + 2d^{-1})/4$. 

Now consider the term
\begin{align*}
& \int_{\mathcal{U}} \nabla_1 k(G_{\theta^*}(u), G_{\theta^*}(v))\nabla_{\theta}G_{\theta^*}(u)\mathbb{U}(\mathrm{d}v) \\
&= \int_{\mathcal{U}} \nabla_1 k(G_{\theta^*}(x), G_{\theta^*}(y))e^{\theta^*}x\phi(x;0, 1)\mathrm{d}x \\
&= \int_{\mathcal{U}} \nabla_1 k(x, G_{\theta^*}(y))x\phi(x;0, s)\mathrm{d}x \\
&= -\int_{\mathcal{U}}  k(x, G_{\theta^*}(y))\nabla\cdot\left(x\phi(x;0, s)\right)\mathrm{d}x \\
&= -\int_{\mathcal{U}}  \phi(x; G_{\theta^*}(y),l^2)\left(d - \frac{|x|^2}{s}\right)\phi(x;0, s)\mathrm{d}x\\ 
&= -\int_{\mathcal{U}} \left(d - \frac{|x|^2}{s}\right)\phi\left(x; \frac{G_{\theta^*}(y)s}{s+l^2}, \frac{sl^2}{s+l^2}\right)\phi(G_{\theta^*}(y);0,s+l^2)\mathrm{d}x\\
&= - \left(\frac{ds}{s+l^2}-\frac{|G_{\theta^*}(y)|^2s}{(s+l^2)^2}\right)\phi(G_{\theta^*}(y);0,s+l^2).
\end{align*}
Then 
\begin{align*}
 \overline{M} & =  \int_{\mathcal{U}} \int_{\mathcal{U}} \nabla_1 k(G_{\theta^*}(x), G_{\theta^*}(y))\nabla_{\theta}G_{\theta^*}(x)\mathbb{U}(\mathrm{d}x)\mathbb{U}(\mathrm{d}y)\\
 &= - \int_{\mathcal{U}} \left(\frac{ds}{s+l^2}-\frac{|G_{\theta^*}(y)|^2s}{(s+l^2)^2}\right)\phi(G_{\theta^*}(y);0,s+l^2)\mathbb{U}(\mathrm{d}y)\\
  &= - \int_{\mathcal{U}} \left(\frac{ds}{s+l^2}-\frac{|y|^2s}{(s+l^2)^2}\right)\phi(y;0,s+l^2)\phi(y;0,s)\mathrm{d}y\\
  &= - \int_{\mathcal{U}} \left(\frac{ds}{s+l^2}-\frac{|y|^2s}{(s+l^2)^2}\right)\phi\left(y; 0, \frac{(s+l^2)s}{l^2+2s}\right)\phi\left(0; 0, l^2+2s\right)\mathrm{d}y\\
  &= -{(2\pi)}^{-d/2}{(l^2+2s)}^{-d/2}\frac{ds}{s+l^2}\left(1-\frac{s}{l^2+2s}\right).
\end{align*}
It follows that
\begin{align*}
& \int_{\mathcal{U}} \left(\int_{\mathcal{U}} \nabla_1 k(G_{\theta^*}(u), G_{\theta^*}(v))\nabla_{\theta}G_{\theta^*}(u)\mathbb{U}(\mathrm{d}u)\right)^2 \mathbb{U}(\mathrm{d}v) \\ &=  \int_{\mathcal{U}} \left(\frac{ds}{s+l^2}-\frac{|G_{\theta^*}(y)|^2s}{(s+l^2)^2}\right)^2 \phi^2(G_{\theta^*}(y);0,s+l^2)\phi(y; 0, 1)\mathrm{d}y\\
&=  \int_{\mathcal{U}} \left(\frac{ds}{s+l^2}-\frac{|y|^2s}{(s+l^2)^2}\right)^2 \phi^2(y;0,s+l^2)\phi(y; 0, s)\mathrm{d}y
\end{align*}
Expanding the terms we have
\begin{align*}
\int &\left(\frac{ds}{s+l^2}-\frac{|y|^2s}{(s+l^2)^2}\right)^2 \phi^2(y;0,s+l^2)\phi(y; 0, s)\mathrm{d}y\\ &= \int \left(\frac{ds}{s+l^2}-\frac{|y|^2s}{(s+l^2)^2}\right)^2 \phi\left(y;0,\frac{s+l^2}{2}\right)\phi(y; 0, s)\mathrm{d}y\phi\left(0;0, 2(l^2+s)\right)\\
&= \int \left(\frac{ds}{s+l^2}-\frac{|y|^2s}{(s+l^2)^2}\right)^2 \phi\left(y;0,\frac{(s+l^2)s}{l^2+3s}\right)\mathrm{d}y\phi\left(0; 0, (l^2+3s)/2\right)\phi\left(0;0, 2(l^2+s)\right)\\
&= \left(\frac{s^2d^2}{(s+l^2)^2}-2\frac{d^2s^3}{(l^2+3s)(s+l^2)^2}+\frac{(d^2+2d)s^2}{(s+l^2)^2}\frac{s^2}{(l^2+3s)^2}\right) \phi\left(0; 0, (l^2+3s)/2\right)\phi\left(0;0, 2(l^2+s)\right)\\
&= \frac{d^2s^2}{(s+l^2)^2}\left(1-2\frac{s}{l^2+3s}+\frac{(1+2d^{-1})s^2}{(l^2+3s)^2}\right) {(2\pi)^{-d}}{(l^2+3s)}^{-d/2}{(l^2+s)}^{-d/2}.
\end{align*}
It follows that
\begin{align*}
\Sigma &= \int_{\mathcal{U}} \left(\int_{\mathcal{U}} \nabla_1 k(G_{\theta^*}(u), G_{\theta^*}(v))\nabla_{\theta}G_{\theta^*}(u)\mathbb{U}(\mathrm{d}u)\right)^2 \mathbb{U}(\mathrm{d}v) - \overline{M}^2\\
&= {(2\pi)}^{-d}\frac{d^2s^2}{(s+l^2)^2}\left[C_1(s,l, d){(l^2+3s)}^{-d/2}{(l^2+s)}^{-d/2} - C_2(s,l,d){(l^2+2s)}^{-d}\right],
\end{align*}
where the terms
\begin{align*}
C_1(s,l, d) & = \left(1-2\frac{s}{l^2+3s}+\frac{(1+2d^{-1})s^2}{(l^2+3s)^2}\right) \quad\mbox{ and }\quad 
C_2(s,l, d) = \left(1-\frac{s}{l^2+2s}\right)^2,
\end{align*}
are bounded uniformly with respect to $s, l, d$. 
The asymptotic variance of the estimator $\hat{\theta}_m$ is then given by $C = g^{-1}(\theta^*)\Sigma g^{-1}(\theta^*)$, as stated. 
\end{proof}


\subsection{Proof of Proposition \ref{prop:critical_scaling_scaling}}

\begin{proof}
The asymptotic variance can be written as 
\begin{align*}
C & = \frac{\left(l^2+2 s\right)^2 \left(\left(l^2+s\right)^{-\frac{d}{2}-2} \left(l^2+2 s\right)^{d+2} \left(l^2+3 s\right)^{-\frac{d}{2}-2} \left( \left(l^2+2 s\right)^2+2 s^2/d\right)-1\right)}{(d+2)^2 s^2}.
\end{align*}
Let $l = d^\alpha$, then it is a straightforward calculation to show that the term
\begin{align*}
 E(d) & := \left(\frac{1+\frac{s}{d^{2\alpha}}}{1+\frac{2s}{d^{2\alpha}}}\right)^{-d/2} \left(\frac{1+\frac{3s}{d^{2\alpha}}}{1+\frac{2s}{d^{2\alpha}}}\right)^{-d/2} = \left(\frac{1 + \frac{4s}{d^{2\alpha}}+\frac{3s^2}{d^{4\alpha}}}{1 + \frac{4s}{d^{2\alpha}} + \frac{4s^2}{d^{4\alpha}}}\right)^{-d/2},
\end{align*}
converges to $1$ if $\alpha > \frac{1}{4}$, $s^2/2$ if $\alpha = 1/4$ and $\infty$ if $\alpha < 1/4$.  Moreover, the convergence in each case is exponentially fast.  We can express the asymptotic variance as $C = (E(d)-1)B(d)+ B(d) -1/A(d)$, where $A(d) = ((d+2)^2 s^2)/(d^{2 \alpha }+2 s)^2$ and 
\begin{align*}
B(d) & =\frac{\left(d^{2 \alpha }+2 s\right)^2 \left(\left(d^{2 \alpha }+2 s\right)^2+\frac{2 s^2}{d}\right)}{\left(d^{2 \alpha }+s\right)^2 \left(d^{2 \alpha }+3 s\right)^2}=\frac{\left(2 s d^{-2 \alpha }+1\right)^2 \left(2 s^2 d^{-4 \alpha -1}+\left(2 s d^{-2 \alpha }+1\right)^2\right)}{\left(s d^{-2 \alpha }+1\right)^2 \left(3 s d^{-2 \alpha }+1\right)^2},
\end{align*}
 By the mean value theorem, $B(d)-1$ is $O(d^{-1-4\alpha})$.  Since $1/A(d)$ is $O(d^{4\alpha-2})$, it follows that $(B(d)-1)/A(d)$ is $O(d^{-1})$.	  It follows that $C$ converges to zero for $\alpha > 1/4$ and to infinity for $\alpha < 1/4$.   When $\alpha = 1/4$ since $B(d)/A(d) = O(d^{-1})$ it follows that $C\rightarrow 0$, completing the proof.
\end{proof}


\section{Additional Details for Numerical Experiments} \label{appendix:experiments}

In this section, we provide additional details and simulation results for experiments in the paper.

\subsection{Gaussian Distributions}\label{appendix:MMDestimators_gaussian_distribution}

In this subsection we extend Figure \ref{fig:MMDestimators_gaussian_model} for the Gaussian location model with different classes of kernels. In Figure \ref{fig:gaussian_loss_landscape_appendix} we plot the loss landscape in each case for different dimensions and different parameter choices.  We note that the inverse multiquadric kernel suffers less from vanishing gradients. In Figure \ref{fig:gaussian_robustness_pollution_appendix} and \ref{fig:gaussian_threshold_robustness_appendix} we compute the error in estimating the parameter of a Gaussian location model, as a function of the location of the Dirac contamination and the percentage of corrupted samples.  As the estimator is qualitatively robust, this influence will be bounded independently of this location, but the maximum error will depend strongly on the choice of kernel and kernel parameters. 

Finally, Figure \ref{fig:Sinkhorn_robustness} provides plots demonstrating the strong lack of robustness of the Sinkhorn algorithm as studied in \cite{Genevay2017}. These results demonstrate that this lack of robustness occurs for a large range of regularisation parameter $\epsilon$. The experiments were performed using an $l_2$ cost, which is standard in this literature. Other cost functions could potentially be used to improve the robustness of this estimator, but this is currently an open question.

\begin{figure}[t!]
\begin{center}
\includegraphics[width=0.235\textwidth,clip,trim = 0 0 0 0]{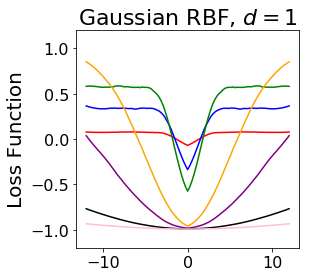}
\includegraphics[width=0.22\textwidth,clip,trim = 0 0 0 0]{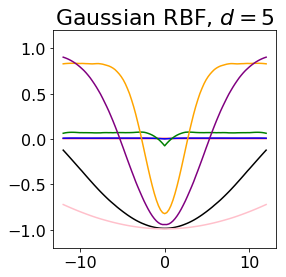}
\includegraphics[width=0.22\textwidth,clip,trim = 0 0 0 0]{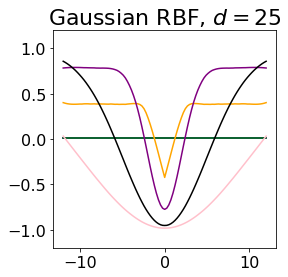}
\includegraphics[width=0.23\textwidth,clip,trim = 0 0 0 0]{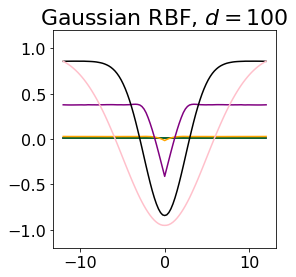}\\
\includegraphics[width=0.235\textwidth,clip,trim = 0 0 0 0]{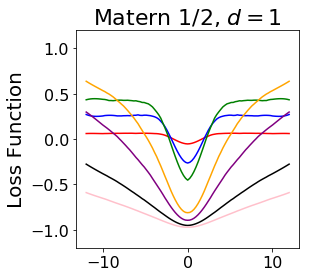}
\includegraphics[width=0.22\textwidth,clip,trim = 0 0 0 0]{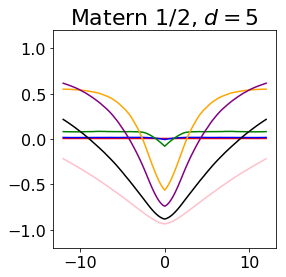}
\includegraphics[width=0.22\textwidth,clip,trim = 0 0 0 0]{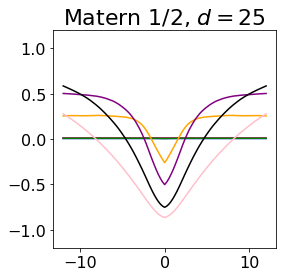}
\includegraphics[width=0.22\textwidth,clip,trim = 0 0 0 0]{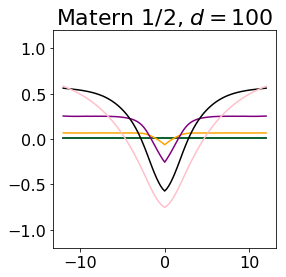}\\
\includegraphics[width=0.235\textwidth,clip,trim = 0 0 0 0]{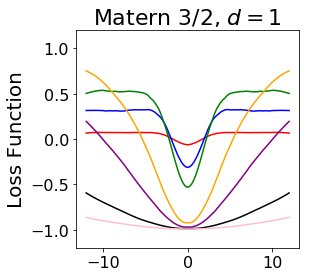}
\includegraphics[width=0.22\textwidth,clip,trim = 0 0 0 0]{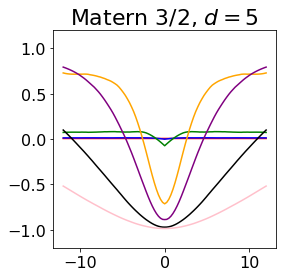}
\includegraphics[width=0.22\textwidth,clip,trim = 0 0 0 0]{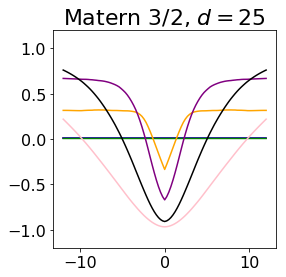}
\includegraphics[width=0.22\textwidth,clip,trim = 0 0 0 0]{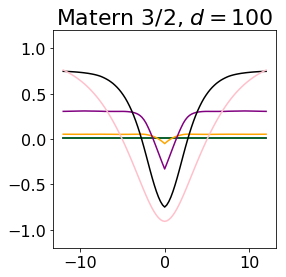}\\
\includegraphics[width=0.235\textwidth,clip,trim = 0 0 0 0]{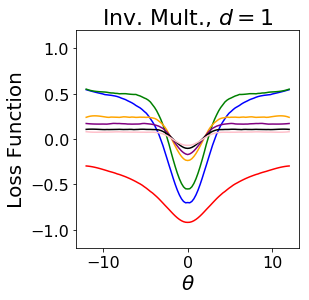}
\includegraphics[width=0.22\textwidth,clip,trim = 0 0 0 0]{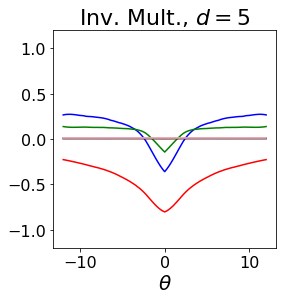}
\includegraphics[width=0.22\textwidth,clip,trim = 0 0 0 0]{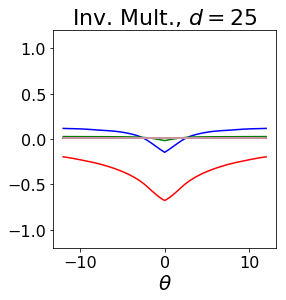}
\includegraphics[width=0.22\textwidth,clip,trim = 0 0 0 0]{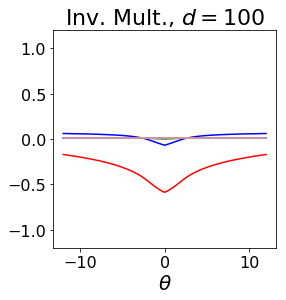}
\caption{MMD Loss landscape for the Gaussian location model in dimensions $d=1,5,25,100$. The landscape is plotted for varying choices of kernels including a Gaussian RBF kernel, a Mat\'ern kernel with smoothness $\frac{1}{2}$ or $\frac{3}{2}$ and an inverse-multiquadric kernel. For each kernel, we plot the loss function for varying values of the lengthscale parameter including $l=0.1$ (red), $l=0.5$ (blue), $l=1$ (green), $l=5$ (orange), $l=10$ (purple), $l=25$ (black) and $l=50$ (pink).\label{fig:gaussian_loss_landscape_appendix}}
\end{center} 
\end{figure}

\begin{figure}[t!]
\begin{center}
\includegraphics[width=0.22\textwidth,clip,trim = 0 0 0 0]{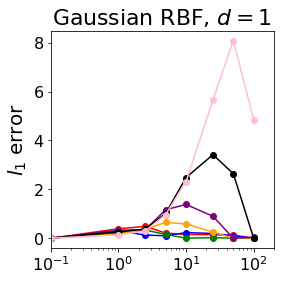}
\includegraphics[width=0.22\textwidth,clip,trim = 0 0 0 0]{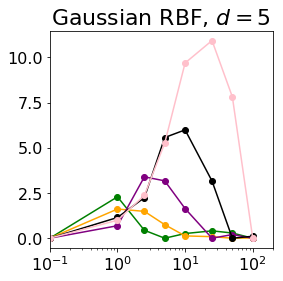}
\includegraphics[width=0.22\textwidth,clip,trim = 0 0 0 0]{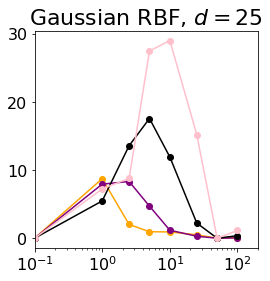}
\includegraphics[width=0.22\textwidth,clip,trim = 0 0 0 0]{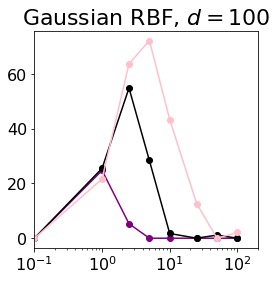}\\
\includegraphics[width=0.22\textwidth,clip,trim = 0 0 0 0]{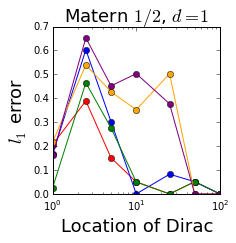}
\includegraphics[width=0.22\textwidth,clip,trim = 0 0 0 0]{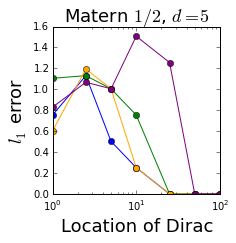}
\includegraphics[width=0.22\textwidth,clip,trim = 0 0 0 0]{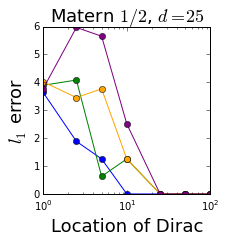}
\includegraphics[width=0.22\textwidth,clip,trim = 0 0 0 0]{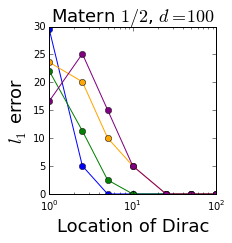}\\
\includegraphics[width=0.22\textwidth,clip,trim = 0 0 0 0]{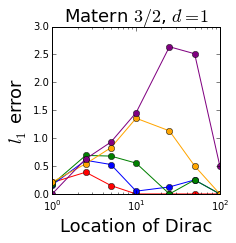}
\includegraphics[width=0.22\textwidth,clip,trim = 0 0 0 0]{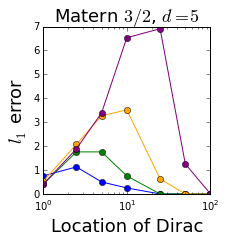}
\includegraphics[width=0.22\textwidth,clip,trim = 0 0 0 0]{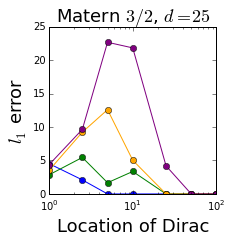}
\includegraphics[width=0.22\textwidth,clip,trim = 0 0 0 0]{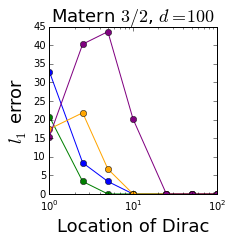}\\
\includegraphics[width=0.22\textwidth,clip,trim = 0 0 0 0]{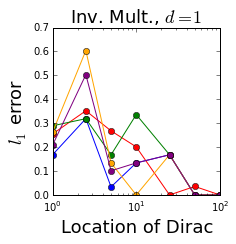}
\includegraphics[width=0.22\textwidth,clip,trim = 0 0 0 0]{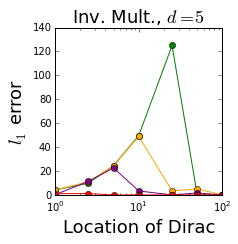}
\includegraphics[width=0.22\textwidth,clip,trim = 0 0 0 0]{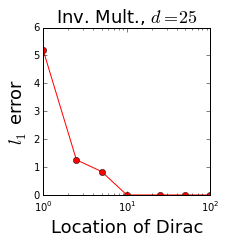}
\includegraphics[width=0.22\textwidth,clip,trim = 0 0 0 0]{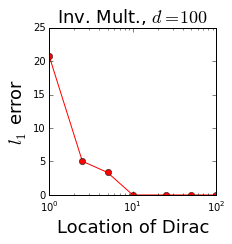}\\
\caption{Gaussian distribution with unknown mean: Robustness as a function of the location of the Dirac for varying kernel and kernel lengthscales in dimensions $d=1,5,25,100$.\label{fig:gaussian_robustness_pollution_appendix}}
\end{center}
\end{figure}

\begin{figure}[t!]
\begin{center}
\includegraphics[width=0.22\textwidth,clip,trim = 0 0 0 0]{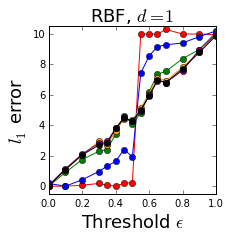}
\includegraphics[width=0.22\textwidth,clip,trim = 0 0 0 0]{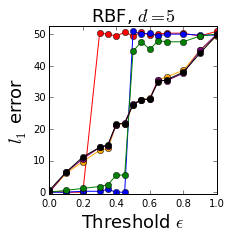}
\includegraphics[width=0.22\textwidth,clip,trim = 0 0 0 0]{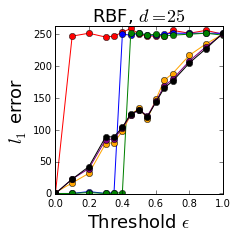}
\includegraphics[width=0.22\textwidth,clip,trim = 0 0 0 0]{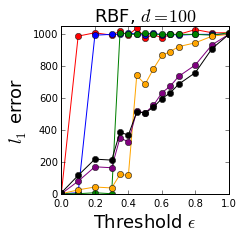}\\
\includegraphics[width=0.22\textwidth,clip,trim = 0 0 0 0]{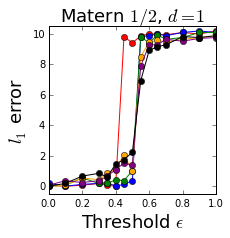}
\includegraphics[width=0.22\textwidth,clip,trim = 0 0 0 0]{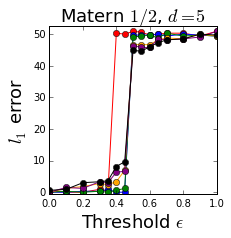}
\includegraphics[width=0.22\textwidth,clip,trim = 0 0 0 0]{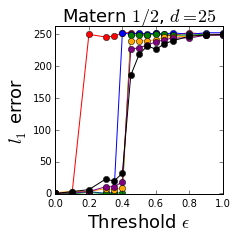}
\includegraphics[width=0.22\textwidth,clip,trim = 0 0 0 0]{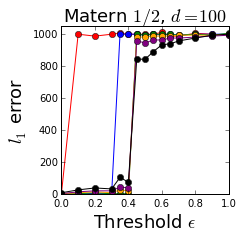}\\
\includegraphics[width=0.22\textwidth,clip,trim = 0 0 0 0]{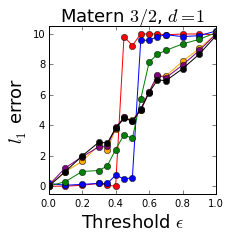}
\includegraphics[width=0.22\textwidth,clip,trim = 0 0 0 0]{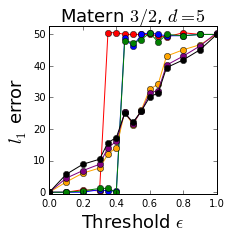}
\includegraphics[width=0.22\textwidth,clip,trim = 0 0 0 0]{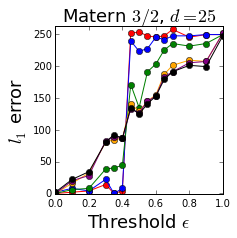}
\includegraphics[width=0.22\textwidth,clip,trim = 0 0 0 0]{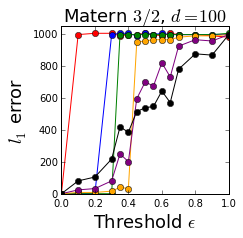}\\
\includegraphics[width=0.22\textwidth,clip,trim = 0 0 0 0]{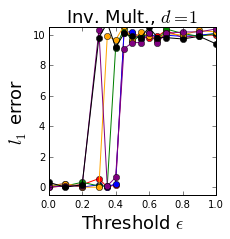}
\includegraphics[width=0.22\textwidth,clip,trim = 0 0 0 0]{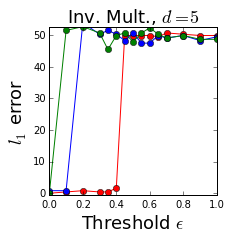}
\includegraphics[width=0.22\textwidth,clip,trim = 0 0 0 0]{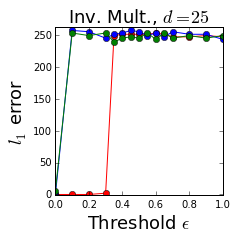}
\includegraphics[width=0.22\textwidth,clip,trim = 0 0 0 0]{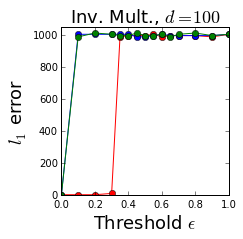}\\
\caption{Gaussian distribution with unknown mean: Error in estimator as a function of the threshold $\epsilon$ for varying kernel and kernel lengthscales in dimensions $d=1,5,25,100$.\label{fig:gaussian_threshold_robustness_appendix}}
\end{center}
\end{figure}

\begin{figure}[t!]
\begin{center}
\includegraphics[width=0.32\textwidth,clip,trim = 0 0 0 0]{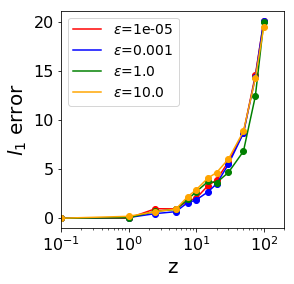}
\includegraphics[width=0.325\textwidth,clip,trim = 0 0 0 0]{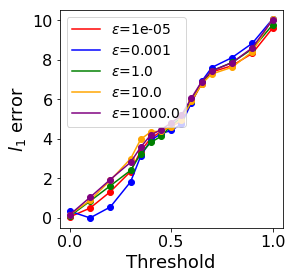}
\end{center}
\caption{\textit{Gaussian location models - Performance of Sinkhorn Estimators in $d=1$ for varying values of the regularisation parameter $\epsilon$.} \textit{Left:} $l_1$ error as a function of the location of the Dirac. \textit{Right:} $l_1$ error as a function of the percentage of corrupted data points.}
\label{fig:Sinkhorn_robustness}
\end{figure}

\subsection{G-and-k Distribution}

In order to implement MMD estimators, we will need to have access to derivatives of the generator, which are given as follows: $\partial G_\theta(u)/ \partial \theta_1 = 1$ and
 \begin{eqnarray*}
\frac{\partial G_\theta(u)}{\partial \theta_2} & = &  \left( 1 + \frac{4}{5} \frac{\big( 1 - \exp(- \theta_3 z(u)\big)}{\big( 1 + \exp(-\theta_3 z(u)\big)}\right)\big(1+z(u)^2 \big)^{\theta_4} z(u) \\
\frac{\partial G_\theta(u)}{\partial \theta_3} & = &  \frac{8}{5}\; \theta_2 \frac{\exp(\theta_3 z(u))}{\big( 1 + \exp(\theta_3 z(u))\big)^2}\big(1+z(u)^2 \big)^{\theta_4} z(u)^2 \\
\frac{ \partial G_\theta(u)}{\partial \theta_4} & = &  \theta_2 \left( 1 + 0.8 \frac{\big( 1 - \exp(- \theta_3 z(u)\big)}{\big( 1 + \exp(-\theta_3 z(u)\big)}\right) \big(1+z(u)^2 \big)^{\theta_4} \log(1+z(u)^2) z(u)
\end{eqnarray*}
Note that these could also be obtained using automatic differentiation.

\subsection{Stochastic Volatility Model} \label{sec:stochvol_appendix}

We can see the stochastic volatility model as a generative model with parameters $\theta = (\theta_1,\theta_2,\theta_3)$, which maps a sample $u = (u_0,u_1,\ldots,u_{2T})$ which is $\mathcal{N}(0,I_{2T \times 2 T})$ distributed to a realisation $y = (y_1,\ldots,y_T)$ of the stochastic volatility model. Here, $\epsilon_t = u_{t}$ for $t\geq 1$, $h_1 = u_{T+1} \sqrt{\sigma^2/(1-\phi^2)}$ and $\eta_t = \sigma u_{T+t}$ for all $t \geq 2$. Note that it is possible to go back to the original parameterisation using $\phi = (\exp(\theta_1)-1)/(\exp(\theta_1)+1)$, $\kappa = \exp(\theta_2)$ and $\sigma = \exp(\theta_3/2)$.

We can obtain the derivative process as follows: $\partial_{\theta_1} y_t = y_t (\partial_{\theta_1} h_t)/2$, $\partial_{\theta_1} h_1 = [(\exp(\theta_1/2)-\exp(-\theta_1/2))/(\exp(\theta_1/2)+\exp(-\theta_1/2))](h_1/2)$, $\partial_{\theta_1} h_t = (\partial_{\theta_1} \phi)  h_{t-1} + \phi (\partial_{\theta_1} h_{t-1})$ for $t>1$,
$\partial_{\theta_1} \phi =  2 \exp(\theta_1)/(\exp(\theta_1)+1)^2$, $\partial_{\theta_2} y_t = y_t$, $\partial_{\theta_2} h_t = 0$, $\partial_{\theta_3} y_t = y_t (\partial_{\theta_3} h_t)/2$, $\partial_{\theta_3} h_1 = h_1/2$, $\partial_{\theta_3} h_t = \phi (\partial_{\theta_3} h_{t-1}) + (\partial_{\theta_3} \eta_t)$ for $t >1$ and $\partial_{\theta_3} \eta_t = \eta_t/2$.

\begin{figure}[t!]
\begin{center}
\includegraphics[width=0.32\textwidth,clip,trim = 0 0 0 0]{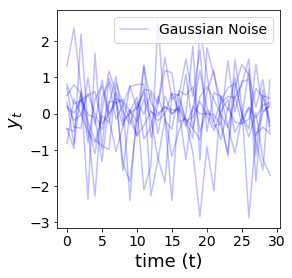}
\includegraphics[width=0.34\textwidth,clip,trim = 0 0 0 0]{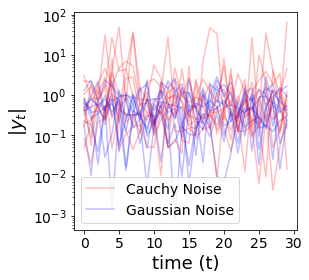}
\caption{\textit{Realisations from the stochastic volatility model.} Left: $10$ realisations from the assumed model $\mathbb{P}_{\theta^*}$ (i.e. stochastic volatility model with Gaussian noise). Right: Absolute value of these same realisations and $10$ realisations from the data generating process $\mathbb{Q}$ (i.e. stochastic volatility model with Cauchy noise).}
\end{center}
\end{figure}


\subsection{Stochastic Lotka-Volterra Model} 

Besides simulating $X_{1,t}$ and $X_{2,t}$ we also required the coupled matrix diffusion process $J_t$ taking values in $\mathbb{R}^{2 \times 2}$ which satisfies the following SDE:
\begin{align*}
dJ_t & = J_t A(X_{1,t}, X_{2,t})\,dt + \sum_{i=1}^3 J_t B_i(X_{1,t}, X_{2,t})dW_{i,t}, 
\end{align*}
where
\begin{align*}
A(x,y) &= \left(\begin{array}{cc} c_1 - c_2 y & c_2 x \\ c_2 y & c_2x-c_3 \end{array}\right)
, \qquad B_1(x,y) = \frac{\sqrt{c_1}}{2}\left(\begin{array}{cc} \frac{1}{\sqrt{x}} & 0 \\ 0 & 0 \end{array}\right),\\
B_2(x,y) & = \frac{\sqrt{c_2}}{2}\left(\begin{array}{cc} -\sqrt{\frac{y}{x}} & -\sqrt{\frac{x}{y}} \\ \sqrt{\frac{y}{x}} & \sqrt{\frac{x}{y}} \end{array}\right), \qquad B_3(x,y) = \frac{\sqrt{c_3}}{2}\left(\begin{array}{cc} 0 & 0 \\ 0 & -\frac{1}{\sqrt{x}} \end{array}\right),
\end{align*}
and subject to the initial condition $J_0 = I_{2\times 2}$.

\end{document}